\documentclass[reprint,
 amsmath,amssymb,
 aps,
 nofootinbib,
 superscriptaddress,
]{revtex4-2}

\usepackage{graphicx}
\usepackage{dcolumn}
\usepackage{bm}
\usepackage{mathtools}
\usepackage{amsthm}
\usepackage[table]{xcolor}
\usepackage{physics}
\usepackage{subcaption}
\usepackage{placeins}
\usepackage{tikz}
\usetikzlibrary{quantikz}
\usepackage[
    bookmarks=true,
    pdftoolbar=true,
    pdfmenubar=true,
    linktocpage=true,
    pdffitwindow=false,
    pdfstartview={FitH},
    pdfnewwindow=true,
    colorlinks=true,
    linkcolor=blue,
    citecolor=blue,
    filecolor=blue,
    urlcolor=blue
]{hyperref}
\usepackage[normalem]{ulem}


\newtheorem{thm}{Theorem}[section]
\newcommand{\CNOT}{\operatorname{CNOT}}
\newcommand{\SWAP}{\mathsf{SWAP}}

\begin{document}
    \begin{flushright}
        KUNS-3119
    \end{flushright}

    \title{Efficient Fidelity Estimation of General Quantum Resource States via Clifford Circuits}

    \author{Shunta Takahashi}
    \affiliation{Perimeter Institute for Theoretical Physics, Waterloo, Ontario N2L 2Y5, Canada}
    \affiliation{Department of Physics and Astronomy, University of Waterloo, Waterloo, Ontario N2L 3G1, Canada}
    \affiliation{Department of Physics, Kyoto University, Kyoto 606-8502, Japan}

    \author{Keisuke Fujii}
    \affiliation{Graduate School of Engineering Science, Osaka University, 1-3 Machikaneyama, Toyonaka, Osaka 560-8531, Japan}
    \affiliation{Graduate School of Informatics, Kyoto University, Sakyo-ku, Kyoto, 606-8501, Japan}
    \affiliation{Fujitsu Quantum Computing Joint Research Division, Center for Quantum Information and Quantum Biology, Osaka University, 1-2 Machikaneyama, Toyonaka, Osaka, 565-8531, Japan}
    \affiliation{RIKEN Center for Quantum Computing (RQC), Hirosawa 2-1, Wako, Saitama 351-0198, Japan}

    \begin{abstract}
        Fidelity estimation is a fundamental tool for assessing quantum resource states across quantum information science.
        For small infidelity $\epsilon$, tomography and direct fidelity estimation (DFE) have worst-case sample complexity $O(1/\epsilon^2)$ at fixed relative accuracy and confidence.
        Direct target projections achieve $O(1/\epsilon)$ but generally require non-Clifford operations whose reliability depends on the resource states that are typically the target of benchmarking.
        Here, we introduce the Bell--coherence fidelity estimation (BCFE) with optimal sample complexity $\Theta(1/\epsilon)$ using only Clifford gates and Pauli measurements.
        We prove that BCFE is applicable to arbitrary pure targets, possibly multi-qubit, with the statistical error of the single-copy coherence part explicitly controlled.
        To evaluate BCFE in a practical setting, we apply BCFE to analog rotation resource states in the candidate early-FTQC architecture called the Space-Time Efficient Analog Rotations (STAR) architecture.
        In the simulations using surface codes, BCFE is estimated to require up to approximately $27$ times fewer resource states than DFE for the same standard error in the infidelity estimate, while its estimates remain close to reference values despite circuit-level noise in the additional Bell measurement operations.
        These results support efficient resource benchmarking within fault-tolerant quantum computing.
    \end{abstract}

    \maketitle

    \section{Introduction}

        Accurate fidelity estimation is essential for benchmarking quantum resource states everywhere in quantum information science, including quantum computation \cite{Eisert:2020qcb, Bluvstein:2024lqp}, communication \cite{Bennett:1995ra, Horodecki:1998zh, Pompili:2021xyt}, cryptography \cite{Gisin:2002qcr, Nadlinger:2022diqkd}, chemistry \cite{AspuruGuzik:2005molecular, McArdle:2020chemistry}, and sensing \cite{Degen:2017qsr, Deng:2024qem}.
        Here we seek to resolve a small infidelity $\epsilon$ relative to a pure target state within a fixed relative error and at a fixed confidence, rather than reconstruct the full density matrix.
        Standard tomography and even direct fidelity estimation (DFE) \cite{Flammia:2011fda, daSilva:2011characterization} have a worst-case sample complexity of $O(1/\epsilon^2)$ in this regime, even though DFE avoids full state reconstruction.

        The ideal two-outcome POVM $\{\ket{\psi}\!\bra{\psi},I-\ket{\psi}\!\bra{\psi}\}$ attains $O(1/\epsilon)$ scaling, but for a generic non-Clifford target state it requires operations beyond Clifford gates and Pauli measurements.
        We restrict benchmarking to these primitives, motivated by stabilizer-code architectures for fault-tolerant quantum computing (FTQC) \cite{Eastin:2009tem, Bravyi:2012rnv, Pastawski:2014fca}.
        In such architectures, non-Clifford gates are typically supplied through resource states and gate synthesis \cite{Bravyi:2004isx, Bravyi:2012lxn, Gidney:2019ntb, Litinski:2019uvg, Rodriguez:2024bhh, Gidney:2024alh, Itogawa:2024hjp, Dawson:2005blj, Ross:2014qpa, Campbell:2017mix}, so the resource states must themselves be benchmarked without assuming trusted non-Clifford operations.

        Within this setting, Ref.\,\cite{Lee:2025lcs} establishes that two-copy measurements implemented with Clifford operations can achieve optimal $O(1/\epsilon)$ sample complexity for fidelity estimation of specific magic states, including the $T$ and $CCZ$ states.
        For general resource states, we seek a protocol with explicit guarantees on the estimation error and failure probability with the optimal sample complexity $O(1/\epsilon)$.

        A potential application includes the Space-Time Efficient Analog Rotations (STAR) architecture \cite{Akahoshi:2023xck, Toshio:2024hut, Toshio:2026dpa}, an alternative to the high-overhead magic-$T$ state distillation in the conventional fault-tolerant quantum computing architecture based on a quantum error correcting (QEC) code.
        In the STAR architecture, an arbitrary angle rotation is performed directly by gate teleportation with the analog rotation resource states $\ket{m_{\theta}}:=e^{iZ\theta}\ket{+}$ instead of long Clifford+$T$ gate sequences.
        The fidelity of these resource states must therefore be assessed across different rotation angles.
        For generic $\theta$, the protocol of Ref.\,\cite{Lee:2025lcs} is inapplicable, and we need a general protocol beyond it.
        Since the STAR architecture is based on a QEC code, it is also necessary to determine whether noise in the additional operations for logical resource states compromises the accuracy of fidelity estimation.

        In this work, we introduce the \textit{Bell-coherence fidelity estimation} (BCFE) protocol, which applies to fidelity estimation of arbitrary qubit resource states, possibly multi-qubit.
        As its name suggests, the protocol consists of a Bell part and a coherence part (see Fig.\,\ref{fig:schematics}).
        The Bell part extracts partial information about the infidelity from the probability that two independently prepared copies of the state being benchmarked are projected onto the antisymmetric Bell state $\ket{\Psi_-}$, without requiring any state-dependent operations.
        The coherence part completes the picture by compensating the target-dependent information that the Bell part leaves undetermined.
        Together, we rigorously prove that these two parts combine to achieve the optimal sample complexity $O(1/\epsilon)$ at fixed relative accuracy and confidence, matching the optimal lower bound \cite{Lee:2025lcs}.

        To assess its practical utility, we apply BCFE to the analog rotation resource states $\ket{m_{\theta}}\!:=e^{iZ\theta}\!\ket{+}$ in the STAR architecture.
        Specifically, we simulate $d=3$ rotated surface codes under circuit-level noise, varying the physical CNOT error rate jointly in the stabilizer measurements and the logical Bell measurement over $0.001$, $0.003$, and $0.005$.
        In these simulations, BCFE requires up to approximately $27$ times fewer resource states than DFE to achieve the same standard error in the infidelity estimate.
        This estimate is based on standard errors approximately $1.3$--$5.2$ times smaller than those of DFE at comparable expected resource consumption.
        The simulations also show that BCFE estimates remain close to the reference infidelities despite circuit-level noise in the additional operations required for the Bell measurement.
        This is because the logical Bell measurement is implemented fault-tolerantly, with a transversal CNOT and decoded readout that correct the effects of single physical gate errors.

        Beyond this application, BCFE applies to an arbitrary qubit state in principle and admits a multi-qubit extension, so its scope is not limited to analog rotation resource states in the STAR architecture.
        Potential targets include other non-Clifford ancillas consumed through gate teleportation, entangled states distributed over quantum networks, cluster states for measurement-based computation, and entangled probe states for quantum sensing \cite{Bravyi:2004isx, Pompili:2021xyt, Ringbauer:2025mbqc, Finkelstein:2024tweezer}.
        BCFE may therefore serve as a benchmark whenever resolving a small state infidelity is more relevant than reconstructing the full density matrix.

        \begin{figure}[h]
            \centering
            \includegraphics[width=0.9\linewidth]{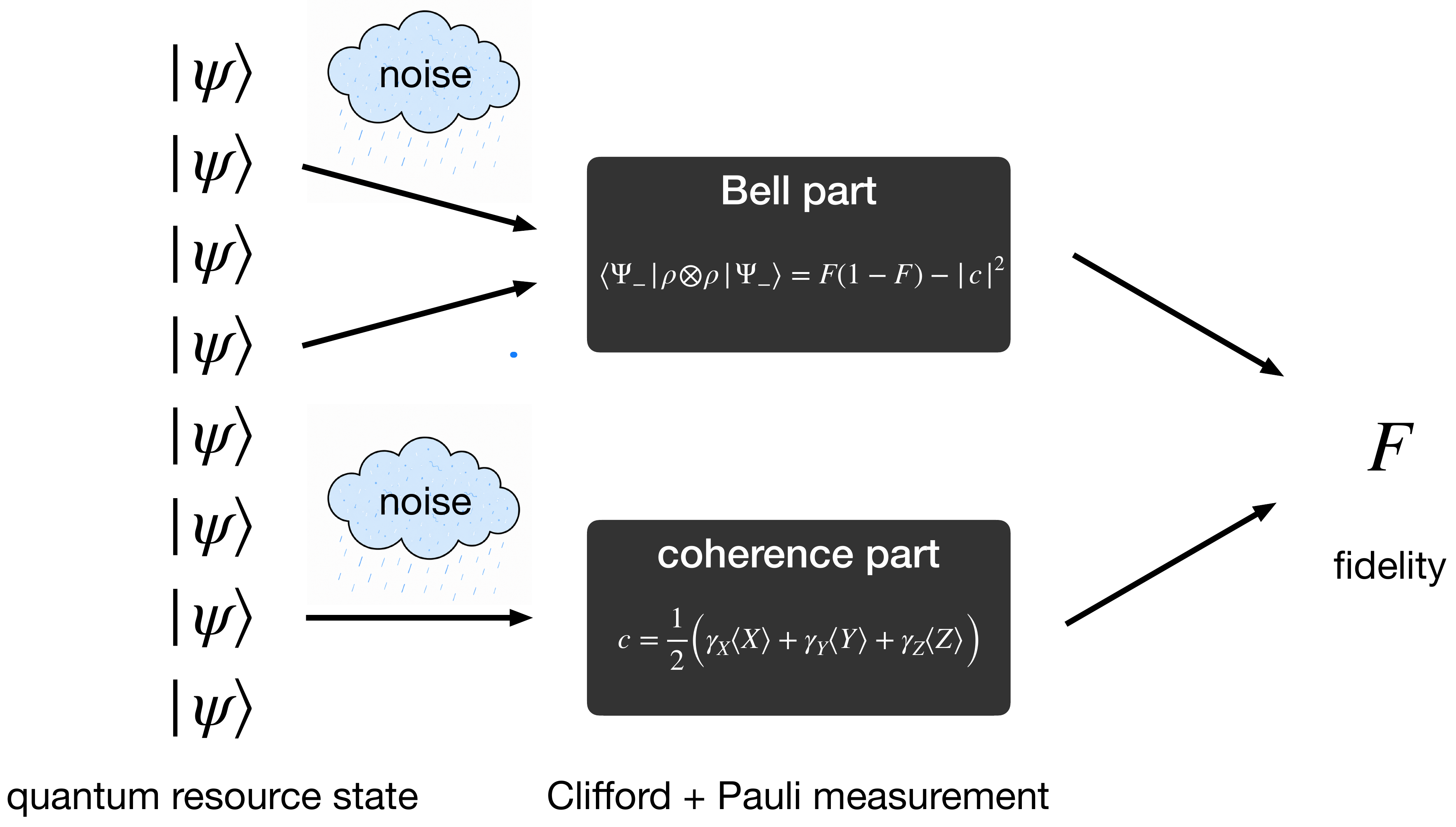}
            \caption{Overview of the Bell-coherence fidelity estimation (BCFE) protocol.
            Copies of the target resource state are subjected to noise and divided between two measurements.
            The Bell part estimates the probability of a rare event by Bell measurements, while the coherence part estimates $c$ from Pauli measurements.
            The two estimates are then combined to recover the fidelity $F$ using only Clifford operations and Pauli measurements with optimal sample complexity $O(1/\epsilon)
            $.}
            \label{fig:schematics}
        \end{figure}

        The rest of the paper is organized as follows.
        In Sec.\,\ref{sec:general_protocol}, we formulate the Bell-coherence fidelity estimation (BCFE) protocol for an arbitrary single qubit resource state and establish its $O(1/\epsilon)$ sample complexity.
        In Sec.\,\ref{sec:theory-and-simulation}, we apply our protocol to the analog rotation resource state used in the STAR architecture and analyze its implementation both on bare physical qubits and on logical qubits encoded in the surface code.
        Sec.\,\ref{sec:conclusion} summarizes the results and discusses open problems and future directions.
        Appendix \ref{append:technical-inequalities} states the concentration and testing inequalities used in the proofs across the paper.
        Appendix \ref{append:multi_qubit_generalization} extends the protocol and its sample complexity analysis to multi-qubit target states.
        In Appendix \ref{append:star}, we review the STAR architecture and explain the motivation for direct analog rotations in more detail.
        Appendix \ref{append:flammia_liu} reviews the DFE estimator and specifies the simulation setup used in the physical and logical level simulations in Sec.\,\ref{sec:theory-and-simulation}.

    \section{General theory of Bell-coherence fidelity estimation (BCFE)} \label{sec:general_protocol}
        In this section, we present a general {protocol} to estimate the fidelity of any given single qubit state with optimal sample complexity $O(1/\epsilon)$ using only Clifford operations and Pauli measurements.

        Ref.\,\cite{Lee:2025lcs} proved that protocols restricted to adaptive single copy Clifford measurements require $\Omega(1/(r^2\epsilon^2))$ copies for generic target states.
        % Although the ideal two-outcome POVM $\{\ket*{\psi}\!\bra*{\psi},I-\ket*{\psi}\!\bra*{\psi}\}$ estimates $\epsilon$ to relative error $r$ using $O(1/(r^2\epsilon))$ copies, realizing it for a generic non-Clifford target would require operations beyond Clifford gates and Pauli measurements and is therefore unrealistic in the FTQC setting considered here.
        BCFE instead uses two copy measurements to attain the same scaling with Clifford gates and Pauli measurements alone.
        To keep the central mechanism in view, we treat the single qubit case here and defer the multi-qubit extension to Appendix \ref{append:multi_qubit_generalization}.

        \subsection{An overview of the BCFE protocol}
            In this subsection, we present a qualitative sketch of our BCFE protocol to gain intuition behind its structure and the improved $O(1/\epsilon)$ sample complexity, leaving aside the mathematically laborious details for a moment.
            The heuristic arguments below are later justified with a rigorous proof in Subsec.\,\ref{subsec:sample complexity}.

            To start, suppose the quantum state to benchmark is
            \begin{align}
                \ket{\psi}=\alpha\ket{0}+\beta\ket{1},\quad |\alpha|^2+|\beta|^2=1,
            \end{align}
            and write the noisy state in the basis $\{|\psi\rangle,\,|\psi^{\perp}\rangle\}$ as
            \begin{align}
                \rho
                =
                \begin{pmatrix}
                    \bra{\psi}\rho\ket{\psi}
                    & \bra{\psi}\rho\ket*{\psi^\perp} \\
                    \bra*{\psi^\perp}\rho\ket{\psi}
                    & \bra*{\psi^\perp}\rho\ket*{\psi^\perp}
                \end{pmatrix}
                =
                \begin{pmatrix}
                    F & c \\
                    c^* & \epsilon
                \end{pmatrix}
                ,
            \end{align}
            where $\ket*{\psi^\perp}:=\beta^*\ket{0}-\alpha^*\ket{1}$ is orthogonal to $\ket{\psi}$.
            Here $F:=\bra{\psi}\rho\ket{\psi}$ is the fidelity and $\epsilon:=1-F$ is the infidelity to be estimated.

            We first remark that the ideal two-outcome POVM $\{\ket*{\psi}\!\bra*{\psi},I\!-\!\ket*{\psi}\!\bra*{\psi}\}$ could estimate the infidelity $\epsilon$ of the prepared noisy state with sample complexity $O(1/\epsilon)$, but it begs the question in the present setting.
            Indeed, for a generic non-Clifford state $\ket{\psi}\!=\!U^{\dagger}\!\ket{0}$, one must apply the non-Clifford unitary $U$ to perform the measurement in the computational basis.
            This is circular since the non-Clifford unitary is typically supplied through a noisy magic state, which is itself the target of fidelity estimation.
            We therefore seek to estimate $\epsilon$ using only fault-tolerant operations available in a realistic architecture, namely, Clifford gates and Pauli measurements.

            To this end, let us note that the antisymmetric singlet Bell state is independent of the choice of basis
            \begin{align}
                \ket*{\Psi_-}
                =\frac{\ket{01}-\ket{10}}{\sqrt{2}}
                =\frac{\ket*{\psi^\perp}\ket{\psi}
                -\ket{\psi}\ket*{\psi^\perp}}{\sqrt{2}}
                .
            \end{align}
            The basis independence fixes the probability of the singlet outcome in a fixed Bell measurement on two independently prepared copies as
            \begin{align} \label{eq:general-determinant}
                b
                :=\bra*{\Psi_-}\rho^{\otimes2}\ket*{\Psi_-}
                =\det\rho
                =F\epsilon-|c|^2
                \geqq0.
            \end{align}
            The last inequality follows from the positive semidefiniteness of the density matrix $\rho\geqq0$.

            The quantity $b$ already captures the incoherent error in $\rho$ since \cite{Hangleiter:2024bell, Lee:2025lcs}
            \begin{align}
                \bra*{\Psi_-}\rho_1\!\otimes\rho_2\ket*{\Psi_-}=0,
            \end{align}
            for the ideal state $\rho_1=\rho_2=\ket{\psi}\!\bra{\psi}$ but can be non-zero, for instance, when each copy undergoes a dephasing channel $\rho_1\!\otimes\rho_2=\mathcal{E}(\ket{\psi}\!\bra{\psi})\otimes\mathcal{E}(\ket{\psi}\!\bra{\psi})$ ($\mathcal{E}(\sigma):=(1-p)\sigma+pZ\sigma Z$ for some $0\!<\!p\!<\!1$).

            The Bell estimator $\widehat{b}$ attains the sample complexity $O(1/\epsilon)$ since $b\leqq\epsilon$.
            To see this, observe that the outcome of each shot follows the Bernoulli distribution $\operatorname{Bern}(b)$, taking the value $1$ with probability $b$ and $0$ otherwise.
            The sample mean $\widehat{b}$ over $N$ shots then has standard deviation $\sqrt{b(1-b)/N}=O(\sqrt{\epsilon/N})$, so resolving it to additive precision $r\epsilon$ requires $N=O(1/(r^2\epsilon))$ samples \cite{Lee:2025lcs}.
            If, instead, the distribution is at $\operatorname{Bern}(O(1))$ as is the case for standard tomography, and thus the standard deviation is given by $O(\sqrt{1/N})$, matching the standard deviation to $r\epsilon$ costs $N=O(1/(r^2\epsilon^2))$.
            The sufficiently small Bernoulli variance $b(1-b)\leqq\epsilon$ is what enables the advantage $N=O(1/(r^2\epsilon))$ over $N=O(1/(r^2\epsilon^2))$ in our protocol.
            The sample complexity theorem (Theorem \ref{thm:sample-complexity}) in the next subsection turns this intuition into a precise statement with failure probability $\delta$.
            The quantum circuit to measure $b$ is shown in Fig.\,\ref{fig:general-bell}.
            \begin{figure*}
                \centering
                \begin{quantikz}
                    \lstick{$\ket{0}$} & \gate{U} & \ctrl{1} & \gate{H} & \meter{} \\
                    \lstick{$\ket{0}$} & \gate{U} & \targ{}  & \qw & \meter{$Z=-1$} \\
                \end{quantikz}
                \captionof{figure}{The first unitary $U$ prepares the resource states $\ket{\psi}=U\ket{0}$.
                The subsequent $\CNOT$ and Hadamard gate implement the Bell measurement.}
                \label{fig:general-bell}
            \end{figure*}
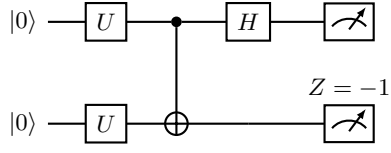

            However, the Bell part outcome \eqref{eq:general-determinant} can also vanish, for instance, for a coherent misrotation, in which both copies are prepared in the same erroneous pure state, e.g.\ $\rho_1\otimes\rho_2=X\!\ket{\psi}\!\bra{\psi}\!X\otimes X\!\ket{\psi}\!\bra{\psi}\!X$.
            Indeed, this example lies in the symmetric subspace and therefore has zero overlap with $\ket*{\Psi_-}$.
            Even though $b$ is zero, the infidelity $\epsilon=\langle\psi^{\perp}|\rho|\psi^{\perp}\rangle=|\langle\psi^{\perp}|X|\psi\rangle|^2$ itself is non-vanishing in general.
            Thus $b$ alone is insufficient to detect the error in $\rho$ in full generality, and we must also estimate the coherence term $|c|^2$ in Eq.\,\eqref{eq:general-determinant} to recover the fidelity $F$.

            The coherence term $c$ is merely a linear combination of the ordinary Pauli expectation values,
            \begin{align} \label{eq:general-coherence}
                c
                &=\bra{\psi}\rho \ket*{\psi^\perp} \\
                &=\frac{1}{2}\Big(
                    (\beta^{*2}-\alpha^{*2})\langle X\rangle
                    +i(\beta^{*2}+\alpha^{*2})\langle Y\rangle
                    +2\alpha^*\beta^*\langle Z\rangle
                \Big). \nonumber
            \end{align}
            At first sight, estimating $c$ from ordinary Pauli measurements might seem to restore the $O(1/\epsilon^2)$ cost because the variances of the Pauli measurement results need not decrease with $\epsilon$, as opposed to the Bell part, whose variance is $b(1-b)=O(\epsilon)$.
            However, the positive semidefiniteness of the density matrix prevents this loss.
            Indeed, the last inequality in Eq.\,\eqref{eq:general-determinant} forces $|c|=O(\sqrt{\epsilon})$, so additive precision $O(r\sqrt{\epsilon})$ in the Pauli expectation values suffices to estimate $|c|^2$ to accuracy $O(r\epsilon)$ with $O(1/(r^2\epsilon))$ sample complexity.
            Again, we make this point precise including the failure probability $\delta$ in the sample complexity theorem (Theorem \ref{thm:sample-complexity}) in the next subsection.

            Finally, the two estimators $\widehat{b}$ and $\widehat{c}$ combine into one for the infidelity itself.
            Equation \eqref{eq:general-determinant} rearranges to
            \begin{align}
                F\epsilon=b+|c|^2,
            \end{align}
            so it determines $F\epsilon$ to within $r\epsilon$.
            Since $F\epsilon=\epsilon(1-\epsilon)$, this quantity is the infidelity up to the correction $\epsilon^2$, which is very small when $\epsilon\ll 1$.
            We can thus preserve the $r\epsilon$ accuracy, leading to the desired precision.

            In the next subsection, we will elaborate on these ideas in a {mathematically} more concrete manner.

        \subsection{Estimation protocol and sample complexity} \label{subsec:sample complexity}
            We now turn to a rigorous formulation and analysis of the BCFE protocol.
            We reiterate that the protocol applies to an arbitrary single qubit state, with a multi-qubit extension given in Appendix \ref{append:multi_qubit_generalization}.

            We first discuss the Bell and coherence parts separately and then combine their guarantees into a single sample complexity theorem (Theorem \ref{thm:sample-complexity}).

            \noindent\paragraph*{Bell estimator and sample complexity.}
            Let $\widehat b$ be the estimator of Eq.\,\eqref{eq:general-determinant}.
            Since $0\leqq b\leqq F(1-F)\leqq\epsilon$, the outcome of the fixed Bell measurement in Fig.\,\ref{fig:general-bell} obeys the Bernoulli distribution with variance at most $\epsilon$.
            We record the result of the $j$th Bell measurement ($1\leqq j\leqq N_B$) as a random variable $B_j$, which equals $1$ if $\ket{\Psi_-}$ is detected and $0$ otherwise, so that $\mathbb{E}[B_j]=b$ and $\widehat b=N_B^{-1}\sum_j B_j$.
            Since $|B_j-b|\leqq1$ and $\operatorname{Var}(B_j)=b(1-b)\leqq\epsilon$, we apply Bernstein's inequality (Theorem \ref{thm:bernstein-inequality}) to $B_1-b,\ldots,B_{N_B}-b$ and obtain
            \begin{align} \label{eq:berenstein-in-bell}
                \Pr\!\left[
                    |\widehat{b}-b|\geqq t
                \right]
                \leqq
                2\exp\!\left(
                    -\frac{N_Bt^2}{2\epsilon+2t/3}
                \right).
            \end{align}
            For $t=r\epsilon$, the bound specializes to
            \begin{align}
                \begin{aligned}
                    \Pr\!\left[
                        |\widehat{b}-b|
                        \geqq
                        r\epsilon
                    \right]
                    &\leqq
                    2\exp\!\left(
                        -\frac{N_Br^2\epsilon}{2+2r/3}
                    \right)
                    \\
                    &\leqq
                    2\exp\!\left(
                        -\frac{3N_Br^2\epsilon}{8}
                    \right),
                \end{aligned}
            \end{align}
            where the second inequality uses $0<r\leqq1$.
            By taking
            \begin{align}
                N_B
                \geqq
                \frac{8}{3r^2\epsilon}
                \log\frac{2}{\delta},
            \end{align}
            the failure probability becomes at most $\delta$, and hence $|\widehat b-b|\leqq r\epsilon$ with probability at least $1-\delta$.

            \noindent\paragraph*{Coherence estimator and sample complexity.}
            Next, we estimate the coherence term by the estimator
            \begin{align}
                \widehat{c}
                :=\frac{1}{2}\Big(
                    (\beta^{*2}-\alpha^{*2})\widehat{\langle X\rangle}
                    +i(\beta^{*2}+\alpha^{*2})\widehat{\langle Y\rangle}
                    +2\alpha^*\beta^*\widehat{\langle Z\rangle}
                \Big),
            \end{align}
            with $N_P$ shots allocated to each of the Pauli measurements.
            Let us rewrite the coherence estimator as
            \begin{align}
                \widehat{c}
                =
                \frac{1}{2}
                \sum_{P=X,Y,Z}
                \gamma_P\widehat{\langle P\rangle},
                \qquad
                \sum_{P=X,Y,Z}
                |\gamma_P|^2
                =
                2,
                \label{eq:general-coherence-vector}
            \end{align}
            where $\gamma_X=\beta^{*2}-\alpha^{*2}$, $\gamma_Y=i(\beta^{*2}+\alpha^{*2})$, and $\gamma_Z=2\alpha^*\beta^*$.
            For each $P\in\{X,Y,Z\}$, we record $N_P\geqq2$ independent measurement outcomes as random variables $M_{P,j}\in\{-1,+1\}$ with $\mathbb{E}[M_{P,j}]=\langle P\rangle$, and estimate $\langle P\rangle$ by their sample mean
            \begin{align}
                \widehat{\langle P\rangle}
                =
                \frac{1}{N_P}\sum_{j=1}^{N_P}M_{P,j}.
            \end{align}
            Since the diagonal terms make $|\widehat{c}|^2$ biased at finite $N_P$, we then use the corrected estimator
            \begin{align} \label{eq:coherence-estimator-square}
                \widehat{q}
                :=
                |\widehat{c}|^2
                -
                \frac{1}{4(N_P-1)}
                \sum_{P=X,Y,Z}
                |\gamma_P|^2
                \Big(
                    1-\widehat{\langle P\rangle}^{\,2}
                \Big).
            \end{align}
            By the relation $\mathbb{E}[(N_P\widehat{\langle P\rangle}^{\,2}-1)/(N_P-1)]=\langle P\rangle^2$, together with the independence of measurements in different Pauli bases, we can show that $\mathbb{E}[\widehat{q}]=|c|^2$.

            The error in the coherence estimate is $e:=\widehat{c}-c=\frac{1}{2}\sum_{P=X,Y,Z}\gamma_P\Delta_P$, where $\Delta_P:=\widehat{\langle P\rangle}-\langle P\rangle$.
            By the Cauchy--Schwarz inequality, this error is bounded by
            \begin{align}
                \begin{aligned}
                    |e|
                    &\leqq
                    \frac{1}{2}
                    \left(
                        \sum_P|\gamma_P|^2
                    \right)^{1/2}
                    \left(
                        \sum_P|\Delta_P|^2
                    \right)^{1/2}
                    \\
                    &\leqq
                    \sqrt{\frac{3}{2}}
                    \max_{P=X,Y,Z}|\Delta_P|.
                \end{aligned}
                \label{eq:general-c-error-from-pauli}
            \end{align}
            Taking $\eta:=\sqrt{2/3}\,r\sqrt{\epsilon}$ as the deviation threshold, the condition $\max_{P=X,Y,Z}|\Delta_P|\leqq\eta$ ensures $|e|\leqq r\sqrt{\epsilon}$.
            For a fixed $P$, the corresponding Hoeffding bound (Theorem \ref{thm:hoeffding-inequality}) is
            \begin{align}
                \Pr\!\left[
                    |\Delta_P|\geqq\eta
                \right]
                \leqq
                2\exp\!\left(-\frac{N_P\eta^2}{2}\right),
            \end{align}
            and the probability that any of the three Pauli means deviates by at least $\eta$ is bounded by
            \begin{align}
                \begin{aligned}
                    \Pr\!\left[
                        \max_{P=X,Y,Z}
                        |\Delta_P|
                        \geqq\eta
                    \right]
                    &\leqq
                    6\exp\!\left(-\frac{N_Pr^2\epsilon}{3}\right).
                \end{aligned}
            \end{align}
            This probability is at most $\delta$ for $N_P\geqq3(r^2\epsilon)^{-1}\log(6/\delta)$, in which case $|\Delta_P|\leqq\sqrt{2/3}\,r\sqrt{\epsilon}$ for all three Pauli means and $|e|\leqq r\sqrt{\epsilon}$ with probability at least $1-\delta$.
            The correction in Eq.\,\eqref{eq:coherence-estimator-square} satisfies
            \begin{align}
                \left|
                    \widehat{q}-|\widehat{c}|^2
                \right|
                \leqq
                \frac{1}{2(N_P-1)},
                \label{eq:coherence-bias-correction-bound}
            \end{align}
            where we used $1-\widehat{\langle P\rangle}^{\,2}\in[0,1]$ and $\sum_P|\gamma_P|^2=2$.
            Since positive semidefiniteness of the density matrix \eqref{eq:general-determinant} implies $|c|^2\leqq F(1-F)\leqq\epsilon$, the total error is bounded by
            \begin{align}
                \begin{aligned}
                    \left|
                        \widehat{q}-|c|^2
                    \right|
                    &\leqq
                    2|c||e|
                    +
                    |e|^2
                    +
                    \frac{1}{2(N_P-1)}
                    \\
                    &\leqq
                    2r\epsilon
                    +
                    r^2\epsilon
                    +
                    \frac{1}{2(N_P-1)}.
                \end{aligned}
                \label{eq:general-q-error}            \end{align}
            The theorem below uses a smaller tolerance parameter $\eta$ for the three Pauli means to bound this error by $r\epsilon/8$.

            \paragraph*{Combining the two parts.}
            We estimate $s:=F(1-F)=b+|c|^2$ by
            \begin{align}
                \widehat{s}:=\widehat{b}+\widehat{q}. \label{eq:general-s-estimator}
            \end{align}
            If the fidelity is sufficiently large $F\geqq1/2$, the infidelity is
            \begin{align}
                \epsilon=1-F=\frac{1-\sqrt{1-4s}}{2}.
            \end{align}
            For finite data, we clip $\widehat{s}$ to $[0,1/4]$ before inverting the relation between $s$ and $\epsilon$
            \begin{align}
                \begin{aligned}
                    \widetilde{s}
                    &:=
                    \min\!\left\{\frac{1}{4},\max\{0,\widehat{s}\}\right\},
                    \\
                    \widehat{\epsilon}
                    &:=
                    \frac{1-\sqrt{1-4\widetilde{s}}}{2},
                    \qquad
                    \widehat{F}:=1-\widehat{\epsilon}.
                \end{aligned}
                \label{eq:general-final-estimator}
            \end{align}
            The following theorem specifies the sample count required for relative error $r$ with failure probability at most $\delta$.
            \begin{thm}[Bell-coherence fidelity estimation] \label{thm:sample-complexity}
                Suppose $0<r\leqq1$ and $0<\delta<1/8$.
                For any single-qubit target state $\ket{\psi}$ and any state $\rho$ with $F=\bra{\psi}\rho\ket{\psi}=1-\epsilon$ and $0<\epsilon\leqq1/4$, the estimator in Eqs.\,\eqref{eq:coherence-estimator-square}, \eqref{eq:general-s-estimator}, and \eqref{eq:general-final-estimator} satisfies
                \begin{align}
                    |\widehat{\epsilon}-\epsilon|
                    \leqq
                    r\epsilon
                \end{align}
                with probability at least $1-\delta$ using
                \begin{align}
                    N_{\rm{total}}
                    =
                    O\!\left(
                        \frac{1}{r^2\epsilon}
                        \log\frac{1}{\delta}
                    \right)
                \end{align}
                copies.
            \end{thm}
            \begin{proof}
                Suppose
                \begin{align}
                    N_B
                    &:=
                    \left\lceil
                        \frac{192}{r^2\epsilon}
                        \log\frac{4}{\delta}
                    \right\rceil, \label{eq:nb-choice}
                    \\
                    N_P
                    &:=
                    \left\lceil
                        \frac{3072}{r^2\epsilon}
                        \log\frac{12}{\delta}
                    \right\rceil
                    +1. \label{eq:n-choice}
                \end{align}
                Here $N_B$ is the number of trials in the Bell part, and $N_P$ is the number of outcomes measured in each Pauli basis in the coherence part.
                In the following discussion, we are going to allow an error of at most $r\epsilon/8$ to each estimator.
                Since $b\leqq\epsilon$, if we instead take $t:=r\epsilon/8$ in the Bernstein inequality \eqref{eq:berenstein-in-bell}, the denominator is bounded by
                \begin{align}
                    2b+\frac{2t}{3}
                    \leqq
                    2\epsilon
                    +
                    \frac{r\epsilon}{12}
                    <
                    3\epsilon,
                \end{align}
                so that the probability of getting a deviation greater than $t=r\epsilon/8$ is
                \begin{align}
                    \begin{aligned}
                        \Pr\!\left[
                            |\widehat{b}-b|
                            \geqq
                            \frac{r\epsilon}{8}
                        \right]
                        &\leqq
                        2\exp\!\left(
                            -\frac{N_Bt^2}{3\epsilon}
                        \right)
                        \\
                        &=
                        2\exp\!\left(
                            -\frac{N_Br^2\epsilon}{192}
                        \right)
                        \\
                        &\leqq
                        \frac{\delta}{2}.
                    \end{aligned}
                    \label{eq:single-qubit-explicit-bell-bound}
                \end{align}
                Here, the last inequality follows from the choice of $N_B$ \eqref{eq:nb-choice}.

                For the Pauli measurements in the coherence part, we take $\eta:=\sqrt{2/3}\,r\sqrt{\epsilon}/32$ in lieu of $\eta=\sqrt{2/3}\,r\sqrt{\epsilon}$ in Eq.\,\eqref{eq:general-c-error-from-pauli}, leading to
                \begin{align}
                    \max_{P=X,Y,Z}|\Delta_P|
                    \leqq
                    \eta
                    \quad\Longrightarrow\quad
                    |e|
                    \leqq
                    \frac{r\sqrt{\epsilon}}{32}.
                \end{align}
                By Hoeffding's inequality (Theorem \ref{thm:hoeffding-inequality}) and a union bound over the three Pauli means, the probability that any of the three deviations $|\Delta_P|$ exceeds $\eta$ is bounded as
                \begin{align}
                    \begin{aligned}
                        \Pr\!\left[
                            \max_{P=X,Y,Z}
                            |\Delta_P|
                            \geqq
                            \eta
                        \right]
                        &\leqq
                        6\exp\!\left(-\frac{N_P\eta^2}{2}\right)
                        \\
                        &=
                        6\exp\!\left(
                            -\frac{N_Pr^2\epsilon}{3072}
                        \right)
                        \\
                        &\leqq
                        \frac{\delta}{2},
                    \end{aligned}
                    \label{eq:single-qubit-explicit-pauli-bound}
                \end{align}
                where the last inequality follows from the choice of $N_P$ \eqref{eq:n-choice}.

                For the chosen $N_P$, Eq.\,\eqref{eq:coherence-bias-correction-bound} becomes
                \begin{align}
                    \left|
                        \widehat{q}-|\widehat{c}|^2
                    \right|
                    \leqq
                    \frac{1}{2(N_P-1)}
                    \leqq
                    \frac{r^2\epsilon}{6144\log(12/\delta)}
                    \leqq
                    \frac{r\epsilon}{1024},
                \end{align}
                where the last inequality follows from $r\leqq6\log(12/\delta)$.
                On the event $\max_{P=X,Y,Z}|\Delta_P|\leqq\eta$, positivity of the density matrix implies $|c|^2\leqq\epsilon$, so the error in $\widehat{q}$ obeys
                \begin{align}
                    \begin{aligned}
                        \left|
                            \widehat{q}-|c|^2
                        \right|
                        &\leqq
                        2\sqrt{\epsilon}\,
                        \frac{r\sqrt{\epsilon}}{32}
                        +
                        \frac{r^2\epsilon}{1024}
                        +
                        \frac{r\epsilon}{1024}
                        \\
                        &=
                        r\epsilon
                        \left(
                            \frac{1}{16}
                            +
                            \frac{r}{1024}
                            +
                            \frac{1}{1024}
                        \right)
                        \\
                        &\leqq
                        \frac{33r\epsilon}{512}
                        <
                        \frac{r\epsilon}{8}.
                    \end{aligned}
                    \label{eq:single-qubit-explicit-q-bound}
                \end{align}
                The penultimate inequality uses $r\leqq1$.

                By the union bound, $|\widehat{b}-b|\leqq r\epsilon/8$ and $\bigl|\widehat{q}-|c|^2\bigr|<r\epsilon/8$ hold simultaneously with probability at least $1-\delta$.
                The total error is then bounded by the triangle inequality as
                \begin{align}
                    |\widehat{s}-s|
                    \leqq
                    |\widehat{b}-b|
                    +
                    \left|\widehat{q}-|c|^2\right|
                    \leqq
                    \frac{r\epsilon}{4}.
                    \label{eq:single-qubit-explicit-s-bound}
                \end{align}
                Since $s\in[0,1/4]$, clipping $\widehat{s}$ to this interval cannot increase its distance from $s$, so $|\widetilde{s}-s|\leqq|\widehat{s}-s|$.
                The denominator in the expression below is at least $1/2$ because $\sqrt{1-4s}=1-2\epsilon\geqq1/2$, and then
                \begin{align}
                    |\widehat{\epsilon}-\epsilon|
                    =
                    \frac{2|\widetilde{s}-s|}{\sqrt{1-4\widetilde{s}}+\sqrt{1-4s}}
                    \leqq
                    4|\widehat{s}-s|
                    \leqq
                    r\epsilon.
                \end{align}
                Finally, each Bell part consumes two copies per shot, while the Pauli measurements in the coherence part consume $3N_P$ copies in total.
                We therefore land on the scaling
                \begin{align}
                    N_{\rm{total}}
                    =2N_B+3N_P=
                    O\!\left(
                        \frac{1}{r^2\epsilon}
                        \log\frac{1}{\delta}
                    \right).
                \end{align}
            \end{proof}

            The result above reveals that the $1/\epsilon$ dependence is a quadratic improvement over generic single copy benchmarking.
            This result is extended to the multi-qubit case in Appendix \ref{append:multi_qubit_generalization}.

        \subsection{Optimal dependence on $r$, $\epsilon$ and $\delta$}
            The upper bound of the sample complexity in Theorem \ref{thm:sample-complexity} is not only efficient but optimal.
            Indeed, even arbitrary collective measurements cannot improve its dependence on $\epsilon$, $r$, or $\delta$.
            Hence the lower bound below also applies, a fortiori, to Clifford-only protocols.
            We note that Theorem 4 of Ref.\,\cite{Lee:2025lcs} already establishes $N=\Omega(1/(r^2\epsilon))$ for arbitrary target states at the fixed success probability $2/3$ using the quantum Pinsker inequality.
            To compare this lower bound directly with Theorem \ref{thm:sample-complexity}, however, we state and prove below the optimal bound with explicit dependence on an arbitrary failure probability $\delta$.
            Our argument adapts the two-state construction of Ref.\,\cite{Lee:2025lcs} and replaces the quantum Pinsker inequality with the Bretagnolle--Huber inequality, adding the general factor $\log(1/\delta)$.

            \begin{thm}[Sample optimality] \label{thm:sample-optimality}
                Let $|\psi\rangle$ be a single-qubit target state, and suppose $0<r\leqq1/8$, $0<\delta<1/8$.
                For every sufficiently small $\epsilon>0$, any protocol that estimates the infidelity $\epsilon:=1-\langle\psi|\rho|\psi\rangle$ of an arbitrary single-qubit state $\rho$ to relative error $r$ with probability at least $1-\delta$ requires, in the worst case,
                \begin{align}
                    N
                    =
                    \Omega\!\left(
                        \frac{1}{r^2\epsilon}
                        \log\frac{1}{\delta}
                    \right)
                \end{align}
                copies.
                Therefore Theorem \ref{thm:sample-complexity} has optimal dependence on $\epsilon$, $r$, and $\delta$ up to constant factors.
            \end{thm}
            \begin{proof}
                Choose a state $|\psi^\perp\rangle$ orthogonal to $|\psi\rangle$ and consider the commuting one-parameter family
                \begin{align}
                    \rho_p
                    :=
                    (1-p)|\psi\rangle\langle\psi|
                    +
                    p|\psi^\perp\rangle\langle\psi^\perp| .
                \end{align}
                To start, we take
                \begin{align}
                    \rho_0:=\rho_{\epsilon},
                    \qquad
                    \rho_1:=\rho_{(1+4r)\epsilon},
                \end{align}
                whose infidelities are
                \begin{align}
                    1-\langle\psi|\rho_0|\psi\rangle=\epsilon,
                    \qquad
                    1-\langle\psi|\rho_1|\psi\rangle=(1+4r)\epsilon.
                \end{align}

                For a successful estimate, we demand that the estimator satisfy $\widehat{\epsilon}\leqq(1+r)\epsilon$ for $\rho_0$ and $\widehat{\epsilon}\geqq(1-r)(1+4r)\epsilon$ for $\rho_1$.
                These bounds are separated because, for $0<r\leqq 1/8$,
                \begin{align}
                    (1-r)(1+4r)
                    =
                    1+3r-4r^2
                    >
                    1+r.
                \end{align}
                Now, take the threshold
                \begin{align}
                    T
                    :=
                    \frac{1}{2}
                    \Big[
                        (1+r)\epsilon
                        +
                        (1-r)(1+4r)\epsilon
                    \Big],
                \end{align}
                so that $(1+r)\epsilon<T<(1-r)(1+4r)\epsilon$.
                We then define a binary decision rule $\varphi$ from the estimator $\widehat\epsilon$ by
                \begin{align} \label{eq:decision-rule}
                    \varphi(\widehat\epsilon)
                    :=
                    \begin{cases}
                        0, & \widehat\epsilon\leq T,\\
                        1, & \widehat\epsilon>T,
                    \end{cases}
                \end{align}
                where the decision $\varphi=0$ means that the input is identified as $\rho_0^{\otimes N}$, while $\varphi=1$ means that the input is identified as $\rho_1^{\otimes N}$.

                Let $\mathbb{P}_i$ denote the probability distribution of $\widehat\epsilon$ when the input state is $\rho_i^{\otimes N}$, for $i=0,1$.
                The assumed success probability under $\rho_0^{\otimes N}$ means that
                \begin{align}
                    \mathbb{P}_0
                    \left[
                        \widehat\epsilon\leq(1+r)\epsilon
                    \right]
                    \geq
                    1-\delta .
                \end{align}
                Since $(1+r)\epsilon<T$, we have $\mathbb{P}_0[\varphi=1]\leq\delta$.
                Likewise, the assumed success probability under $\rho_1^{\otimes N}$ means that
                \begin{align}
                    \mathbb{P}_1
                    \left[
                        \widehat\epsilon\geq(1-r)(1+4r)\epsilon
                    \right]
                    \geq
                    1-\delta,
                \end{align}
                and $T<(1-r)(1+4r)\epsilon$ ensures that $\mathbb{P}_1[\varphi=0]\leq\delta$.
                Therefore the average error probability of the test $\varphi$ is
                \begin{align}
                    P_{\rm err}(\varphi)
                    :=
                    \frac{1}{2}\mathbb{P}_0[\varphi=1]
                    +
                    \frac{1}{2}\mathbb{P}_1[\varphi=0]
                    \leq
                    \delta,
                \end{align}
                when equal prior probability $1/2$ is assigned to each state.

                It remains to lower-bound the number of copies required for this binary distinguishing problem.
                Since $\rho_0$ and $\rho_1$ commute, their quantum relative entropy $D$ is the classical Bernoulli relative entropy
                \begin{align}
                    \begin{aligned}
                        D(\rho_0\Vert\rho_1)
                        &=
                        \epsilon
                        \log\frac{\epsilon}{(1+4r)\epsilon}
                        +
                        (1-\epsilon)
                        \log\frac{1-\epsilon}{1-(1+4r)\epsilon} \\[0.5em]
                        &=D_{\mathrm{KL}}\!\left(
                                \operatorname{Bern}(\epsilon)
                                \Vert
                                \operatorname{Bern}((1+4r)\epsilon)
                            \right).
                    \end{aligned}
                \end{align}
                Using the bound for the Bernoulli distribution
                \begin{align}
                    D_{\mathrm{KL}}\!\left(
                        \operatorname{Bern}(p)
                        \Vert
                        \operatorname{Bern}(q)
                    \right)
                    \leqq
                    \frac{(p-q)^2}{q(1-q)},
                \end{align}
                for $p=\epsilon$ and $q=(1+4r)\epsilon$, the right hand side satisfies
                \begin{align}
                    D(\rho_0\Vert\rho_1)
                    \leqq
                    \frac{16r^2\epsilon}
                    {(1+4r)(1-(1+4r)\epsilon)}
                    \leqq
                    32r^2\epsilon,
                \end{align}
                and hence
                \begin{align} \label{eq:relative-entropy-bound}
                    D(\rho_0^{\otimes N}\Vert\rho_1^{\otimes N})
                    =
                    N D(\rho_0\Vert\rho_1)
                    \leqq
                    32Nr^2\epsilon.
                \end{align}

                For any two-outcome measurement $\{M_0,M_1\}$ used to distinguish $\rho_0^{\otimes N}$ from $\rho_1^{\otimes N}$ with equal prior probability $1/2$, define its average error probability by
                \begin{align}
                    P_{\rm err}(M_0,M_1)
                    :=
                    \frac{1}{2}\operatorname{tr}\!\left(M_1\rho_0^{\otimes N}\right)
                    +
                    \frac{1}{2}\operatorname{tr}\!\left(M_0\rho_1^{\otimes N}\right),
                \end{align}
                where outcome $0$ means guessing $\rho_0^{\otimes N}$ and outcome $1$ means guessing $\rho_1^{\otimes N}$.
                The measurement can be represented as a completely positive trace-preserving (CPTP) map that stores its outcome in a classical register,
                \begin{align}
                    \mathcal{M}(\sigma)
                    :=\sum_{a=0,1}\operatorname{tr}(M_a\sigma)\ket{a}\!\bra{a}.
                \end{align}
                For input $\rho_i^{\otimes N}$, the outcome distribution is $Q_i(a):=\operatorname{tr}(M_a\rho_i^{\otimes N})$, so the output state is diagonal with entries $Q_i(a)$.
                For these diagonal output states, quantum relative entropy equals the classical relative entropy of the two outcome distributions, and by the data processing inequality for quantum relative entropy,
                \begin{align}
                    \begin{aligned}
                        D_{\mathrm{KL}}(Q_0\Vert Q_1)
                        &=D\!\left(\mathcal{M}(\rho_0^{\otimes N})
                        \Vert\mathcal{M}(\rho_1^{\otimes N})\right)
                        \\
                        &\leqq D(\rho_0^{\otimes N}\Vert\rho_1^{\otimes N}).
                    \end{aligned}
                \end{align}
                Applying the classical Bretagnolle--Huber inequality (Theorem \ref{thm:bretagnolle-huber-inequality}) to $Q_0$ and $Q_1$, we obtain
                \begin{align}
                    \begin{aligned}
                        P_{\rm{err}}(M_0,M_1)
                        &\geqq\frac{1}{4}\exp\bigl(-D_{\mathrm{KL}}(Q_0\Vert Q_1)\bigl)
                        \\
                        &\geqq\frac{1}{4}\exp\bigl(-D(\rho_0^{\otimes N}\Vert\rho_1^{\otimes N})\bigl).
                    \end{aligned}
                \end{align}
                We apply this general bound to the two-outcome measurement obtained by grouping the estimator's outcomes according to the decision rule \eqref{eq:decision-rule}.
                Since $P_{\rm err}(\varphi)\leqq\delta$,
                \begin{align}
                    \delta
                    \geqq
                    \frac{1}{4}
                    \exp\bigl(
                        -N D(\rho_0\Vert\rho_1)
                    \bigl).
                \end{align}
                Combining this inequality with Eq.\,\eqref{eq:relative-entropy-bound},
                \begin{align}
                    N
                    \geqq
                    \frac{1}{32r^2\epsilon}
                    \log\frac{1}{4\delta}
                    \geqq
                    \frac{1}{96r^2\epsilon}
                    \log\frac{1}{\delta},
                \end{align}
                where the second inequality follows from $\delta<1/8$.
            \end{proof}

    \section{Application: the analog rotation resource state in the STAR architecture} \label{sec:theory-and-simulation}
        We now bring the general protocol to bear on the analog rotation resource state used in the STAR architecture \cite{Akahoshi:2023xck, Toshio:2024hut, Toshio:2026dpa}.
        The target resource state is the analog rotation resource state
        \begin{align} \label{eq:analog-rotation-resource-state}
            \ket*{m_{\theta}}:=R_Z(\theta)\!\ket{+}=\cos\theta\ket{+}+i\sin\theta\ket{-},
        \end{align}
        where we employ the notation $R_Z(\theta):=e^{iZ\theta}$ rather than the conventional $R_Z(\theta):=e^{-\frac{i}{2}Z\theta}$ for convenience.
        In the STAR architecture, this state is consumed directly as an alternative to implementing arbitrary rotations through high overhead magic $T$ state preparation and gate synthesis \cite{Akahoshi:2023xck, Toshio:2024hut, Toshio:2026dpa, Dawson:2005blj}.
        A self-contained review of the STAR architecture and the role of this resource state is given in more detail in Appendix \ref{append:star}.

        For later use, the tensor product of the two rotation magic states with different angles $\theta$, $\theta'$ is
        \begin{align}
            \begin{aligned}
                \ket*{m_{\theta}}\otimes\!\ket*{m_{\theta'}}
                &=
                \frac{1}{\sqrt{2}}\Big(
                    \cos(\theta+\theta')\ket*{\Phi_+}
                    \\
                    &\quad
                    +i\sin(\theta+\theta')\ket*{\Phi_-}
                    \\
                    &\quad
                    +\cos (\theta-\theta')\ket*{\Psi_+}
                    \\
                    &\quad
                    +i\sin(\theta-\theta')\ket*{\Psi_-}
                \Big),
            \end{aligned}
        \end{align}
        which simplifies when $\theta'=\theta$
        \begin{align}
            \begin{aligned}
                \ket*{m_{\theta}}\otimes\!\ket*{m_{\theta}}
                &=
                \frac{1}{\sqrt{2}}\Big(
                    \cos 2\theta\ket*{\Phi_+}
                    \\
                    &\quad
                    +i\sin 2\theta\ket*{\Phi_-}
                    +\ket*{\Psi_+}
                \Big).
            \end{aligned}
        \end{align}

        \subsection{Physical level} \label{subsec:physical_level}
            Before treating $\ket{m_{\theta}}$ as a logical state of a surface code patch, which is how the STAR architecture consumes it, we first perform the numerical simulations of the protocol on bare physical qubits.
            This allows us to isolate the advantage in Theorem \ref{thm:sample-complexity} from technical details in QEC that are inessential to the argument, such as encoding, transversal gate implementations, and decoding.
            In this subsection, we assume an idealized setting where noise enters only through the analog rotation that prepares the resource state, with the remaining operations, namely the Clifford operations or Pauli measurements, treated as noiseless.
            Although this is unrealistic on an actual noisy quantum device, we neglect additional error sources to focus on the effect of noise in the resource state preparation.
            We leave the more realistic fault-tolerant setting for Subsec.\,\ref{subsec:logical_level}, where errors also affect the other operations yet the advantage of the protocol is shown to survive under such noise.

            To clarify, the target of this subsection is to estimate the fidelity $F:=\bra*{m_{\theta}}\rho\ket*{m_{\theta}}$.
            With $\ket*{m_\theta^\perp}:=Z\ket*{m_\theta}$ and $c:=\bra*{m_\theta}\rho\ket*{m_\theta^\perp}$, the Bell part of Sec.\,\ref{sec:general_protocol} is            \begin{align}
                b
                :=
                \bra*{\Psi_-}\rho\otimes\rho\ket*{\Psi_-}
                =
                \det\rho
                =
                F(1-F)-|c|^2.
            \end{align}
            In principle, the BCFE protocol applies to any states under any noise model as we see in Sec.\,\ref{sec:general_protocol}, but we first restrict our attention to a specific error channel to exemplify essential physics of the protocol.
            For this purpose, we consider the depolarizing channel
            \begin{align}
                \mathcal{E}_0(\rho)=p\rho+\sum_{P=X,Y,Z}\epsilon_P P\rho P,
                \label{eq:physical-pauli-channel}
            \end{align}
            during state preparation \eqref{eq:analog-rotation-resource-state}, where $p:=1-\epsilon_X-\epsilon_Y-\epsilon_Z$.
            As mentioned above, we assume only this channel is noisy at the physical level, so the Bell measurement in the circuit of Fig.\,\ref{fig:backward} and the Pauli measurements used below carry no faults of their own.
            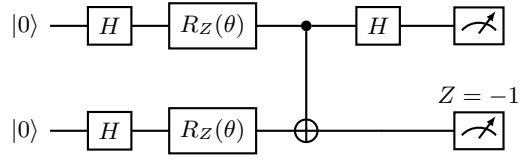
\begin{figure*}
                \centering
                \begin{quantikz}
                    \lstick{$\ket{0}$} & \gate{H} & \gate{R_Z(\theta)} & \ctrl{1} & \gate{H} & \meter{} \\
                    \lstick{$\ket{0}$} & \gate{H} & \gate{R_Z(\theta)} & \targ{}  & \qw & \meter{$Z=-1$} \\
                \end{quantikz}
                \caption{The first Hadamard gate and $R_Z(\theta)$ prepare the resource states.
                The subsequent $\CNOT$ and Hadamard gate implement the Bell measurement.}
                \label{fig:backward}
            \end{figure*}
            The prepared state is $\rho:=\mathcal{E}_0(\ket*{m_\theta}\!\bra*{m_\theta})$, and the projector onto the measurement outcome $Z=-1$ is $\Pi_{Z,-}:=(I-Z)/2$.
            After the CNOT, the normalized state of the first qubit conditioned on obtaining $Z_2=-1$ is
            \begin{align}
                \begin{aligned}
                    \rho'
                    &:=
                    \frac{1}{\operatorname{Pr}(Z_2=-1)}
                    \operatorname{Tr}_2\!\left[
                        (I\otimes\Pi_{Z,-})
                        \right.
                    \\
                    &\qquad\left.
                        \CNOT_{12}
                        (\rho\otimes\rho)
                        \CNOT_{12}^{\dagger}
                        (I\otimes\Pi_{Z,-})
                    \right].
                \end{aligned}
            \end{align}
            The final Hadamard gate converts the computational-basis measurement of the first qubit into an $X$ measurement before that gate.
            The conditional probability $\operatorname{Tr}[\rho'(I-X_1)/2]$ under the channel in Eq.\,\eqref{eq:physical-pauli-channel} is
            \begin{align}
                \begin{aligned}
                    &\operatorname{Pr}(X_1=-1\mid Z_2=-1)
                    \\
                    &=
                    2\epsilon_X\sin^2(2\theta)
                    +
                    2\epsilon_Y\cos^2(2\theta)
                    +
                    2\epsilon_Z
                    \\
                    &\quad
                    -
                    2\epsilon_X^2\sin^2(2\theta)
                    -
                    2\epsilon_Y^2\cos^2(2\theta)
                    -
                    2\epsilon_Z^2
                    \\
                    &\quad
                    -
                    4\epsilon_X\epsilon_Z\sin^2(2\theta)
                    -
                    4\epsilon_Y\epsilon_Z\cos^2(2\theta).
                \end{aligned}
            \end{align}
            For the isotropic case $\epsilon_X=\epsilon_Y=\epsilon_Z=\epsilon$,
            \begin{align}
                \begin{aligned}
                    \operatorname{Pr}(X_1=-1\mid Z_2=-1)
                    &=
                    4\epsilon(1-3\epsilon)+4\epsilon^2
                    \\
                    &=
                    4\epsilon-8\epsilon^2.
                \end{aligned}
            \end{align}
            This conditional probability is $O(\epsilon)$, so the same variance argument used for the Bell estimator gives a sample complexity of $O(1/(r^2\epsilon))$ for additive precision $r\epsilon$ at fixed confidence.
            Resolving it to additive precision $r\epsilon$ therefore requires $N=O(1/(r^2\epsilon))$ samples, in agreement with Theorem \ref{thm:sample-complexity}.

            \paragraph{Estimators used in the simulation.}
                For the target resource state $\ket*{m_\theta}$, the fidelity can be written as
                \begin{align}
                    F
                    =
                    \frac{1}{2}
                    \big(
                        1
                        +\cos(2\theta)\langle X\rangle
                        -\sin(2\theta)\langle Y\rangle
                    \big).
                \end{align}
                For this target, Eq.\,\eqref{eq:general-coherence} specializes to $c=\frac{1}{2}\big(\langle Z\rangle+ia\big)$ with $a:=\sin(2\theta)\langle X\rangle+\cos(2\theta)\langle Y\rangle$, for which
                \begin{align} \label{eq:physical-coherence-specialization}
                    |c|^2
                    =
                    \frac{1}{4}
                    \left(
                        \langle Z\rangle^2+a^2
                    \right).
                \end{align}
                Under the depolarizing channel in Eq.\,\eqref{eq:physical-pauli-channel}, $\langle Z\rangle=0$ follows from the action of the three errors,
                \begin{align}
                    \begin{aligned}
                        X\ket*{m_\theta}
                        &=
                        \ket*{m_{-\theta}},
                        \\
                        Y\ket*{m_\theta}
                        &=
                        -\ket*{m_{\pi/2-\theta}},
                        \\
                        Z\ket*{m_\theta}
                        &=
                        -i\ket*{m_{\theta+\pi/2}},
                    \end{aligned}
                \end{align}
                and the fact that $\bra*{m_\varphi}Z\ket*{m_\varphi}=0$ for any angle $\varphi$.
                The numerical implementation nevertheless measures all three Pauli observables and applies Eq.\,\eqref{eq:physical-coherence-specialization}, without substituting the model specific identity $\langle Z\rangle=0$.

                For $N_P$ independent outcomes $A_{P,j}\in\{+1,-1\}$ in basis $P$, define
                \begin{align}
                    \begin{aligned}
                        \widehat\mu_P
                        &:=
                        \frac{1}{N_P}
                        \sum_{j=1}^{N_P}A_{P,j},
                        \\
                        \widehat{\mu_P^2}
                        &:=
                        \frac{N_P\widehat\mu_P^2-1}{N_P-1}.
                    \end{aligned}
                \end{align}
                Since $A_{P,j}^2=1$, this correction satisfies $\mathbb{E}[\widehat{\mu_P^2}]=\langle P\rangle^2$.

                Using these estimators, an unbiased estimator of $|c|^2$ is given by
                \begin{align}
                    \begin{aligned}
                        \widehat q
                        &:=
                        \frac{1}{4}
                        \Big(
                            \widehat{\mu_Z^2}
                            +\sin^2(2\theta)\widehat{\mu_X^2}
                            +\cos^2(2\theta)\widehat{\mu_Y^2}
                            \\
                            &\hspace{7em}
                            +2\sin(2\theta)\cos(2\theta)
                            \widehat\mu_X\widehat\mu_Y
                        \Big).
                    \end{aligned}
                    \label{eq:physical-unbiased-q}
                \end{align}
                The physical simulation combines $\widehat b$ and $\widehat q$ through Eqs.\,\eqref{eq:general-s-estimator} and \eqref{eq:general-final-estimator}.
                Theorem \ref{thm:sample-complexity} already established the corresponding multiplicative error guarantee and sample complexity, so we do not repeat that analysis here.

                The error bars in the BCFE simulation are approximate standard errors computed by the first-order delta method from independent Bell and Pauli samples,
                \begin{align}
                    \widehat{\mathrm{SE}}(\widehat\epsilon)
                    =
                    \frac{1}{\sqrt{1-4\widetilde s}}
                    \left[
                        \frac{\widehat b(1-\widehat b)}{N_B}
                        +\sum_{P=X,Y,Z}
                        \left(\frac{\partial\widehat q}{\partial\widehat\mu_P}\right)^2
                        \frac{1-\widehat\mu_P^2}{N_P}
                    \right]^{1/2}.
                    \label{eq:bcfe-standard-error}
                \end{align}
                The derivatives include the bias correction in Eq.\,\eqref{eq:physical-unbiased-q}.
                The error bars for the Bell and coherence parts are the square roots of the first term and the sum, respectively.
                Here $N_B$ counts measured Bell pairs and $N_P$ counts outcomes in basis $P$.
                Clipping is applied to the point estimate, and the derivative of the square-root inversion is evaluated at $\widetilde s$, without accounting for the change in the sampling distribution caused by clipping.

                \begin{figure*}[t]
                    \centering
                    \begin{subfigure}[t]{0.48\textwidth}
                        \centering
                        \includegraphics[width=0.8\linewidth]{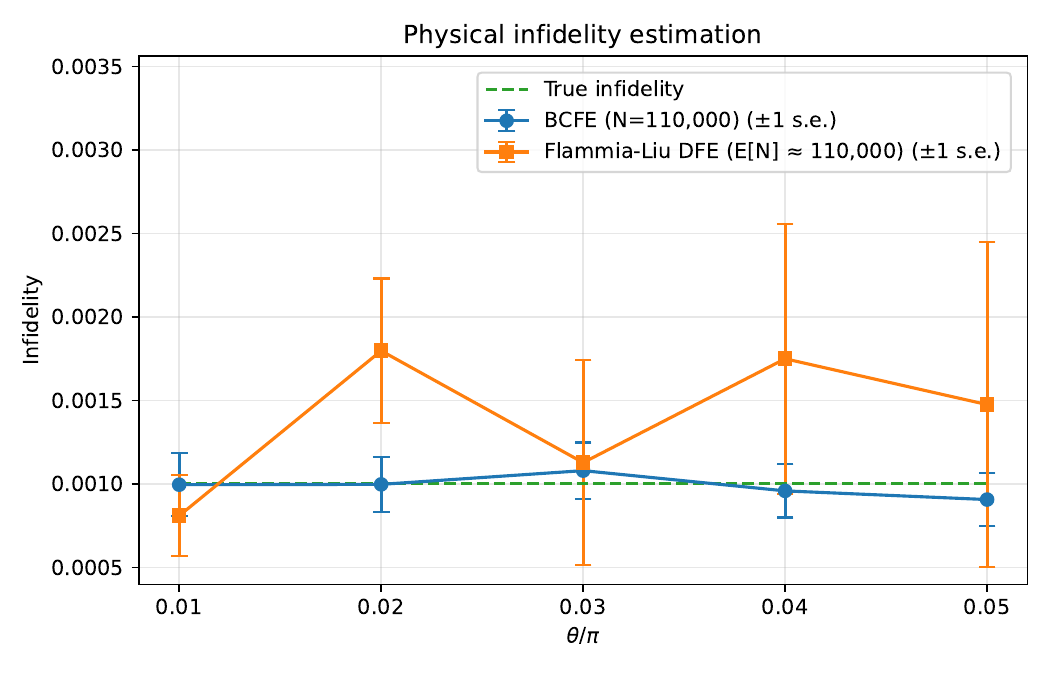}
                        \caption{$N_B=40{,}000$, $N_P=10{,}000$, and $N=110{,}000$.}
                    \end{subfigure}
                    \hfill
                    \begin{subfigure}[t]{0.48\textwidth}
                        \centering
                        \includegraphics[width=0.8\linewidth]{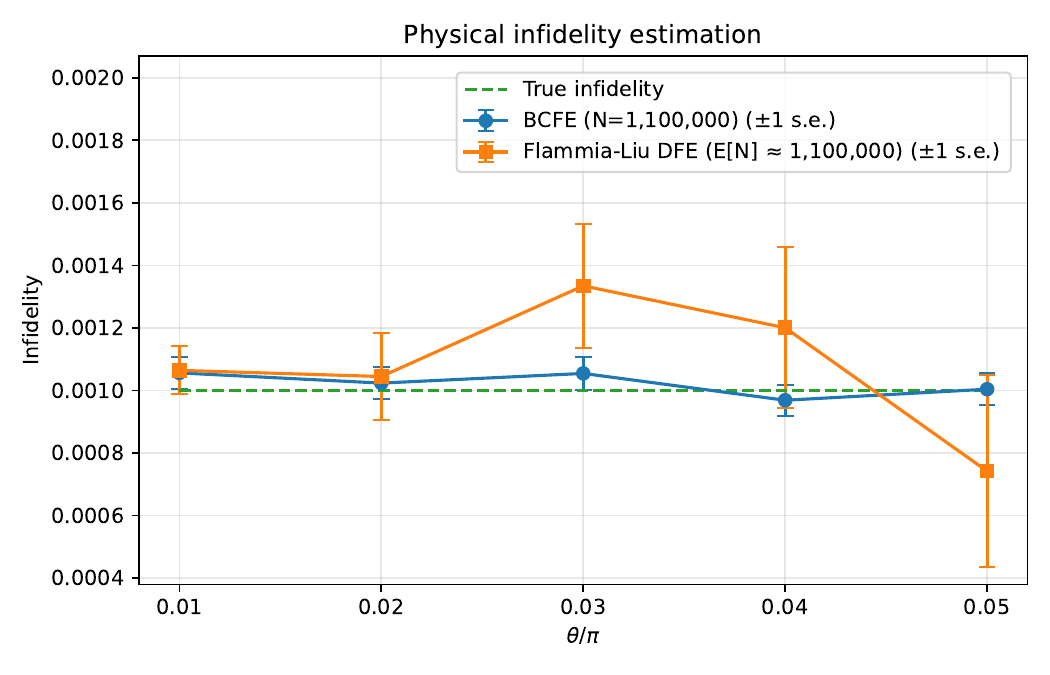}
                        \caption{$N_B=400{,}000$, $N_P=100{,}000$, and $N=1{,}100{,}000$.}
                    \end{subfigure}
                    \caption{Physical level fidelity estimation under the isotropic depolarizing channel with $\epsilon_X=\epsilon_Y=\epsilon_Z=5\times10^{-4}$.
                    The green dashed curve is the exact infidelity; blue circles and orange squares denote our BCFE protocol and Flammia--Liu DFE, respectively.
                    BCFE uses exactly $N$ resource states, while DFE uses approximately $N$ in expectation, with a fluctuating realized count.
                    Error bars indicate one standard error.
                    BCFE error bars use the delta-method approximation in Eq.\,\eqref{eq:bcfe-standard-error}.}
                    \label{fig:physical-simulation}
                \end{figure*}

                \begin{figure*}[t]
                    \centering
                    \begin{subfigure}[t]{0.48\textwidth}
                        \centering
                        \includegraphics[width=0.8\linewidth]{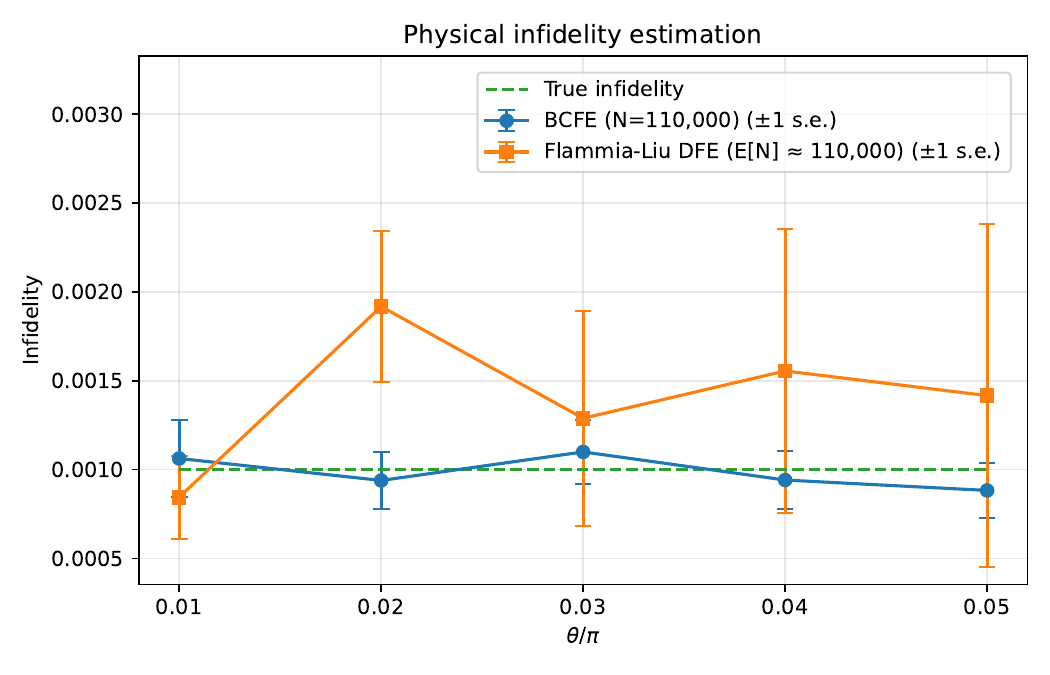}
                        \caption{$N_B=40{,}000$, $N_P=10{,}000$, and $N=110{,}000$.}
                    \end{subfigure}
                    \hfill
                    \begin{subfigure}[t]{0.48\textwidth}
                        \centering
                        \includegraphics[width=0.8\linewidth]{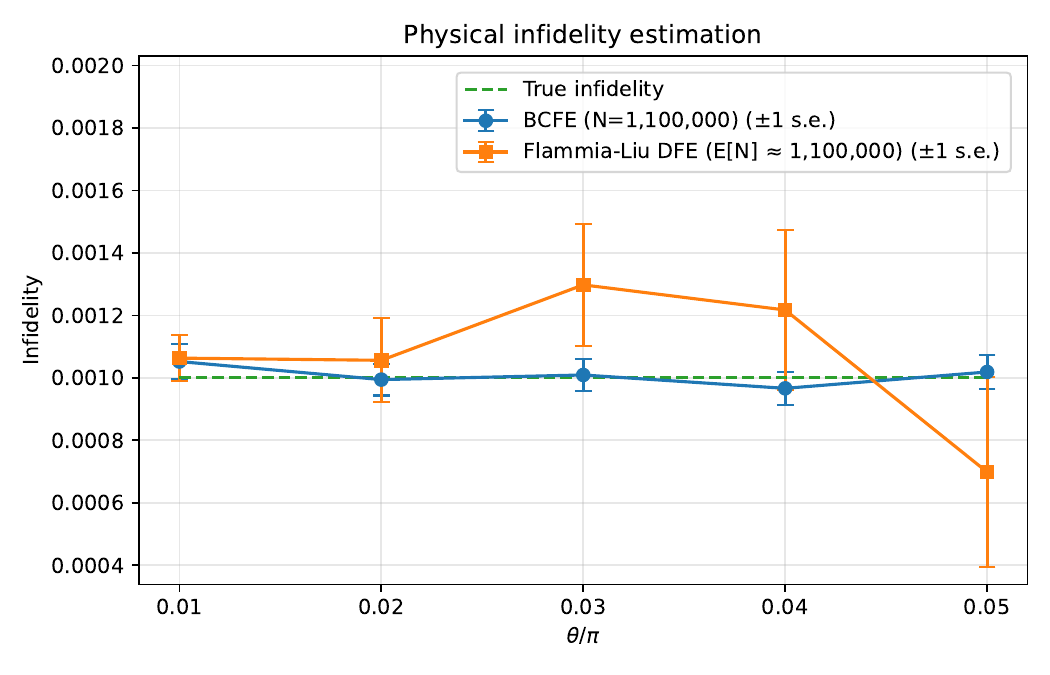}
                        \caption{$N_B=400{,}000$, $N_P=100{,}000$, and $N=1{,}100{,}000$.}
                    \end{subfigure}
                    \caption{Physical level fidelity estimation for the density matrix in Eq.\,\eqref{eq:physical-general-noisy-state}, without assuming a depolarizing channel.
                    Plotting conventions and the comparison of expected copy counts are the same as in Fig.\,\ref{fig:physical-simulation}.
                    Error bars indicate one standard error.
                    BCFE error bars use the delta-method approximation in Eq.\,\eqref{eq:bcfe-standard-error}.}
                    \label{fig:physical-simulation-general}
                \end{figure*}

                We compare this protocol with Flammia--Liu direct fidelity estimation (DFE) \cite{Flammia:2011fda} at approximately equal expected numbers of prepared resource states.
                For $N_B$ Bell shots and $N_P$ measurements in each Pauli basis, BCFE consumes $N=2N_B+3N_P$ copies of the analog rotation state $\ket{m_{\theta}}$, and the DFE parameters are chosen to match this expected cost up to integer rounding.
                See Appendix \ref{append:flammia_liu} for the DFE sampling procedure and its statistical properties.

                Fig.\,\ref{fig:physical-simulation} exhibits the comparison for the isotropic specialization of Eq.\,\eqref{eq:physical-pauli-channel}, with $\epsilon_X=\epsilon_Y=\epsilon_Z=5\times10^{-4}$ and hence $1-F=10^{-3}$.
                We use Qulacs \cite{Suzuki:2020nbf} for all Pauli and Bell measurement samples.
                At $N=110{,}000$, the proposed estimator already lies close to the exact infidelity, while the DFE standard errors are larger and increase with the angle $\theta$.
                At $N=1{,}100{,}000$, the proposed BCFE estimates nearly coincide with the exact curve throughout the angle range, but the DFE error bars remain visibly larger.

                At this point, some reader may not yet be convinced that the observed statistical advantage is not specific to the depolarizing channel.
                However, this advantage is not because of the choice of noise channel.
                To ground this claim, we further consider a noisy state specified directly in the target basis
                \begin{align}
                    \begin{aligned}
                        \rho
                        &=
                        \begin{pmatrix}
                            1-\epsilon & c \\
                            c^* & \epsilon
                        \end{pmatrix},
                        \qquad
                        \epsilon=10^{-3},
                        \\
                        c
                        &=
                        0.20\sqrt{(1-\epsilon)\epsilon}\,e^{0.35\pi i}.
                    \end{aligned}
                    \label{eq:physical-general-noisy-state}
                \end{align}
                This state has comparably large coherence in the target basis and is not obtained by applying the Pauli channel in Eq.\,\eqref{eq:physical-pauli-channel}.
                Fig.\,\ref{fig:physical-simulation-general} confirms the same ordering of statistical errors, with smaller estimated standard errors for BCFE at both budgets across the sampled angles.
                Theorem \ref{thm:sample-complexity} holds for arbitrary $\rho$, so the sample complexity advantage is independent of the error channel.
                This result demonstrates that the same finite sample advantage persists for noisy density matrices beyond the depolarizing channel.

        \subsection{Logical level} \label{subsec:logical_level}
            We next turn to the logical level, where the protocol runs on encoded logical qubits and its Bell measurement must itself be made fault-tolerant.
            The physical level simulations in the preceding subsection assume ideal measurements after the preparation of the analog rotation resource state $\ket{m_{\theta}}$, so it is unclear whether the advantage in sample complexity (Figs.\,\ref{fig:physical-simulation},\,\ref{fig:physical-simulation-general}) can survive under the realistic noise that afflicts an actual device.
            The logical simulations therefore repeat the comparison of Subsec.\,\ref{subsec:physical_level} under the circuit-level noise.
            Note that the following FTQC setup is exactly the one adopted in the STAR architecture (see Appendix \ref{append:star} for a review).

            Consider two distance-$d$ rotated surface code patches for the Bell part, as illustrated for $d=3$ in Fig.\,\ref{fig:two-surface-code-patches}.
            To reduce the cost on a classical computer, the simulations allocate one ancilla instead of $d^2-1$ ancillas to each patch for the stabilizer measurements following Ref.\,\cite{Katsuda:2022mcn}.
            These measurements are performed sequentially, and the ancilla is reset and reused after each readout.
            The logical resource state $\ket{m_{\theta,L}}$ is prepared by applying $e^{i\alpha Z}$ along a path $\gamma$ on which the logical $Z$ on the surface code acts (see Fig.\,\ref{fig:two-surface-code-patches}), followed by post-selection on the trivial stabilizer syndrome \cite{Choi:2023bvh}
            \begin{align}
                \begin{aligned}
                    \prod_{v\in\gamma}e^{i\alpha Z_v}\ket*{+_L}
                    &=
                    \sum_{n=0}^{|\gamma|}
                    i^n\cos^{|\gamma|-n}\!\alpha\,\sin^n\!\alpha
                    \\
                    &\quad\times
                    \sum_{\substack{\gamma'\subset\gamma \\ |\gamma'|=n}}
                    \prod_{i\in\gamma'}Z_i\ket*{+_L}
                    \\
                    &=
                    \cos^d\alpha\ket*{+_L}
                    +i^d\sin^d\alpha\ket*{-_L}
                    \\
                    &\quad
                    +(\text{states outside the logical subspace})
                    \\
                    &\rightarrow
                    \frac{
                        \cos^d\alpha\ket*{+_L}
                        +i^d\sin^d\alpha\ket*{-_L}
                    }{
                        \sqrt{\sin^{2d}\alpha+\cos^{2d}\alpha}
                    } .
                \end{aligned}
                \label{eq:physical_rotation}
            \end{align}
            For odd $d$, the physical angle $\alpha$ and the target logical angle $\theta$ satisfy
            \begin{align}
                \theta=\arctan\bigl((-1)^{(d-1)/2}\tan^d\alpha\bigl),
            \end{align}
            where we choose $\alpha,\theta\in(-\pi/2,\pi/2)$.
            For the $d=3$ simulations with $\theta>0$, we therefore choose $\alpha=-\arctan((\tan\theta)^{1/3})$.
            The post-selection also determines which resource preparation faults enter the accepted state.
            A single $X$ or $Y$ fault after a physical rotation flips a $Z$-type stabilizer syndrome and therefore has zero probability of passing the ideal post-selection.
            Such a fault changes the acceptance rate but does not contribute to the conditional state.
            Motivated by this observation, we adopt a dephasing channel for resource preparation, with an independent $Z$ fault after each physical rotation on $\gamma$
            \begin{align} \label{eq:dephasing-channel}
                \mathcal{E}_Z(\rho)
                =
                (1-p_Z)\rho + p_Z Z\rho Z.
            \end{align}

            To examine the protocol under circuit-level noise, the logical simulations instead implement the Bell measurement with the circuit in Fig.\,\ref{fig:logical-physical-circuit} for the $d=3$ surface codes in Fig.\,\ref{fig:two-surface-code-patches}.
            It is the logical analog of Fig.\,\ref{fig:backward}, and each patch is first encoded into $\ket*{+_L}$ from physical $\ket*{0}^{\otimes 9}$ states.
            The two resource states are then generated by physical rotations along a length three logical path and projected back to the code subspace.

            In practice, each patch undergoes three stabilizer measurement rounds denoted by $\Pi_L$, and only a syndrome whose outcomes are all zero is accepted to suppress the effect of noise during projection.
            The ancilla assigned to each patch is reset after every stabilizer measurement, and the CNOT ordering orients any resulting hook error so that it does not reduce the effective code distance \cite{Katsuda:2022mcn}.
            The Bell part applies a transversal $\CNOT$ from the first patch to the second, measures every data qubit in the second patch destructively in the $Z$ basis, and post-selects the outcome $Z_L^{(2)}=-1$ after classical decoding.
            The surviving trials then proceed to destructive physical $X$ measurement and classical decoding on the first patch.
            Since the $\CNOT$ is transversal, a single gate fault produces at most one data error in each patch.
            The resulting errors in the logical $X$ and $Z$ readout are corrected by a classical decoder using minimum weight perfect matching (MWPM) on the physical measurement outcomes.

            The corresponding single patch circuit used for the coherence term $|c|^2$ is shown in Fig.\,\ref{fig:logical-coherence-circuit}.
            It uses the same resource state preparation and projection, followed by one of the three destructive logical Pauli $X$, $Y$, and $Z$ measurements.
            The destructive logical $X$ and $Z$ measurements are also decoded from all nine physical outcomes by the MWPM decoder.
            The $Y_L$ measurement does not determine a complete surface code syndrome, so this part is not decoded.

            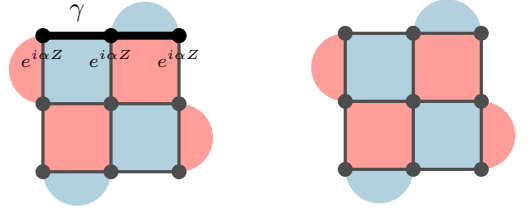
\begin{figure}[tbp]
                \centering
                \begin{minipage}[c]{0.45\columnwidth}
                    \centering
                    \begin{tikzpicture}[scale=0.9]
                        \definecolor{pastel_blue}{RGB}{143, 190, 207}
                        \definecolor{pastel_red}{RGB}{255, 117, 110}
                        \fill[pastel_blue,opacity=0.7] (0,1) rectangle (1,2);
                        \fill[pastel_blue,opacity=0.7] (1,0) rectangle (2,1);
                        \fill[pastel_red,opacity=0.7] (1,1) rectangle (2,2);
                        \fill[pastel_red,opacity=0.7] (0,0) rectangle (1,1);

                        \fill[pastel_blue,opacity=0.7]
                        (1,2) arc[start angle=180,end angle=0,radius=.5] -- cycle;
                        \fill[pastel_blue,opacity=0.7]
                        (1,0) arc[start angle=0,end angle=-180,radius=.5] -- cycle;
                        \fill[pastel_red,opacity=0.7]
                        (0,1) arc[start angle=270,end angle=90,radius=.5] -- cycle;
                        \fill[pastel_red,opacity=0.7]
                        (2,0) arc[start angle=-90,end angle=90,radius=.5] -- cycle;

                        \draw[line width=1.2pt, black!70!white] (0,0) grid (2,2);

                        \foreach \x in {0,1,2}
                        {
                          \foreach \y in {0,1,2}
                          {
                            \fill[black!70!white] (\x,\y) circle (0.11);
                          }
                        }
                        \draw[line width=3.0pt] (0,2) -- (2,2);
                        \foreach \x in {0,1,2}
                        {
                          \fill (\x,2) circle (0.11);
                          \node[below=2pt,font=\scriptsize] at (\x,2) {$e^{i\alpha Z}$};
                        }
                        \node[above=2pt,font=\normalsize] at (0.5,2) {$\gamma$};
                    \end{tikzpicture}
                \end{minipage}
                \begin{minipage}[c]{0.45\columnwidth}
                    \centering
                    \begin{tikzpicture}[scale=0.9]
                        \definecolor{pastel_blue}{RGB}{143, 190, 207}
                        \definecolor{pastel_red}{RGB}{255, 117, 110}
                        \fill[pastel_blue,opacity=0.7] (0,1) rectangle (1,2);
                        \fill[pastel_blue,opacity=0.7] (1,0) rectangle (2,1);
                        \fill[pastel_red,opacity=0.7] (1,1) rectangle (2,2);
                        \fill[pastel_red,opacity=0.7] (0,0) rectangle (1,1);

                        \fill[pastel_blue,opacity=0.7]
                        (1,2) arc[start angle=180,end angle=0,radius=.5] -- cycle;
                        \fill[pastel_blue,opacity=0.7]
                        (1,0) arc[start angle=0,end angle=-180,radius=.5] -- cycle;
                        \fill[pastel_red,opacity=0.7]
                        (0,1) arc[start angle=270,end angle=90,radius=.5] -- cycle;
                        \fill[pastel_red,opacity=0.7]
                        (2,0) arc[start angle=-90,end angle=90,radius=.5] -- cycle;

                        \draw[line width=1.2pt, black!70!white] (0,0) grid (2,2);

                        \foreach \x in {0,1,2}
                        {
                          \foreach \y in {0,1,2}
                          {
                            \fill[black!70!white] (\x,\y) circle (0.11);
                          }
                        }
                    \end{tikzpicture}
                \end{minipage}
                \caption{Two $d=3$ rotated surface code patches.
                The rotations $e^{i\alpha Z}$ act on the three data qubits of the upper horizontal logical $Z_L$ path $\gamma$ in each patch.
                The path is marked only in the left patch.}
                \label{fig:two-surface-code-patches}
            \end{figure}

            \begin{figure*}[t!]
                \centering
                \begin{subfigure}[t]{0.58\textwidth}
                    \centering
                    \scalebox{0.45}{\begin{quantikz}[row sep=0.015cm, column sep=0.055cm]
                        \lstick{$\ket{0}_{v_1}^{(1)}$} & \gate{H}\gategroup[wires=18,steps=9,style={fill=blue!8,draw=blue!25,rounded corners},background]{} & \gate{H} & \qw & \qw & \qw & \qw & \ctrl{3} & \qw & \qw & \phantomgate{\qquad} & \qw & \qw & \qw & \gate{R_Z(\alpha)}\gategroup[wires=18,steps=2,style={fill=green!10,draw=green!30!black,rounded corners},background]{} & \gate[9]{\Pi_L^{(1)}} & \qw & \qw & \ctrl{9}\gategroup[wires=18,steps=9,style={fill=red!7,draw=red!25,rounded corners},background]{} & \qw & \qw & \qw & \qw & \qw & \qw & \qw & \qw & \qw & \qw & \phantomgate{\qquad} & \meter{$X$}\gategroup[wires=9,steps=1,style={fill=orange!8,draw=orange!40!black,rounded corners},background]{} \\
                        \lstick{$\ket{0}_{v_2}^{(1)}$} & \gate{H} & \qw & \targ{} & \ctrl{3} & \qw & \qw & \qw & \qw & \qw & \qw & \qw & \qw & \qw & \gate{R_Z(\alpha)} & \qw & \qw & \qw & \qw & \ctrl{9} & \qw & \qw & \qw & \qw & \qw & \qw & \qw & \qw & \qw & \qw & \meter{$X$} \\
                        \lstick{$\ket{0}_{v_3}^{(1)}$} & \gate{H} & \gate{H} & \ctrl{-1} & \qw & \qw & \qw & \qw & \ctrl{3} & \qw & \qw & \qw & \qw & \qw & \gate{R_Z(\alpha)} & \qw & \qw & \qw & \qw & \qw & \ctrl{9} & \qw & \qw & \qw & \qw & \qw & \qw & \qw & \qw & \qw & \meter{$X$} \\
                        \lstick{$\ket{0}_{u_1}^{(1)}$} & \gate{H} & \qw & \qw & \qw & \targ{} & \qw & \targ{} & \qw & \qw & \qw & \qw & \qw & \qw & \qw & \qw & \qw & \qw & \qw & \qw & \qw & \ctrl{9} & \qw & \qw & \qw & \qw & \qw & \qw & \qw & \qw & \meter{$X$} \\
                        \lstick{$\ket{0}_{u_2}^{(1)}$} & \gate{H} & \qw & \qw & \targ{} & \qw & \qw & \qw & \qw & \targ{} & \qw & \qw & \qw & \qw & \qw & \qw & \qw & \qw & \qw & \qw & \qw & \qw & \ctrl{9} & \qw & \qw & \qw & \qw & \qw & \qw & \qw & \meter{$X$} \\
                        \lstick{$\ket{0}_{u_3}^{(1)}$} & \gate{H} & \qw & \qw & \qw & \qw & \targ{} & \qw & \targ{} & \qw & \qw & \qw & \qw & \qw & \qw & \qw & \qw & \qw & \qw & \qw & \qw & \qw & \qw & \ctrl{9} & \qw & \qw & \qw & \qw & \qw & \qw & \meter{$X$} \\
                        \lstick{$\ket{0}_{u_4}^{(1)}$} & \gate{H} & \qw & \targ{} & \qw & \ctrl{-3} & \qw & \qw & \qw & \qw & \qw & \qw & \qw & \qw & \qw & \qw & \qw & \qw & \qw & \qw & \qw & \qw & \qw & \qw & \ctrl{9} & \qw & \qw & \qw & \qw & \qw & \meter{$X$} \\
                        \lstick{$\ket{0}_{u_5}^{(1)}$} & \gate{H} & \gate{H} & \ctrl{-1} & \qw & \qw & \qw & \qw & \qw & \ctrl{-3} & \qw & \qw & \qw & \qw & \qw & \qw & \qw & \qw & \qw & \qw & \qw & \qw & \qw & \qw & \qw & \ctrl{9} & \qw & \qw & \qw & \qw & \meter{$X$} \\
                        \lstick{$\ket{0}_{u_6}^{(1)}$} & \gate{H} & \gate{H} & \qw & \qw & \qw & \ctrl{-3} & \qw & \qw & \qw & \qw & \qw & \qw & \qw & \qw & \qw & \qw & \qw & \qw & \qw & \qw & \qw & \qw & \qw & \qw & \qw & \ctrl{9} & \qw & \qw & \qw & \meter{$X$} \\
                        \lstick{$\ket{0}_{v_1}^{(2)}$} & \gate{H} & \gate{H} & \qw & \qw & \qw & \qw & \ctrl{3} & \qw & \qw & \phantomgate{\qquad} & \qw & \qw & \qw & \gate{R_Z(\alpha)} & \gate[9]{\Pi_L^{(2)}} & \qw & \qw & \targ{} & \qw & \qw & \qw & \qw & \qw & \qw & \qw & \qw & \phantomgate{\qquad} & \meter{$Z$}\gategroup[wires=9,steps=1,style={fill=yellow!14,draw=orange!55!black,rounded corners},background]{} \\
                        \lstick{$\ket{0}_{v_2}^{(2)}$} & \gate{H} & \qw & \targ{} & \ctrl{3} & \qw & \qw & \qw & \qw & \qw & \qw & \qw & \qw & \qw & \gate{R_Z(\alpha)} & \qw & \qw & \qw & \qw & \targ{} & \qw & \qw & \qw & \qw & \qw & \qw & \qw & \qw & \meter{$Z$} \\
                        \lstick{$\ket{0}_{v_3}^{(2)}$} & \gate{H} & \gate{H} & \ctrl{-1} & \qw & \qw & \qw & \qw & \ctrl{3} & \qw & \qw & \qw & \qw & \qw & \gate{R_Z(\alpha)} & \qw & \qw & \qw & \qw & \qw & \targ{} & \qw & \qw & \qw & \qw & \qw & \qw & \qw & \meter{$Z$} \\
                        \lstick{$\ket{0}_{u_1}^{(2)}$} & \gate{H} & \qw & \qw & \qw & \targ{} & \qw & \targ{} & \qw & \qw & \qw & \qw & \qw & \qw & \qw & \qw & \qw & \qw & \qw & \qw & \qw & \targ{} & \qw & \qw & \qw & \qw & \qw & \qw & \meter{$Z$} \\
                        \lstick{$\ket{0}_{u_2}^{(2)}$} & \gate{H} & \qw & \qw & \targ{} & \qw & \qw & \qw & \qw & \targ{} & \qw & \qw & \qw & \qw & \qw & \qw & \qw & \qw & \qw & \qw & \qw & \qw & \targ{} & \qw & \qw & \qw & \qw & \qw & \meter{$Z$} \\
                        \lstick{$\ket{0}_{u_3}^{(2)}$} & \gate{H} & \qw & \qw & \qw & \qw & \targ{} & \qw & \targ{} & \qw & \qw & \qw & \qw & \qw & \qw & \qw & \qw & \qw & \qw & \qw & \qw & \qw & \qw & \targ{} & \qw & \qw & \qw & \qw & \meter{$Z$} \\
                        \lstick{$\ket{0}_{u_4}^{(2)}$} & \gate{H} & \qw & \targ{} & \qw & \ctrl{-3} & \qw & \qw & \qw & \qw & \qw & \qw & \qw & \qw & \qw & \qw & \qw & \qw & \qw & \qw & \qw & \qw & \qw & \qw & \targ{} & \qw & \qw & \qw & \meter{$Z$} \\
                        \lstick{$\ket{0}_{u_5}^{(2)}$} & \gate{H} & \gate{H} & \ctrl{-1} & \qw & \qw & \qw & \qw & \qw & \ctrl{-3} & \qw & \qw & \qw & \qw & \qw & \qw & \qw & \qw & \qw & \qw & \qw & \qw & \qw & \qw & \qw & \targ{} & \qw & \qw & \meter{$Z$} \\
                        \lstick{$\ket{0}_{u_6}^{(2)}$} & \gate{H} & \gate{H} & \qw & \qw & \qw & \ctrl{-3} & \qw & \qw & \qw & \qw & \qw & \qw & \qw & \qw & \qw & \qw & \qw & \qw & \qw & \qw & \qw & \qw & \qw & \qw & \qw & \targ{} & \qw & \meter{$Z$} \\
                        \lstick{$\ket{0}_{a}^{(1)}$} & \qw & \qw & \qw & \qw & \qw & \qw & \qw & \qw & \qw & \qw & \qw & \qw & \qw & \qw & \gate{(\mathcal{M}_{\rm{stab}}^{(1)})^{\times3}} & \meter{$Z$} & \gate{\ket{0}} & \qw & \qw & \qw & \qw & \qw & \qw & \qw & \qw & \qw & \qw & \qw & \qw & \qw \\
                        \lstick{$\ket{0}_{a}^{(2)}$} & \qw & \qw & \qw & \qw & \qw & \qw & \qw & \qw & \qw & \qw & \qw & \qw & \qw & \qw & \gate{(\mathcal{M}_{\rm{stab}}^{(2)})^{\times3}} & \meter{$Z$} & \gate{\ket{0}} & \qw & \qw & \qw & \qw & \qw & \qw & \qw & \qw & \qw & \qw & \qw & \qw & \qw
                    \end{quantikz}
                    }
                    \caption{Logical Bell measurement with a transversal $\CNOT$ for two $d=3$ patches.}
                    \label{fig:logical-physical-circuit}
                \end{subfigure}
                \hfill
                \begin{subfigure}[t]{0.38\textwidth}
                    \centering
                    \scalebox{0.65}{\begin{quantikz}[row sep=0.018cm, column sep=0.10cm]
                        \lstick{$\ket{0}_{v_1}$} & \gate{H}\gategroup[wires=9,steps=9,style={fill=blue!8,draw=blue!25,rounded corners},background]{} & \gate{H} & \qw & \qw & \qw & \qw & \ctrl{3} & \qw & \qw & \phantomgate{\qquad} & \gate{R_Z(\alpha)}\gategroup[wires=9,steps=2,style={fill=green!10,draw=green!30!black,rounded corners},background]{} & \gate[9]{\Pi_L} & \qw \\
                        \lstick{$\ket{0}_{v_2}$} & \gate{H} & \qw & \targ{} & \ctrl{3} & \qw & \qw & \qw & \qw & \qw & \qw & \gate{R_Z(\alpha)} & \qw & \qw \\
                        \lstick{$\ket{0}_{v_3}$} & \gate{H} & \gate{H} & \ctrl{-1} & \qw & \qw & \qw & \qw & \ctrl{3} & \qw & \qw & \gate{R_Z(\alpha)} & \qw & \qw \\
                        \lstick{$\ket{0}_{u_1}$} & \gate{H} & \qw & \qw & \qw & \targ{} & \qw & \targ{} & \qw & \qw & \qw & \qw & \qw & \qw \\
                        \lstick{$\ket{0}_{u_2}$} & \gate{H} & \qw & \qw & \targ{} & \qw & \qw & \qw & \qw & \targ{} & \qw & \qw & \qw & \qw \\
                        \lstick{$\ket{0}_{u_3}$} & \gate{H} & \qw & \qw & \qw & \qw & \targ{} & \qw & \targ{} & \qw & \qw & \qw & \qw & \qw \\
                        \lstick{$\ket{0}_{u_4}$} & \gate{H} & \qw & \targ{} & \qw & \ctrl{-3} & \qw & \qw & \qw & \qw & \qw & \qw & \qw & \qw \\
                        \lstick{$\ket{0}_{u_5}$} & \gate{H} & \gate{H} & \ctrl{-1} & \qw & \qw & \qw & \qw & \qw & \ctrl{-3} & \qw & \qw & \qw & \qw \\
                        \lstick{$\ket{0}_{u_6}$} & \gate{H} & \gate{H} & \qw & \qw & \qw & \ctrl{-3} & \qw & \qw & \qw & \qw & \qw & \qw & \qw
                    \end{quantikz}
                    }
                    \begin{align*}
                        \begin{array}{c|ccccc}
                            \text{logical readout} & v_1 & v_2 & v_3 & u_1 & u_4 \\
                            \hline
                            X_L & X & - & - & X & X \\
                            Y_L & Y & Z & Z & X & X \\
                            Z_L & Z & Z & Z & - & -
                        \end{array}
                    \end{align*}
                    \caption{Logical Pauli readout on one patch for $|c|^2$.}
                    \label{fig:logical-coherence-circuit}
                \end{subfigure}
                \caption{Circuit implementations of the logical Bell measurement on two $d=3$ rotated surface code patches and the coherence measurement on one such patch.
                In panel (a), the wires $v_1,v_2,v_3,u_1,\ldots,u_6$ are the nine data qubits of each patch, and $\gamma=\{v_1,v_2,v_3\}$ is the logical $Z_L$ path shown in Fig.\,\ref{fig:two-surface-code-patches}.
                The ancilla $a^{(j)}$ performs three complete rounds of the stabilizer measurements $\mathcal{M}_{\rm{stab}}^{(j)}$ sequentially and is reset after every stabilizer readout.
                The transversal $\CNOT$ is followed by destructive physical $Z$ readout and classical decoding on the second patch.
                Trials with decoded outcome $Z_L^{(2)}=-1$ proceed to destructive physical $X$ readout and classical decoding on the first patch, whereas outcome $+1$ terminates the trial before this readout.
                In panel (b), the same single-patch preparation is followed by the postselected projection $\Pi_L$ and one terminal readout setting, $X_L$, $Y_L$, or $Z_L$, whose measured physical support is listed below the circuit.}
                \label{fig:logical-circuits}
            \end{figure*}
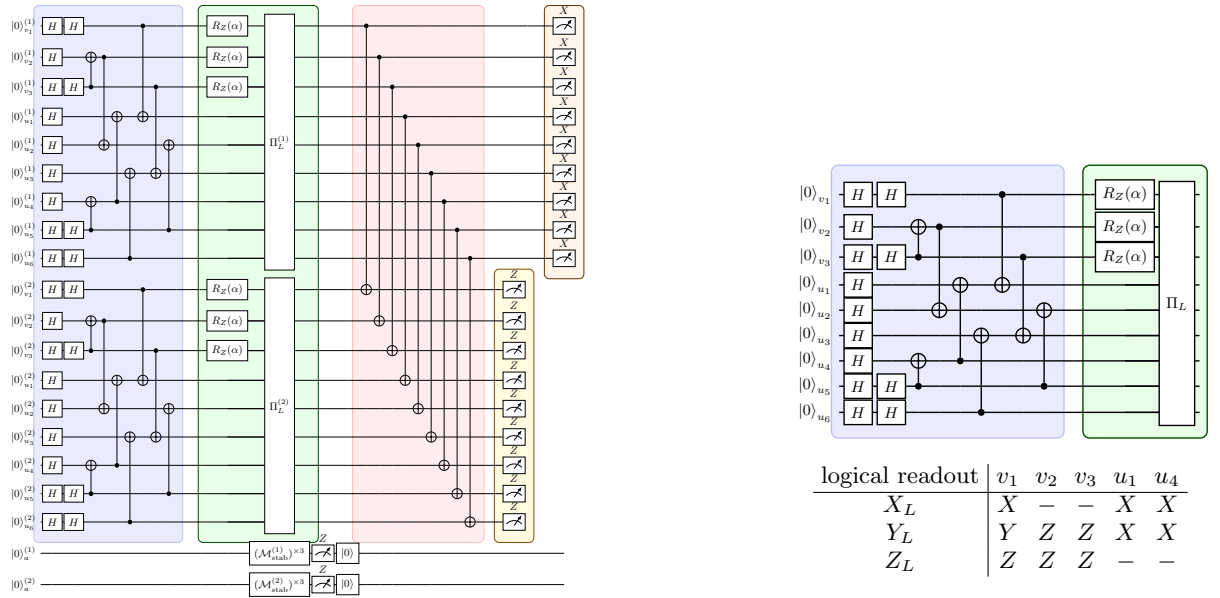

            Before performing the numerical simulations, let us evaluate the contribution of the dephasing channel \eqref{eq:dephasing-channel} on the logical path $\gamma$ from first principles.
            For a subset of $k$ ($0\leqq k\leqq d$) faulty qubits on $\gamma$, two terms from Eq.\,\eqref{eq:physical_rotation} pass the stabilizer projection onto the code subspace, namely,
            \begin{align}
                \begin{aligned}
                    \ket*{\phi_k}
                    &=
                    i^k\cos^{d-k}\!\alpha\,\sin^k\!\alpha\ket*{+_L}
                    \\
                    &\qquad+
                    i^{d-k}\cos^k\!\alpha\,\sin^{d-k}\!\alpha\ket*{-_L}.
                \end{aligned}
            \end{align}
            All subsets of the same size $k$ produce the same unnormalized logical state $\ket*{\phi_k}$.
            For $N$ attempts to prepare one resource state, the number of accepted trials with no path error is
            \begin{align}
                n_0
                =
                N(1-p_Z)^d
                \left(
                    \cos^{2d}\!\alpha+\sin^{2d}\!\alpha
                \right).
            \end{align}
            For $1\leqq k\leqq d$, the corresponding count for exactly $k$ errors on $\gamma$ is
            \begin{align}
                \begin{aligned}
                    n_k
                    &=
                    N\binom{d}{k}p_Z^k(1-p_Z)^{d-k}
                    \\
                    &\quad\times
                    \left(
                        \cos^{2(d-k)}\!\alpha\,\sin^{2k}\!\alpha
                        +\cos^{2k}\!\alpha\,\sin^{2(d-k)}\!\alpha
                    \right).
                \end{aligned}
            \end{align}

            The normalized single-copy state conditioned on acceptance is
            \begin{align}
                \rho
                =
                \frac{
                    \displaystyle\sum_{k=0}^{d}
                    \binom{d}{k}p_Z^k(1-p_Z)^{d-k}
                    \ket*{\phi_k}\!\bra*{\phi_k}
                }{
                    \displaystyle\sum_{k=0}^{d}
                    \binom{d}{k}p_Z^k(1-p_Z)^{d-k}
                    \braket*{\phi_k}{\phi_k}
                }.
            \end{align}
            Its off-diagonal element is proportional to
            \begin{align}
                \begin{aligned}
                    &\sum_{k=0}^d
                    \binom{d}{k}
                    p_Z^k(1-p_Z)^{d-k}
                    \langle+_L|\phi_k\rangle
                    \langle\phi_k|-_L\rangle
                    \\
                    &\qquad=
                    (-1)^d i^d
                    (\cos\alpha\sin\alpha)^d
                    \sum_{k=0}^d
                    \binom{d}{k}
                    (-p_Z)^k(1-p_Z)^{d-k}
                    \\
                    &\qquad=
                    (-1)^d i^d
                    (\cos\alpha\sin\alpha)^d
                    (1-2p_Z)^d.
                \end{aligned}
            \end{align}
            For two independently prepared copies, the Bell part outcome is $b=(1-\operatorname{tr}\rho^2)/2$, while the post-selection on $Z_2=-1$ succeeds with probability $1/2$.
            The exact conditional probability is therefore
            \begin{align}
                \begin{aligned}
                    &\operatorname{Pr}(X_1=-1\mid Z_2=-1)
                    \\
                    &=
                    2
                    \frac{
                        \begin{aligned}
                            &\Big[
                                \big((1-p_Z)\cos^2\!\alpha+p_Z\sin^2\!\alpha\big)
                                \\
                                &\qquad\times
                                \big((1-p_Z)\sin^2\!\alpha+p_Z\cos^2\!\alpha\big)
                            \Big]^d
                            \\
                            &\quad-
                            (\cos^2\!\alpha\sin^2\!\alpha)^d
                            (1-2p_Z)^{2d}
                        \end{aligned}
                    }{
                        \begin{aligned}
                            \Big[
                                &\big((1-p_Z)\cos^2\!\alpha+p_Z\sin^2\!\alpha\big)^d
                                \\
                                &+
                                \big((1-p_Z)\sin^2\!\alpha+p_Z\cos^2\!\alpha\big)^d
                            \Big]^2
                        \end{aligned}
                    }.
                \end{aligned}
                \label{eq:explicit_pr}
            \end{align}
            At first order, only one path error on either copy contributes, and Eq.\,\eqref{eq:explicit_pr} reduces to
            \begin{align}
                \begin{aligned}
                    \operatorname{Pr}(X_1=-1\mid Z_2=-1)
                    &=
                    \frac{
                        2d p_Z
                        (\cos^2\!\alpha\sin^2\!\alpha)^{d-1}
                    }{
                        \left(
                            \cos^{2d}\!\alpha+\sin^{2d}\!\alpha
                        \right)^2
                    }
                    \\
                    &\quad+
                    O(p_Z^2).
                \end{aligned}
                \label{eq:first_order_limit}
            \end{align}
            Since this probability is proportional to the physical error rate, it retains the $O(1/\epsilon)$ sample complexity found at the physical level in Subsec.\,\ref{subsec:physical_level}.

            Below are several observations of the simulation in Figs.\,\ref{fig:bell-benchmark-on-qec},\,\ref{fig:logical-realistic-shot-benchmark}.
            Figures\,\ref{fig:bell-benchmark-on-qec} and \ref{fig:logical-realistic-shot-benchmark} use $d=3$, $\theta/\pi\in\{0.01,0.05,0.1\}$, and $p_{\rm{prep}}=0.005$, with an independent $Z$ fault of probability $p_Z=p_{\rm{prep}}/3$ after each physical rotation.
            Rows (I)--(III) use $p_{\rm{op}}=p_{\rm{proj}}=0.001$, $0.003$, and $0.005$, respectively.
            Here $p_{\rm{op}}$ is the two-qubit depolarizing probability after each physical CNOT in the transversal logical $\CNOT$ of the Bell part, whereas $p_{\rm{proj}}$ applies after each CNOT in the stabilizer measurements of every branch.
            The larger simulation (Fig.\,\ref{fig:bell-benchmark-on-qec}) uses $N_B=400{,}000$ raw Bell attempts and $N_P=100{,}000$ raw attempts in each Pauli basis for the coherence part, while the smaller simulation (Fig.\,\ref{fig:logical-realistic-shot-benchmark}) uses $N_B=40{,}000$ and $N_P=10{,}000$.
            These counts are taken before the post-selection by the stabilizer measurement.
            Only attempts that pass the post-selection enter the estimators.
            In the Bell part, the subsequent physical $X$ measurement is performed only on the smaller subset with decoded $Z_L^{(2)}=-1$, while the denominator of the singlet probability estimator $\widehat b$ still contains every pair accepted by the syndrome post-selection.
            The BCFE error bars are computed using Eq.\,\eqref{eq:bcfe-standard-error}, with $N_B$ and $N_P$ set to the realized counts after syndrome post-selection and the error calculation conditioned on these counts.
            Since a Bell attempt consumes two resource states and each Pauli attempt consumes one, the total raw resource state budget is
            \begin{align}
                N=2N_B+3N_P.
            \end{align}
            The two figures therefore use $N=1{,}100{,}000$ and $110{,}000$, respectively, and DFE is assigned the same value of $N$ in each comparison.

            \begin{figure*}[t]
                \centering

                \begin{subfigure}{0.32\textwidth}
                    \centering
                    \includegraphics[width=\linewidth,trim=0 0 0 25pt,clip]{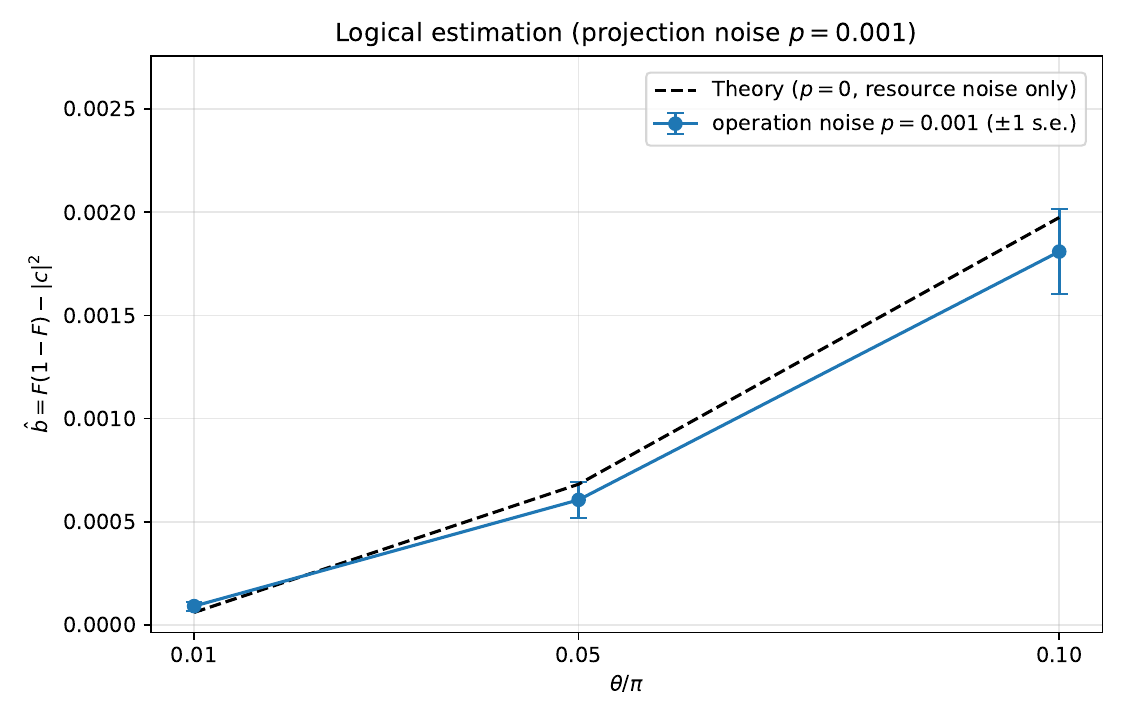}
                    \caption*{(I-a) Bell term}
                \end{subfigure}
                \begin{subfigure}{0.32\textwidth}
                    \centering
                    \includegraphics[width=\linewidth,trim=0 0 0 25pt,clip]{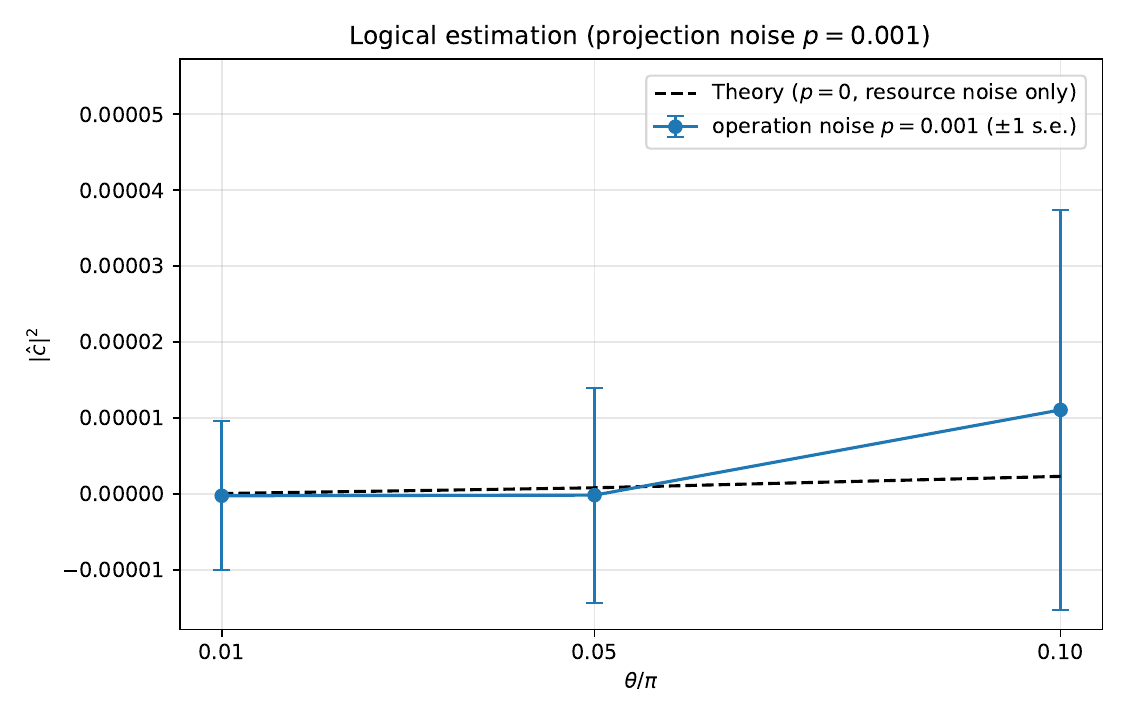}
                    \caption*{(I-b) Coherence term}
                \end{subfigure}
                \begin{subfigure}{0.32\textwidth}
                    \centering
                    \includegraphics[width=\linewidth,trim=0 0 0 25pt,clip]{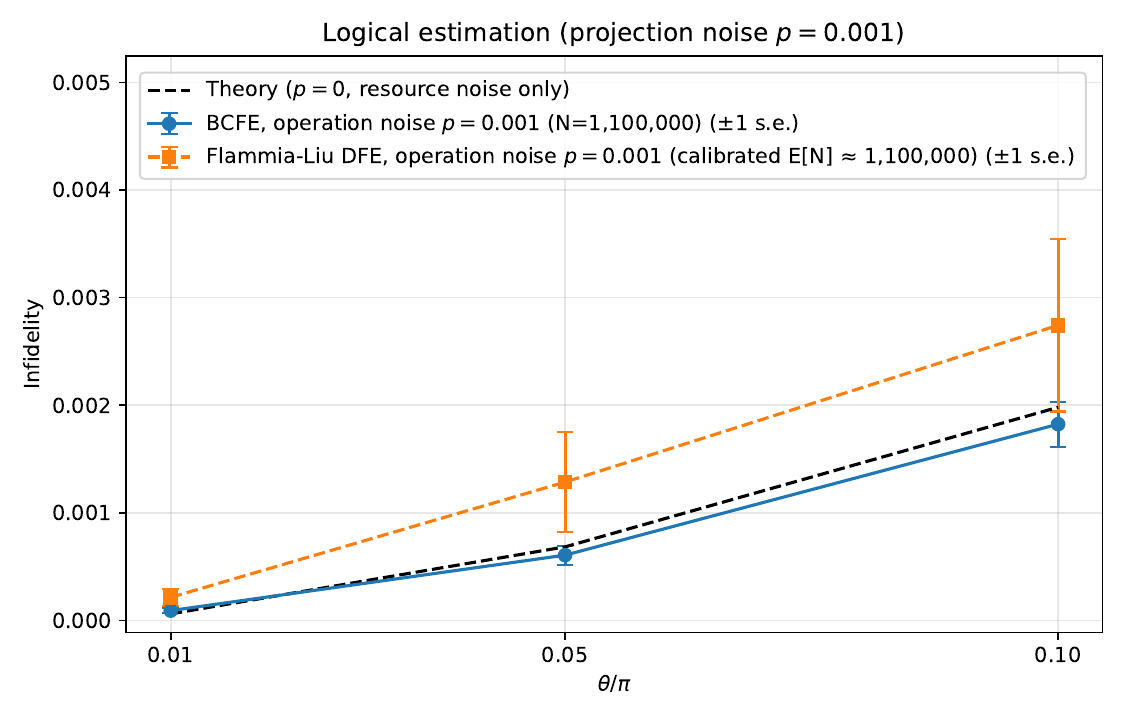}
                    \caption*{(I-c) Infidelity}
                \end{subfigure}

                \vspace{0.4em}

                \begin{subfigure}{0.32\textwidth}
                    \centering
                    \includegraphics[width=\linewidth,trim=0 0 0 25pt,clip]{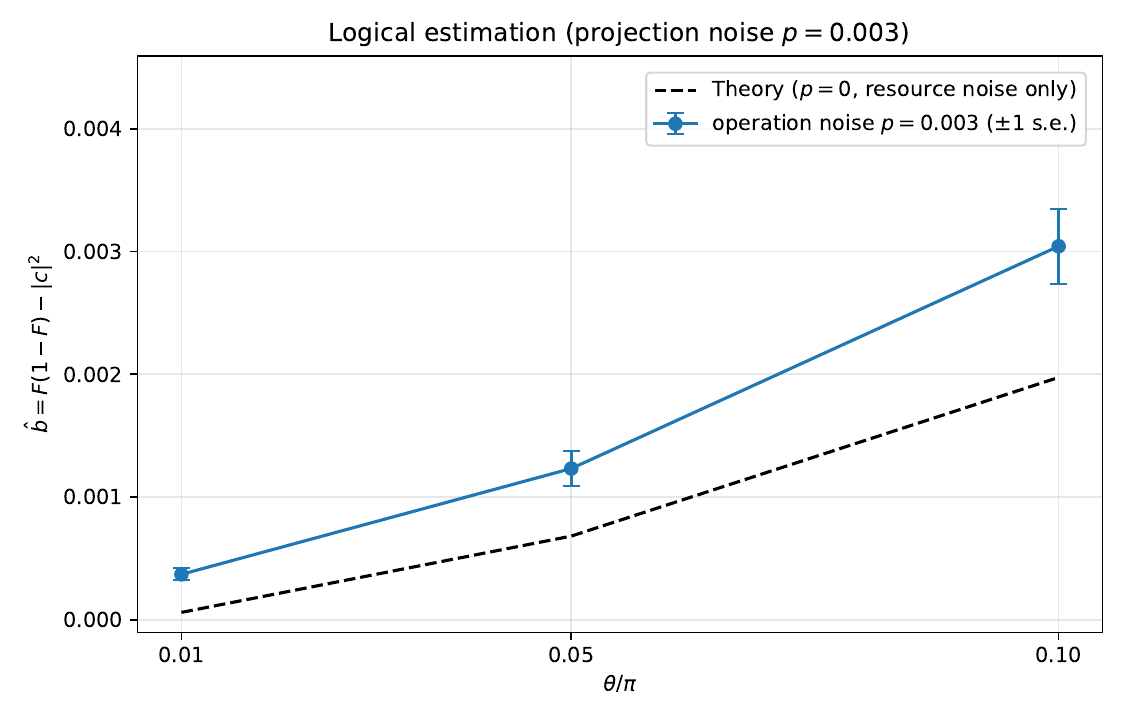}
                    \caption*{(II-a) Bell term}
                \end{subfigure}
                \begin{subfigure}{0.32\textwidth}
                    \centering
                    \includegraphics[width=\linewidth,trim=0 0 0 25pt,clip]{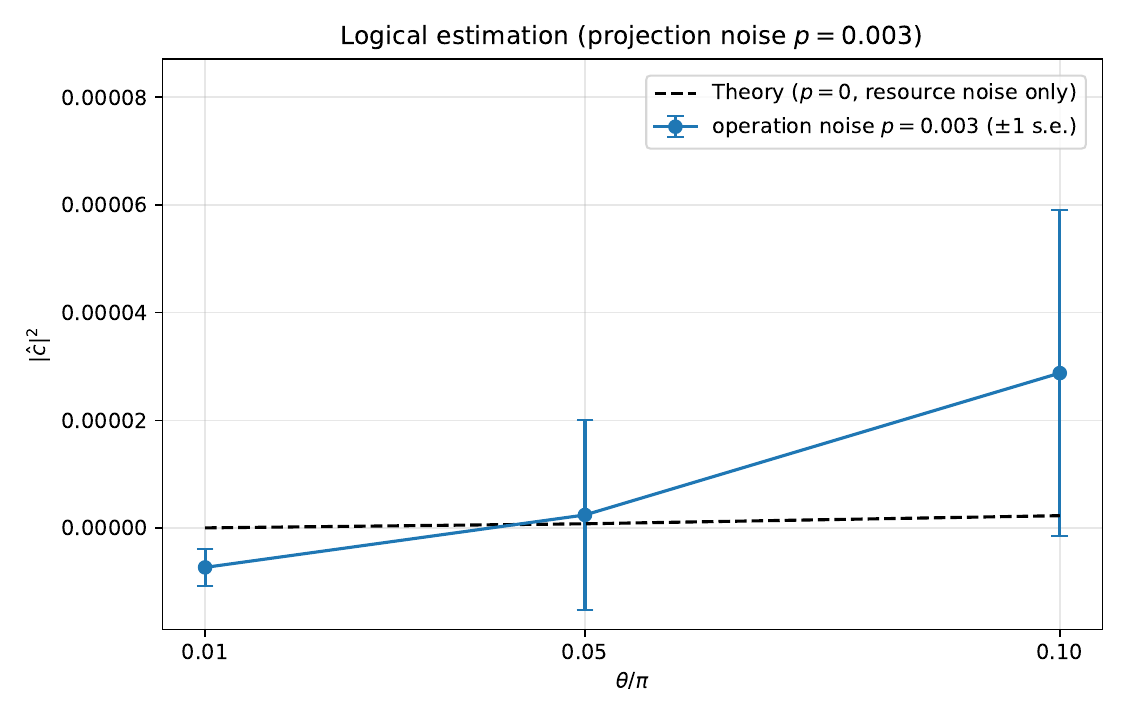}
                    \caption*{(II-b) Coherence term}
                \end{subfigure}
                \begin{subfigure}{0.32\textwidth}
                    \centering
                    \includegraphics[width=\linewidth,trim=0 0 0 25pt,clip]{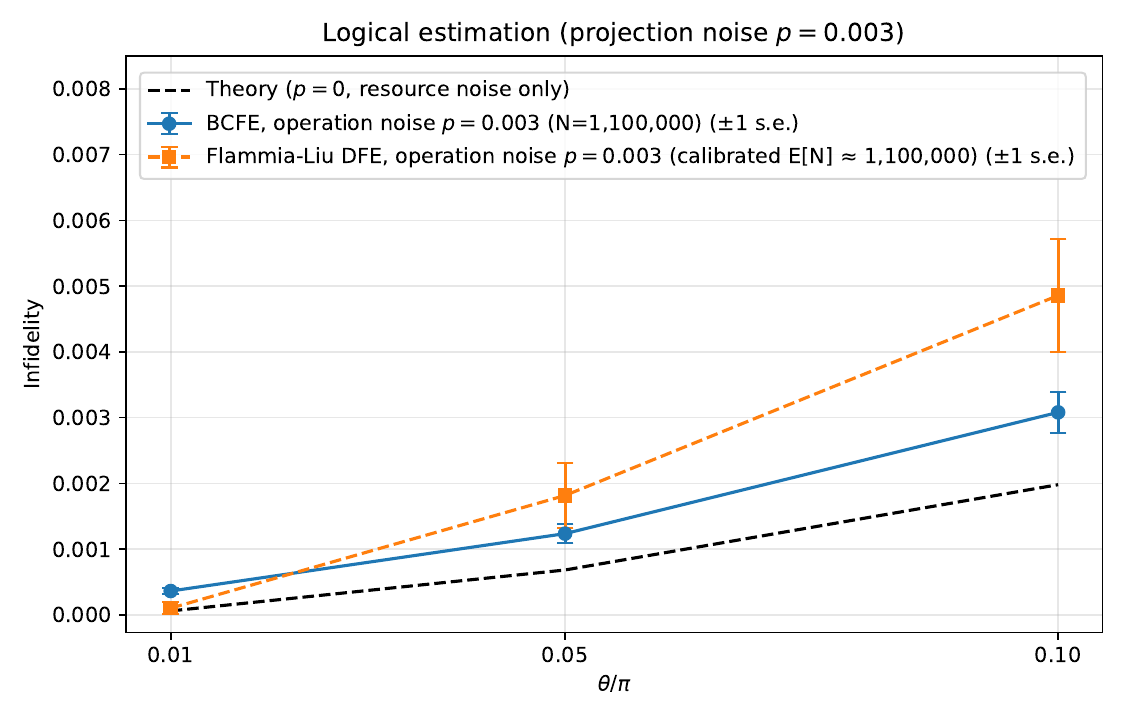}
                    \caption*{(II-c) Infidelity}
                \end{subfigure}

                \vspace{0.4em}

                \begin{subfigure}{0.32\textwidth}
                    \centering
                    \includegraphics[width=\linewidth,trim=0 0 0 25pt,clip]{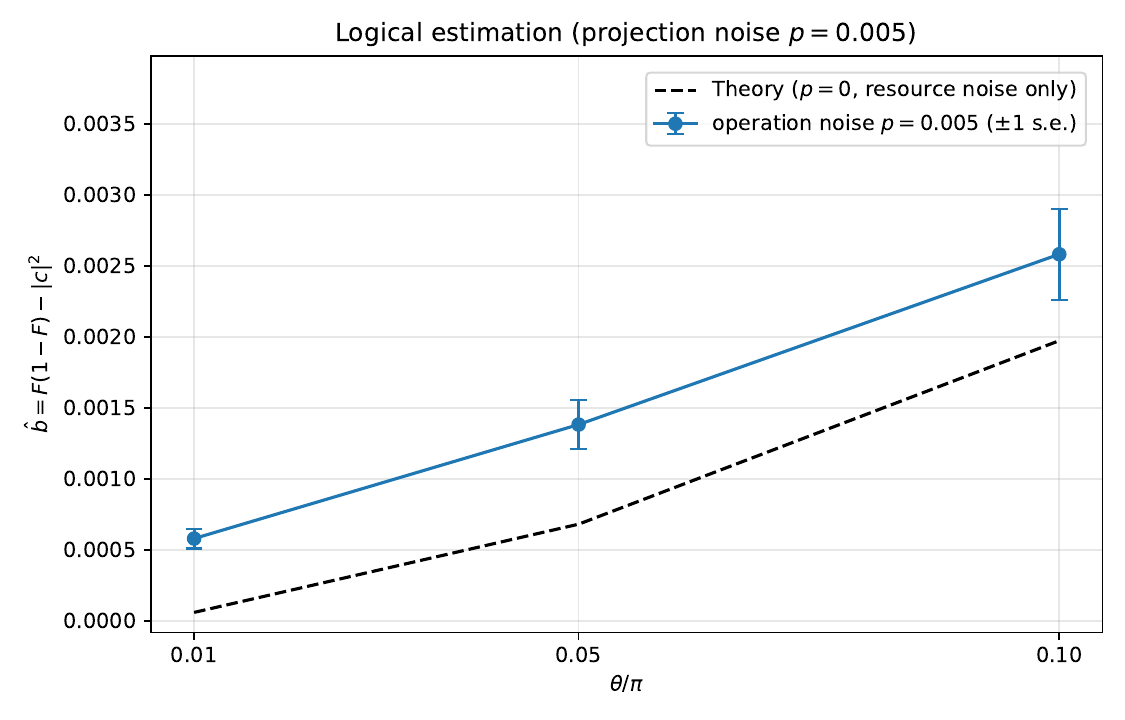}
                    \caption*{(III-a) Bell term}
                \end{subfigure}
                \begin{subfigure}{0.32\textwidth}
                    \centering
                    \includegraphics[width=\linewidth,trim=0 0 0 25pt,clip]{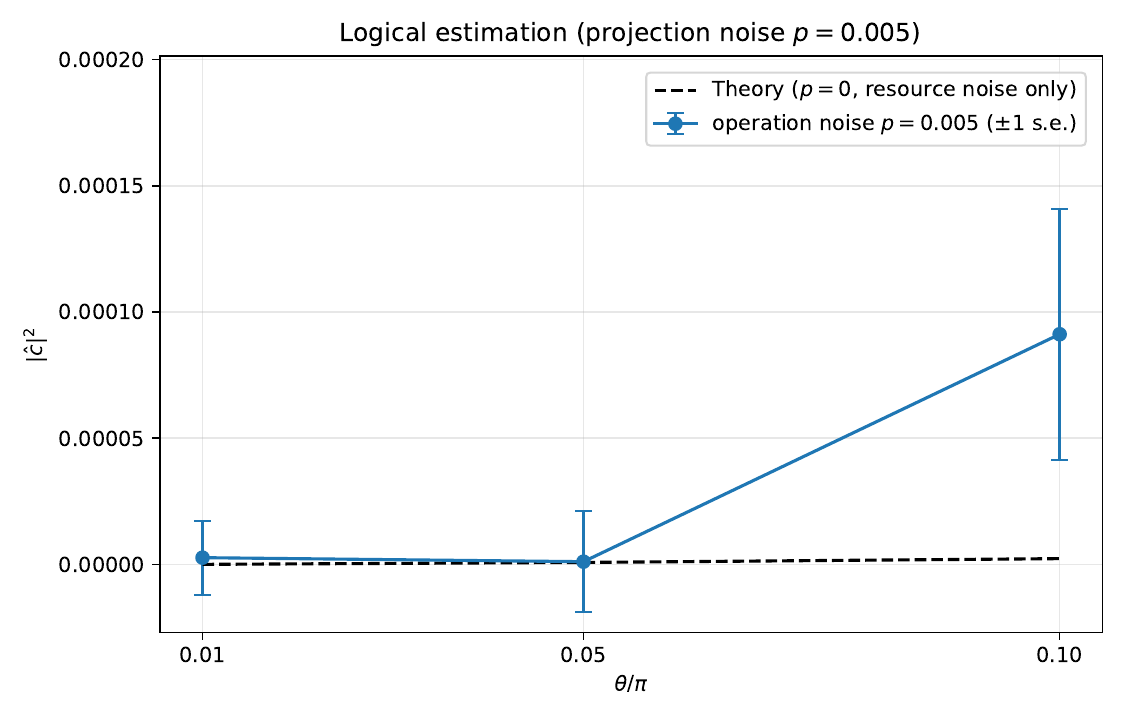}
                    \caption*{(III-b) Coherence term}
                \end{subfigure}
                \begin{subfigure}{0.32\textwidth}
                    \centering
                    \includegraphics[width=\linewidth,trim=0 0 0 25pt,clip]{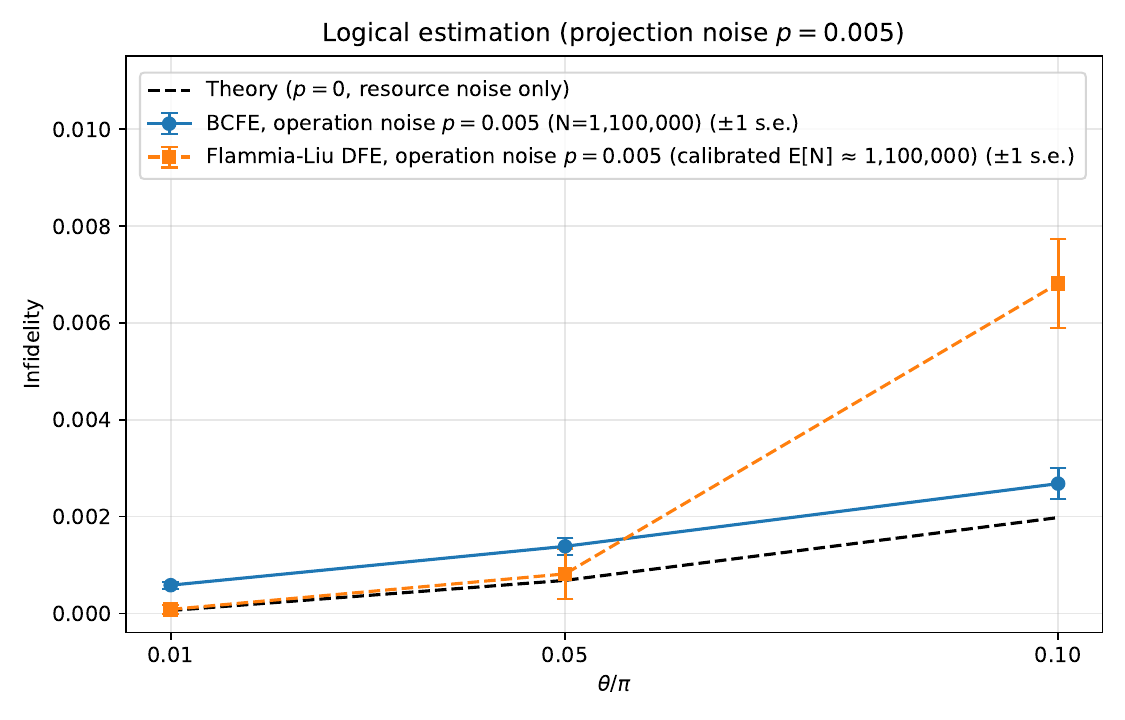}
                    \caption*{(III-c) Infidelity}
                \end{subfigure}

                \caption{Logical level fidelity estimation with $1{,}100{,}000$ raw resource states, $d=3$ surface code patches, and preparation parameter $p_{\rm{prep}}=0.005$.
                Each rotated physical qubit is independently subject to a $Z$ fault with probability $p_Z=p_{\rm{prep}}/3$.
                Rows (I), (II), and (III) use $p_{\rm{op}}=p_{\rm{proj}}=0.001$, $0.003$, and $0.005$, respectively.
                Here $p_{\rm{op}}$ is the two-qubit depolarizing probability after each physical CNOT in the transversal logical $\CNOT$ of the Bell part, while $p_{\rm{proj}}$ is the corresponding probability after each CNOT in the three complete stabilizer measurement rounds used for the syndrome post-selection.
                The reference curve includes the resource preparation faults specified by $p_{\rm{prep}}$ and sets $p_{\rm{op}}=p_{\rm{proj}}=0$.
                Columns (a), (b), and (c) display the Bell term, the coherence term, and the resulting infidelity estimate.
                The Bell part applies a transversal CNOT, destructively measures and decodes $Z_L^{(2)}$, and performs the destructive $X_L^{(1)}$ measurement only when the decoded $Z_L^{(2)}$ outcome is $-1$.
                Each point uses $N_B=400{,}000$ raw Bell attempts and $N_P=100{,}000$ raw attempts in each Pauli basis.
                These counts precede the syndrome post-selection.
                Flammia--Liu DFE in column (c) uses the two-stage estimator with its expected raw consumption calibrated to $2N_B+3N_P=1{,}100{,}000$ using the acceptance rate observed in the BCFE single-copy experiments.
                Physical readout errors are absent in this comparison.
                Error bars indicate one standard error.
                BCFE error bars use the delta-method approximation in Eq.\,\eqref{eq:bcfe-standard-error}.
                Below the error-correction threshold, increasing the code distance is expected to suppress the contribution of residual logical errors in the measurement circuits to the deviation from the reference infidelity.}
                \label{fig:bell-benchmark-on-qec}
            \end{figure*}

            \begin{figure*}[t]
                \centering

                \begin{subfigure}{0.32\textwidth}
                    \centering
                    \includegraphics[width=\linewidth,trim=0 0 0 25pt,clip]{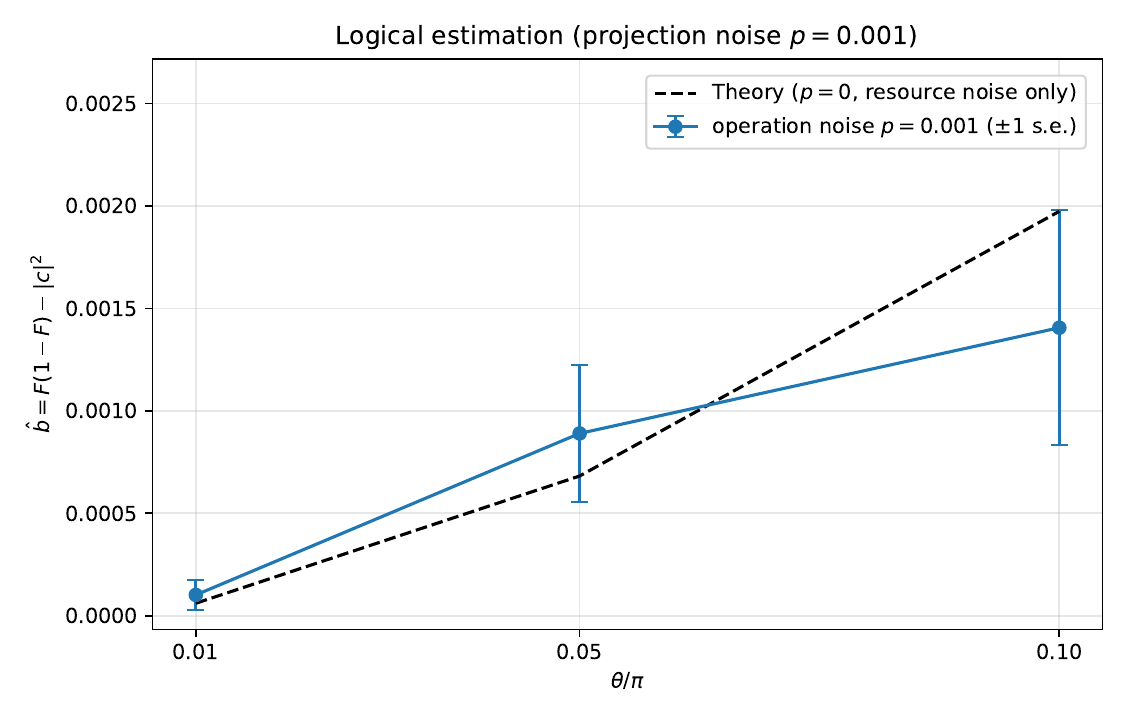}
                    \caption*{(I-a) Bell term}
                \end{subfigure}
                \begin{subfigure}{0.32\textwidth}
                    \centering
                    \includegraphics[width=\linewidth,trim=0 0 0 25pt,clip]{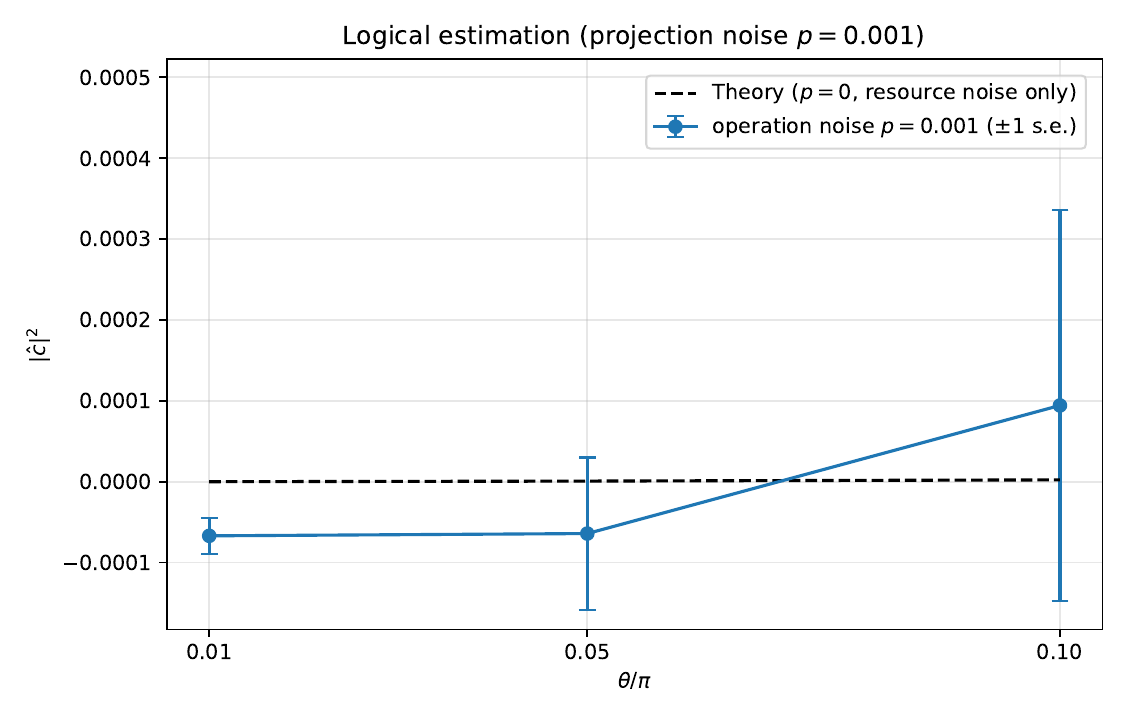}
                    \caption*{(I-b) Coherence term}
                \end{subfigure}
                \begin{subfigure}{0.32\textwidth}
                    \centering
                    \includegraphics[width=\linewidth,trim=0 0 0 25pt,clip]{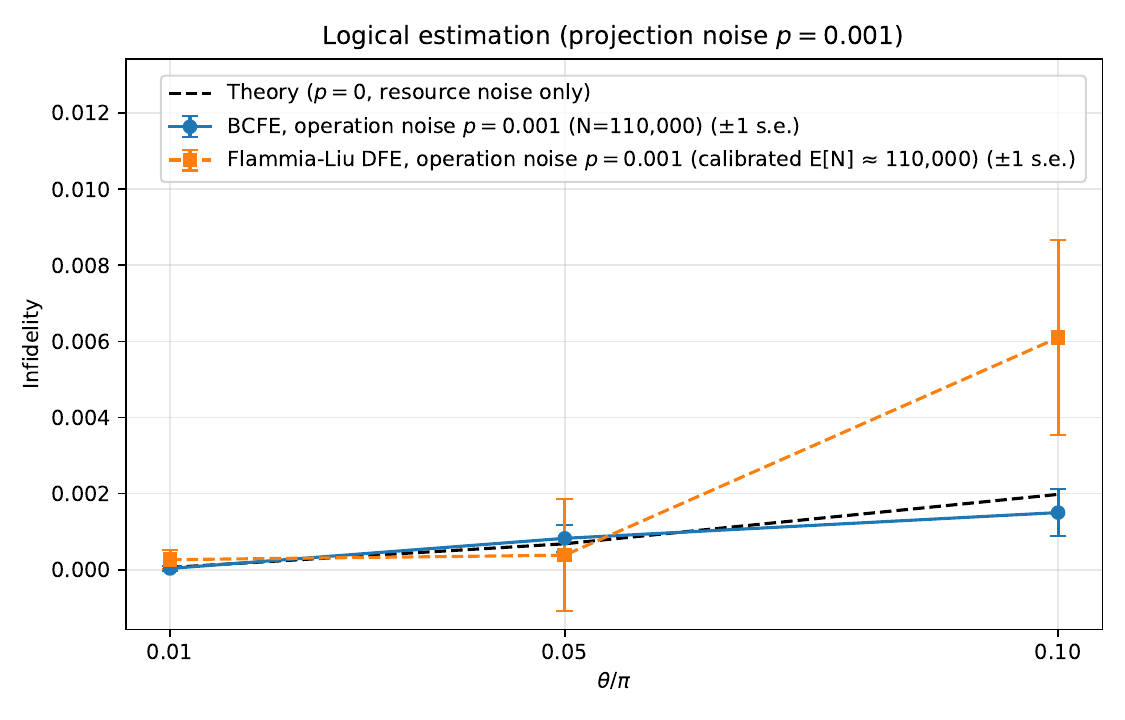}
                    \caption*{(I-c) Infidelity}
                \end{subfigure}

                \vspace{0.4em}

                \begin{subfigure}{0.32\textwidth}
                    \centering
                    \includegraphics[width=\linewidth,trim=0 0 0 25pt,clip]{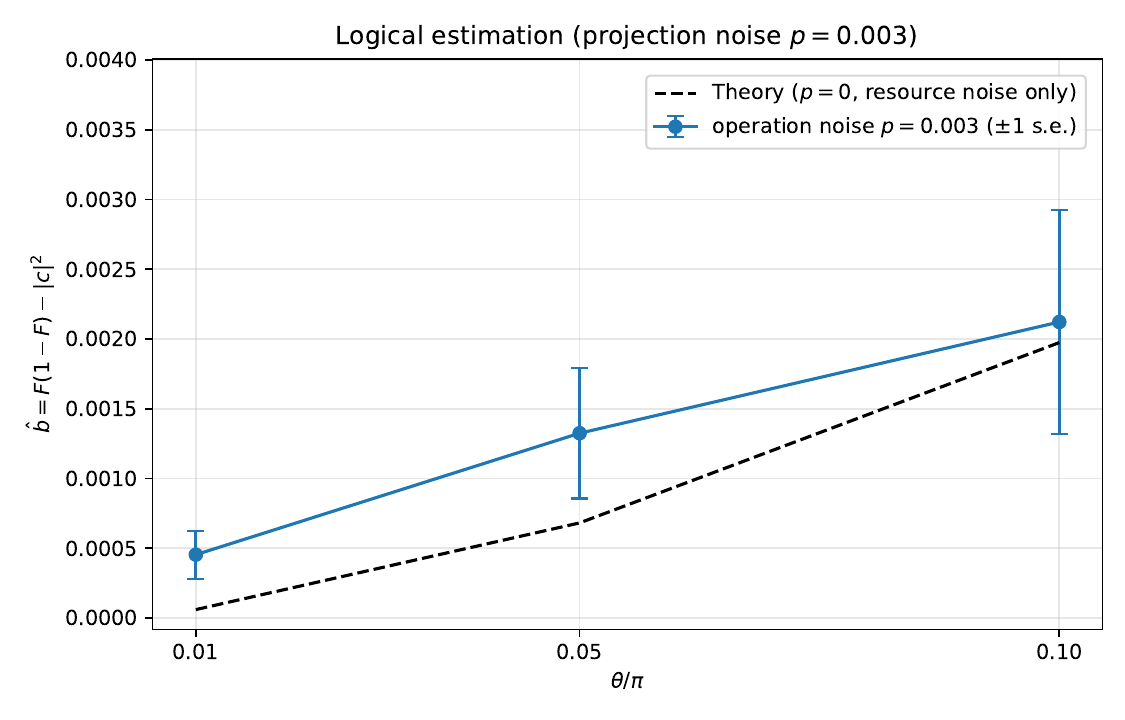}
                    \caption*{(II-a) Bell term}
                \end{subfigure}
                \begin{subfigure}{0.32\textwidth}
                    \centering
                    \includegraphics[width=\linewidth,trim=0 0 0 25pt,clip]{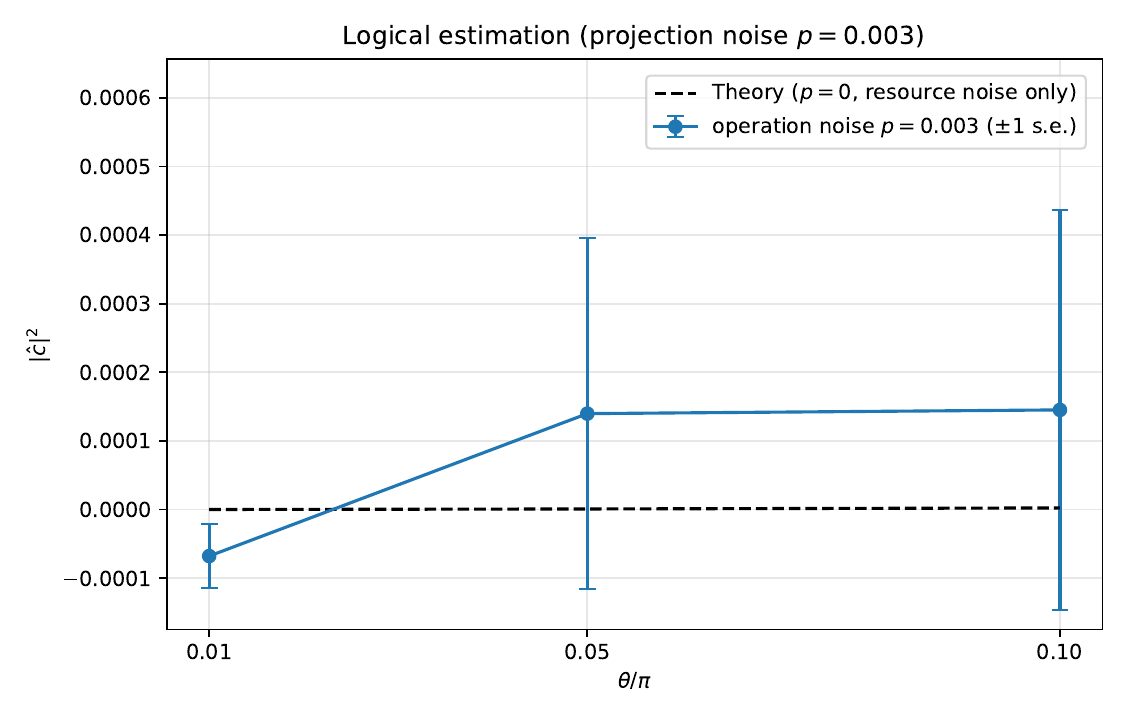}
                    \caption*{(II-b) Coherence term}
                \end{subfigure}
                \begin{subfigure}{0.32\textwidth}
                    \centering
                    \includegraphics[width=\linewidth,trim=0 0 0 25pt,clip]{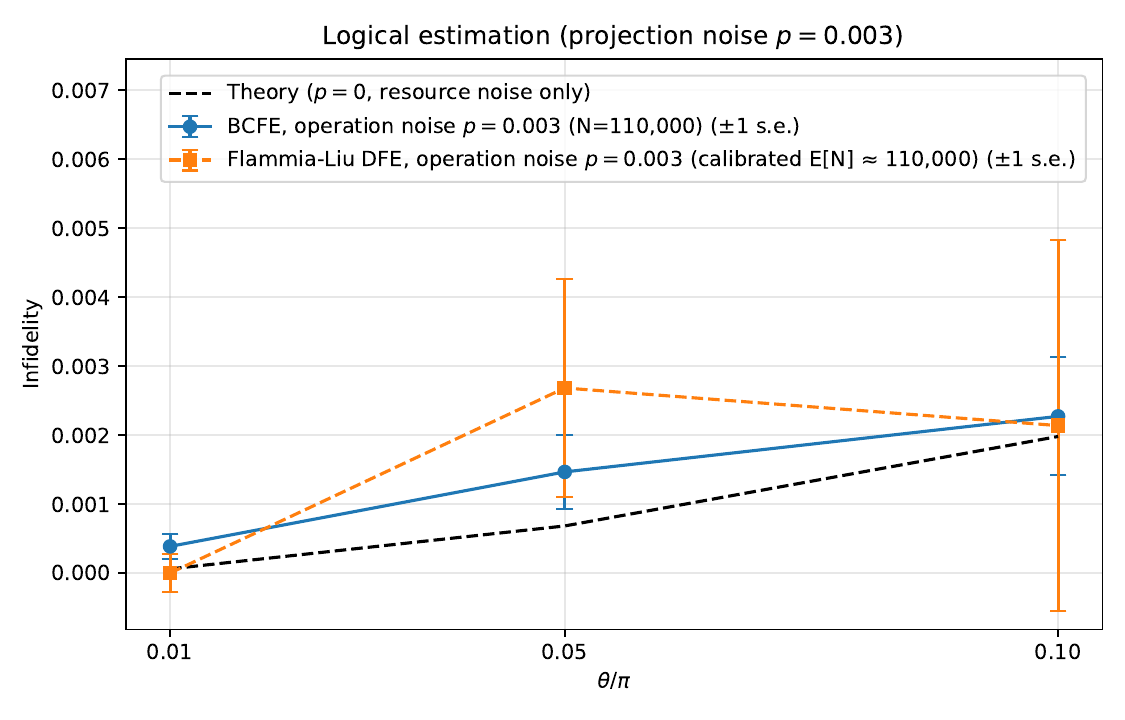}
                    \caption*{(II-c) Infidelity}
                \end{subfigure}

                \vspace{0.4em}

                \begin{subfigure}{0.32\textwidth}
                    \centering
                    \includegraphics[width=\linewidth,trim=0 0 0 25pt,clip]{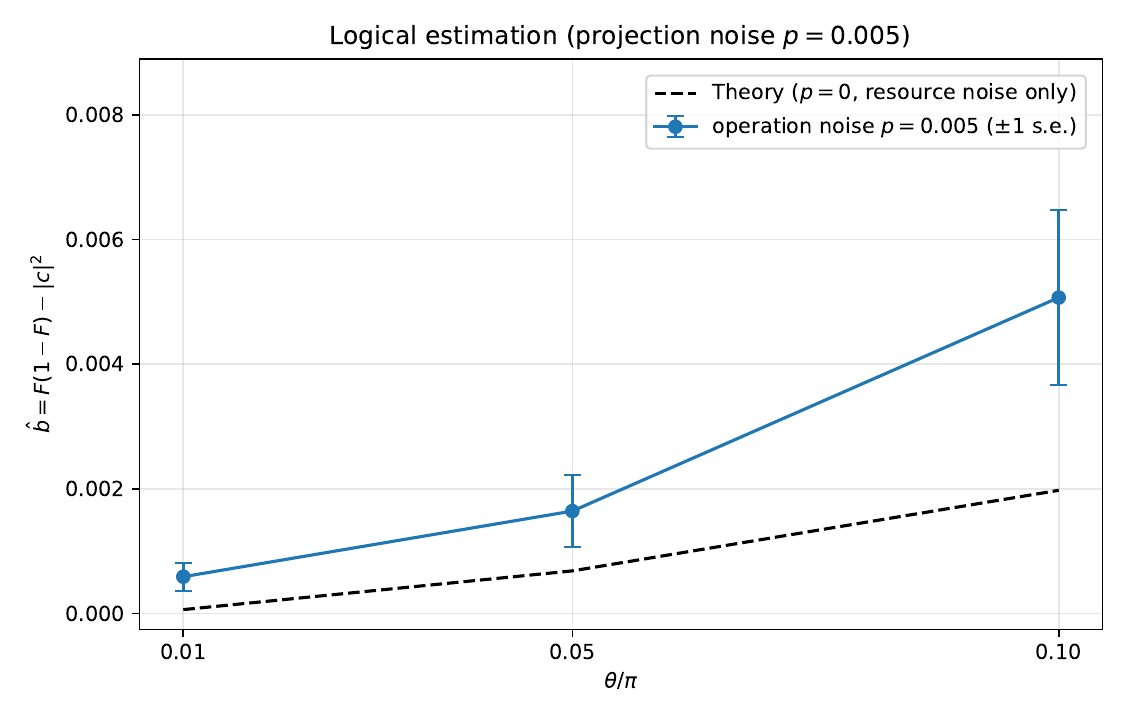}
                    \caption*{(III-a) Bell term}
                \end{subfigure}
                \begin{subfigure}{0.32\textwidth}
                    \centering
                    \includegraphics[width=\linewidth,trim=0 0 0 25pt,clip]{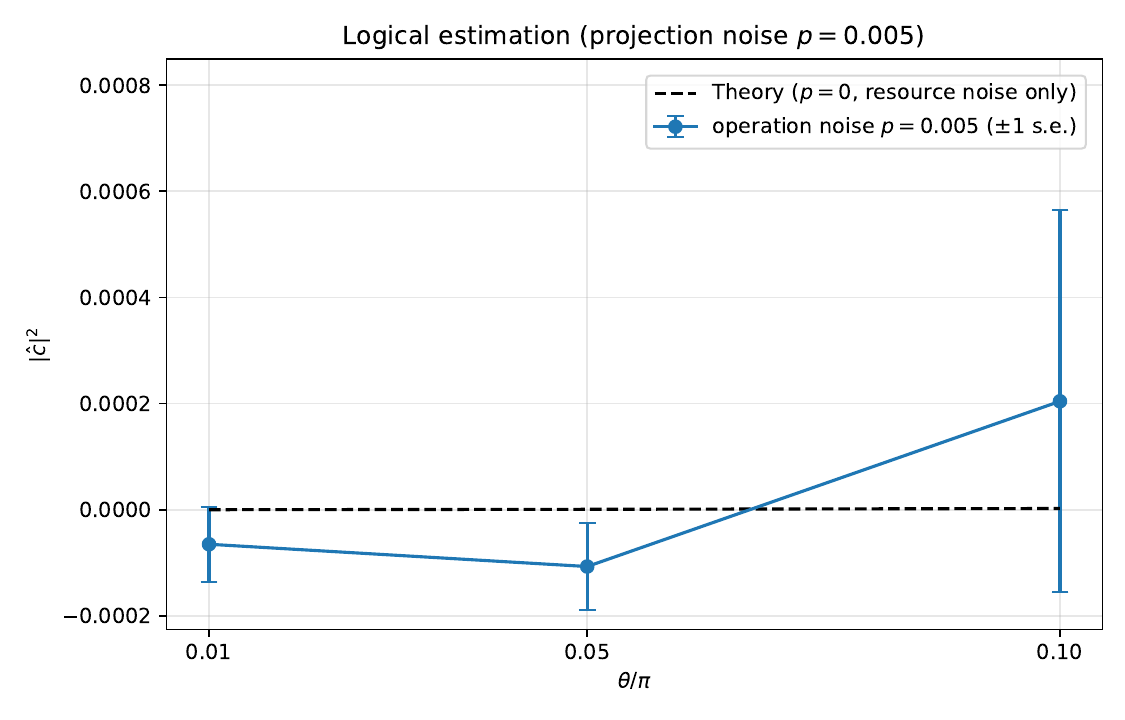}
                    \caption*{(III-b) Coherence term}
                \end{subfigure}
                \begin{subfigure}{0.32\textwidth}
                    \centering
                    \includegraphics[width=\linewidth,trim=0 0 0 25pt,clip]{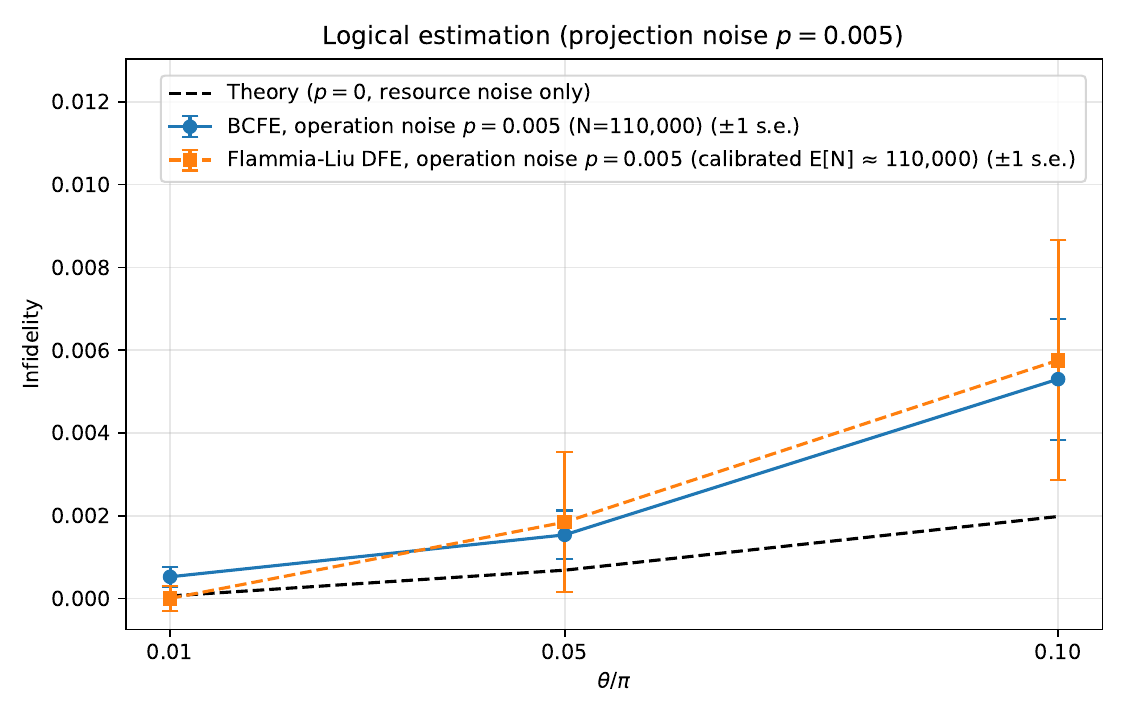}
                    \caption*{(III-c) Infidelity}
                \end{subfigure}

                \caption{Logical level fidelity estimation with a raw budget of $110{,}000$ resource states, chosen to represent an experimentally realistic data set.
                The circuit, noise rates, and arrangement of panels are the same as in Fig.\,\ref{fig:bell-benchmark-on-qec}.
                Each point uses $N_B=40{,}000$ raw Bell attempts and $N_P=10{,}000$ raw attempts in each Pauli basis, for a total of $2N_B+3N_P=110{,}000$ resource states.
                Flammia--Liu DFE uses the two-stage estimator with its expected raw consumption calibrated to this budget using the acceptance rate observed in the BCFE single-copy experiments.
                BCFE error bars use the delta-method approximation in Eq.\,\eqref{eq:bcfe-standard-error}.
                Below the error-correction threshold, increasing the code distance is expected to suppress the contribution of residual logical errors in the measurement circuits to the deviation from the reference infidelity.}
                \label{fig:logical-realistic-shot-benchmark}
            \end{figure*}

            Figure\,\ref{fig:bell-benchmark-on-qec} shows that the estimated infidelity retains its increase with $\theta$ under all three paired error settings.
            At $p_{\rm{op}}=p_{\rm{proj}}=0.001$, the proposed estimates closely follow the reference curve.
            At the two larger noise rates, the estimates lie above the reference curve, with a maximum displacement of approximately $1.10\times10^{-3}$ at $p_{\rm{op}}=p_{\rm{proj}}=0.003$ and $\theta/\pi=0.1$ in the simulation with $1{,}100{,}000$ resource states.
            However, this discrepancy need not represent an intrinsic limitation of BCFE.
            Indeed, for $d=3$, the logical error rate due to noise in the Bell measurement is $O(p_{\rm{op}}^2)$, with $p_{\rm{op}}^2$ ranging from $10^{-6}$ to $2.5\times10^{-5}$ under the simulated noise rates $p_{\rm{op}}\in\{0.001,0.003,0.005\}$.
            This contribution can be comparable to the reference infidelity of approximately $6\times10^{-5}$ to $2\times10^{-3}$, corresponding to half the conditional Bell probability in eq.\,\eqref{eq:first_order_limit} to leading order.
            As the code distance $d$ increases below the error-correction threshold, the leading contribution to the logical error rate is suppressed as $O(p_{\rm{op}}^{(d+1)/2})$, so this source of discrepancy is expected to asymptote to zero.
            The estimator continues to resolve infidelities of order $10^{-3}$.
            The coherence contribution stays near zero on this scale, although its residual discrepancy may partly reflect the absence of syndrome decoding for $Y_L$.

            Figure\,\ref{fig:logical-realistic-shot-benchmark} repeats the same simulation with one tenth as many resource states, $N_B=40{,}000$ and $N_P=10{,}000$.
            The statistical fluctuations increase, as expected, but the Bell part and the reconstructed infidelity retain their overall increase with $\theta$, and every reconstructed value remains below $6\times10^{-3}$.
            At both resource budgets, the estimated standard errors shown for our protocol are also smaller than those of DFE.
            Thus the behavior in Fig.\,\ref{fig:bell-benchmark-on-qec} is not merely a restatement of the ideal measurement result at the physical level.
            It is obtained with the transversal $\CNOT$, the noisy syndrome post-selection, and decoding of the destructive logical measurements.
            These results demonstrate that the protocol can estimate the resource fidelity under the stated physical noise model when its Bell measurement is implemented fault-tolerantly.

    \section{Conclusion and discussion} \label{sec:conclusion}
            We have developed the Bell-coherence fidelity estimation (BCFE) protocol that can be applied to arbitrary quantum resource states.
            It achieves the optimal $1/\epsilon$ sample complexity $O\!\left(r^{-2}\epsilon^{-1}\log(1/\delta)\right)$ for relative precision $r\epsilon$ with only Clifford gates and Pauli measurements.
            The key observation is that the infidelity can be reconstructed from two experimentally accessible pieces of information.
            A target independent Bell measurement estimates the determinant of the noisy state, and ordinary Pauli measurements estimate the coherence in the target basis.
            We also proved that the dependence on $r$ and $\epsilon$ is optimal up to constant factors, even when collective measurements are allowed.

            We applied BCFE to analog rotation resource states in the STAR architecture and evaluated it for both physical qubits and logical qubits in $d=3$ rotated surface codes.
            In the simulations using surface codes, the observed standard errors indicate that BCFE requires up to approximately $27$ times fewer resource states than DFE for the same standard error in the infidelity estimate.
            This estimate follows from standard errors $1.3$--$5.2$ times smaller at comparable expected resource consumption.
            At a physical CNOT error rate of $0.001$ in both the stabilizer measurements and the logical Bell measurement, the estimates closely follow the reference infidelities computed with ideal projection and measurement circuits.
            The transversal CNOT and decoded readout correct the effects of single physical gate errors, although the residual deviations at higher noise rates show that the implemented protection does not remove all effects of circuit noise.

            Finally, we briefly comment on possible future directions.
            One possible direction is to improve the numerical constants and dimension dependence of the multi-qubit extension in Appendix \ref{append:multi_qubit_generalization}.
            At fixed dimension, the dependence on the estimation parameters is optimal.
            However, the necessary dependence on the number of qubits under Clifford and Pauli measurement restrictions remains unresolved.
            Global Clifford measurements improve the dimension dependence of the sample complexity bound, but their circuit depth and classical processing can be costly.
            A sharper analysis of the squared-distance estimator, together with measurement schemes adapted to the structure of the target state, could reduce the number of samples without requiring such demanding circuits.

            A second direction involves the assumption that resource states are prepared independently in the same state $\rho$.
            The Bell measurement probes $\rho^{\otimes2}$ under this assumption, whereas correlations between preparations or drift during data acquisition can change both its statistics and its relation to the Pauli estimates used to reconstruct the fidelity.
            For such sources, the quantity to be estimated must first be specified, for example the mean fidelity over a sequence of preparations.
            It remains to determine under what assumptions on the correlations and drift BCFE can estimate this quantity with a controlled error, and whether the $O(1/\epsilon)$ sample complexity can be retained.

            For applications to the STAR architecture, a further direction is to bound the diamond-norm error of the implemented analog rotation in terms of the resource state infidelities measured by BCFE.
            In particular, the resource states are consumed through gate teleportation and the repeat-until-success (RUS) protocol reviewed in Appendix \ref{append:star}.
            Such bounds should account for the sequence of resource states consumed throughout the RUS protocol.
            We could then determine the resource state fidelities and estimation accuracies needed to keep the accumulated error within a prescribed tolerance in an algorithm using many analog rotations.

    \acknowledgments
        Research at Perimeter Institute is supported in part by the Government of Canada through the Department of Innovation, Science and Economic Development and by the Province of Ontario through the Ministry of Colleges and Universities.
        KF is supported by MEXT Quantum Leap Flagship Program (MEXT Q-LEAP) Grant No.\,JPMXS0120319794, JST COI-NEXT Grant No.\,JPMJPF2014, JST Moonshot R\&D Grant No.\,JPMJMS256E, and JST CREST JPMJCR24I3. \\

    \noindent\paragraph*{Note added.}
        During the finalization of this manuscript, we became aware of Ref.\,\cite{Daguerre:2025logical}, whose experiment part employs the same Bell and single-copy Pauli measurements to bound the fidelity and states $F\ge 1-O(1/N_{\rm shots})$ scaling.
        We establish finite-sample guarantees for infidelity estimation, including the statistical error of the single-copy contribution, and also formulate a general protocol for multi-qubit target states.

    \appendix

    \section{Technical inequalities} \label{append:technical-inequalities}
        In this appendix, we collect the four probabilistic inequalities used in the sample complexity proofs.
        The first three control the finite sample fluctuations of the estimators, while the fourth converts relative entropy into a lower bound on the error probability of binary hypothesis testing.
        Only a variance bound enters Chebyshev's inequality, whose tail decays polynomially in the deviation rather than exponentially.
        Bernstein's inequality retains the total variance in the exponential bound, which becomes tighter as that variance decreases.
        Hoeffding's inequality instead uses only the range of each random variable, making it suitable when no comparably sharp variance bound is available.

        \begin{thm}[Chebyshev's inequality] \label{thm:chebyshev-inequality}
            Let $Y$ be a real random variable with finite variance.
            For every $t>0$,
            \begin{align}
                \Pr\!\left[
                    \left|
                        Y-\mathbb{E}[Y]
                    \right|
                    \geqq t
                \right]
                \leqq
                \frac{\operatorname{Var}(Y)}{t^2}.
                \label{eq:technical-chebyshev}
            \end{align}
            In particular, the mean $\overline Y_N:=N^{-1}\sum_{j=1}^{N}Y_j$ of $N$ independent copies of $Y$ obeys
            \begin{align}
                \Pr\!\left[
                    \left|
                        \overline Y_N-\mathbb{E}[Y]
                    \right|
                    \geqq t
                \right]
                \leqq
                \frac{\operatorname{Var}(Y)}{Nt^2}.
                \label{eq:technical-chebyshev-average}
            \end{align}
        \end{thm}

        \begin{thm}[Bernstein's inequality] \label{thm:bernstein-inequality}
            Let $X_1,\ldots,X_N$ be independent real random variables.
            Suppose that every centered variable is bounded almost surely
            \begin{align}
                |X_j-\mathbb{E}[X_j]|\leqq b.
            \end{align}
            Define their total variance by
            \begin{align}
                V
                :=
                \sum_{j=1}^{N}\operatorname{Var}(X_j).
            \end{align}
            For every $t>0$,
            \begin{align}
                \Pr\!\left[
                    \left|
                        \sum_{j=1}^{N}
                        \left(
                            X_j-\mathbb{E}[X_j]
                        \right)
                    \right|
                    \geqq t
                \right]
                \leqq
                2\exp\!\left(
                    -\frac{t^2}{2V+2bt/3}
                \right).
                \label{eq:technical-bernstein}
            \end{align}
            In particular, if $b=1$ and $\operatorname{Var}(X_j)\leqq\sigma^2$ for every $j$, then
            \begin{align}
                \Pr\!\left[
                    \left|
                        \frac{1}{N}
                        \sum_{j=1}^{N}
                        \left(
                            X_j-\mathbb{E}[X_j]
                        \right)
                    \right|
                    \geqq t
                \right]
                \leqq
                2\exp\!\left(
                    -\frac{Nt^2}{2\sigma^2+2t/3}
                \right).
                \label{eq:technical-bernstein-average}
            \end{align}
        \end{thm}

        \begin{thm}[Hoeffding's inequality \cite{Hoeffding:1963inequalities}] \label{thm:hoeffding-inequality}
            Let $X_1,\ldots,X_N$ be independent real random variables with $X_j\in[a_j,b_j]$ almost surely.
            For every $t>0$,
            \begin{align}
                \Pr\!\left[
                    \sum_{j=1}^{N}
                    \left(
                        X_j-\mathbb{E}[X_j]
                    \right)
                    \geqq t
                \right]
                \leqq
                \exp\!\left[
                    -\frac{2t^2}
                    {\sum_{j=1}^{N}(b_j-a_j)^2}
                \right],
                \label{eq:technical-hoeffding-one-sided}
            \end{align}
            and hence
            \begin{align}
                \begin{aligned}
                    &\Pr\!\left[
                        \left|
                            \sum_{j=1}^{N}
                            \left(
                                X_j-\mathbb{E}[X_j]
                            \right)
                        \right|
                        \geqq t
                    \right]
                    \\
                    &\qquad\leqq
                    2\exp\!\left[
                        -\frac{2t^2}
                        {\sum_{j=1}^{N}(b_j-a_j)^2}
                    \right].
                \end{aligned}
                \label{eq:technical-hoeffding-two-sided}
            \end{align}
        \end{thm}

        \begin{thm}[Bretagnolle--Huber inequality] \label{thm:bretagnolle-huber-inequality}
            Let $P$ and $Q$ be probability measures with $P$ absolutely continuous with respect to $Q$.
            Their relative entropy is
            \begin{align}
                D_{\mathrm{KL}}(P\Vert Q)
                :=
                \int
                \log\!\left(
                    \frac{\mathrm{d}P}{\mathrm{d}Q}
                \right)
                \mathrm{d}P.
            \end{align}
            Every measurable event $E$ satisfies
            \begin{align}
                P(E^{\mathrm{c}})+Q(E)
                \geqq
                \frac{1}{2}
                \exp\!\left[
                    -D_{\mathrm{KL}}(P\Vert Q)
                \right].
                \label{eq:technical-bretagnolle-huber}
            \end{align}
            Thus, a binary test that selects $P$ on $E$ and $Q$ on $E^{\mathrm{c}}$ has equal-prior average error probability
            \begin{align}
                P_{\mathrm{err}}
                =
                \frac{1}{2}P(E^{\mathrm{c}})
                +
                \frac{1}{2}Q(E)
                \geqq
                \frac{1}{4}
                \exp\!\left[
                    -D_{\mathrm{KL}}(P\Vert Q)
                \right].
                \label{eq:technical-bretagnolle-huber-testing}
            \end{align}
        \end{thm}

    \section{Multi-qubit generalization of BCFE} \label{append:multi_qubit_generalization}
        This appendix extends the single-qubit BCFE protocol of Sec.\,\ref{sec:general_protocol} to an $m$-qubit target state.

        \subsection{The protocol}
            As in the single-qubit protocol of Sec.\,\ref{sec:general_protocol}, we estimate the infidelity from a Bell measurement and a correction determined from single-copy measurements.
            In the multi-qubit formulation, this correction is expressed through the squared Hilbert--Schmidt distance from the target state and plays the role of the coherence part.
            \paragraph{Exact relation and Bell measurement.}
            Suppose the Hilbert space has dimension $D=2^m$, and let
            \begin{align}
                P
                :=
                |\psi\rangle\langle\psi|
            \end{align}
            be the projector onto a given $m$-qubit target state.
            For a prepared state $\rho$, define
            \begin{align}
                \begin{aligned}
                    F
                    &:=
                    \langle\psi|\rho|\psi\rangle,
                    \qquad
                    \epsilon
                    :=
                    1-F,
                    \\
                    \Delta
                    &:=
                    \rho-P,
                    \qquad
                    h
                    :=
                    \operatorname{tr}(\Delta^2).
                \end{aligned}
            \end{align}

            Suppose $\SWAP$ exchanges two $m$-qubit registers and let $\Pi_-:=(I-\SWAP)/2$ project onto their antisymmetric subspace.
            The probability of the antisymmetric outcome is
            \begin{align}
                b
                :=
                \operatorname{tr}\!\left[
                    \Pi_-\rho^{\otimes2}
                \right]
                =
                \frac{1-\operatorname{tr}\rho^2}{2}.
            \end{align}
            Combined with the squared distance $h$, this purity measurement reproduces the infidelity exactly since
            \begin{align}
                \begin{aligned}
                    b+\frac{h}{2}
                    &=
                    \frac{1-\operatorname{tr}\rho^2}{2}
                    +
                    \frac{
                        \operatorname{tr}\rho^2
                        +1
                        -2F
                    }{2}
                    \\
                    &=
                    1-F
                    =
                    \epsilon.
                \end{aligned}
                \label{eq:multiqubit-direct-infidelity}
            \end{align}
            Thus the two quantities determine $\epsilon$ linearly, with no branch choice or nonlinear inversion.
            They also vanish with the infidelity:
            \begin{align}
                h
                \leqq
                2\epsilon,
                \qquad
                b
                \leqq
                \frac{1-F^2}{2}
                \leqq
                \epsilon.
                \label{eq:multiqubit-component-bounds}
            \end{align}
            The first inequality follows from $\operatorname{tr}\rho^2\leqq1$, while the second uses $\operatorname{tr}\rho^2\geqq F^2$.

            The antisymmetric measurement does not require a controlled-$\SWAP$ gate.
            The register exchange is the product of the pairwise exchanges, $\SWAP=\prod_{j=1}^{m}\SWAP_j$, and a Bell measurement resolves the eigenvalue of every $\SWAP_j$.
            Let $s_j\in\{0,1\}$ indicate whether the $j$th pair is measured in the singlet state.
            The measured eigenvalue of the register exchange and the final classical bit are
            \begin{align}
                \begin{aligned}
                    \lambda_{\SWAP}
                    &=
                    (-1)^{\sum_{j=1}^{m}s_j},
                    \\
                    g
                    &:=
                    s_1\mathbin{\oplus}\cdots\mathbin{\oplus}s_m,
                    \qquad
                    \Pr(g=1)
                    =
                    b.
                \end{aligned}
                \label{eq:multiqubit-bell-parity}
            \end{align}
            Figure~\ref{fig:multiqubit-bell-circuit} illustrates the circuit applied in parallel to the $m$ qubit pairs.
            \begin{figure}[t]
                \centering
                \begin{quantikz}[column sep=0.35cm]
                    \lstick{$q_j^{(1)}$} & \ctrl{1} & \gate{H} & \meter{} & \rstick{$a_j$} \\
                    \lstick{$q_j^{(2)}$} & \targ{} & \qw & \meter{} & \rstick{$b_j$}
                \end{quantikz}
                \par\smallskip
                $s_j:=a_jb_j,\qquad g:=s_1\mathbin{\oplus}\cdots\mathbin{\oplus}s_m.$
                \caption{Bell measurement circuit for the antisymmetric measurement.
                The circuit is applied to the $j$th qubits of two independently prepared $m$-qubit registers for every $j=1,\ldots,m$.
                The output $(a_j,b_j)=(1,1)$ identifies the local singlet.
                The antisymmetric outcome is recorded when $g=1$, equivalently when an odd number of qubit pairs are singlets.}
                \label{fig:multiqubit-bell-circuit}
            \end{figure}
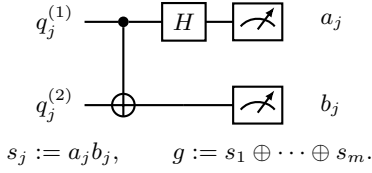
            Each trial consumes two independently prepared copies of $\rho$.

            \paragraph{Pauli estimate of the squared distance.}
            Let $\mathcal{P}_m:=\{I,X,Y,Z\}^{\otimes m}$ be the $m$-qubit Pauli group, with the orthogonal relation $\operatorname{tr}(RS)=D\delta_{R,S}$ for $R,S\in\mathcal{P}_m$.
            For every Pauli string, the known target and the prepared state have expectation values
            \begin{align}
                \begin{aligned}
                    p_R
                    &:=
                    \operatorname{tr}(RP),
                    \qquad
                    m_R
                    :=
                    \operatorname{tr}(R\rho),
                    \\
                    g_R
                    &:=
                    m_R-p_R
                    =
                    \operatorname{tr}(R\Delta).
                \end{aligned}
            \end{align}
            By completeness of the Pauli basis, the squared distance is expressed as
            \begin{align}
                h
                =
                \frac{1}{D}
                \sum_{R\in\mathcal P_m}
                g_R^2.
                \label{eq:multiqubit-pauli-distance}
            \end{align}
            The identity string contributes nothing because $g_I=\operatorname{tr}\Delta=0$.

            For each $R\neq I$, we measure $R$ on $N_P$ independent copies of $\rho$ and denote the outcomes by $X_{R,1},\ldots,X_{R,N_P}\in\{-1,+1\}$.
            The following U-statistic of degree two \cite{Hoeffding:1948statistics} removes the bias that would arise from squaring a sample mean
            \begin{align}
                \begin{aligned}
                    U_R
                    &:=
                    \frac{1}{N_P(N_P-1)}
                    \sum_{\substack{1\leqq i,j\leqq N_P\\i\neq j}}
                    (X_{R,i}-p_R)(X_{R,j}-p_R)
                    \\
                    &=
                    \frac{
                        \left[
                            \sum_{i=1}^{N_P}
                            (X_{R,i}-p_R)
                        \right]^2
                        -
                        \sum_{i=1}^{N_P}
                        (X_{R,i}-p_R)^2
                    }{
                        N_P(N_P-1)
                    }.
                \end{aligned}
                \label{eq:multiqubit-pauli-u-statistic}
            \end{align}
            Here, the expectation value is $\mathbb E[U_R]=g_R^2$.
            The values $p_R$ enter only through this classical calculation.
            In particular, the quantum circuit never applies $P$ or performs the measurement $\{P,I-P\}$ so that it does not violate the fundamental demand that any operations must be either Clifford gates or Pauli measurements.

            One estimation block combines these Pauli data with an independent empirical antisymmetric frequency $Y_-$ obtained from $N_B$ Bell trials.
            Define
            \begin{align}
                \begin{aligned}
                    \widehat h_{\rm Pauli}
                    &:=
                    \frac{1}{D}
                    \sum_{R\in\mathcal P_m\setminus\{I\}}
                    U_R,
                    \\
                    Y
                    &:=
                    Y_-+\frac{1}{2}\widehat h_{\rm Pauli}.
                \end{aligned}
                \label{eq:multiqubit-pauli-block-estimator}
            \end{align}
            Equations\,\eqref{eq:multiqubit-direct-infidelity} and \eqref{eq:multiqubit-pauli-distance} ensure $\mathbb E[Y]=\epsilon$.
            The block is repeated independently an odd number $L$ of times.
            If $Y^{(1)},\ldots,Y^{(L)}$ are the resulting estimates, the final output is
            \begin{align}
                \begin{aligned}
                    \widehat\epsilon_{\rm{med}}
                    &:=
                    \operatorname{med}
                    \big(
                        Y^{(1)},\ldots,Y^{(L)}
                    \big),
                    \\
                    \widehat\epsilon
                    &:=
                    \min\!\left\{
                        1,
                        \max\!\left\{
                            0,
                            \widehat\epsilon_{\rm{med}}
                        \right\}
                    \right\},
                    \\
                    \widehat F
                    &:=
                    1-\widehat\epsilon.
                \end{aligned}
            \end{align}

            The concentration of this estimator does not depend on the Pauli expansion of the target state.
            \begin{thm}[Multi-qubit extension]
                \label{thm:multi-qubit-extension}
                Fix $0<r\leqq1$ and $0<\delta<1/8$.
                Let $|\psi\rangle$ be a given $m$-qubit target state and let $\rho$ have infidelity $0<\epsilon=1-\langle\psi|\rho|\psi\rangle\leqq1$.
                The estimator above can be implemented using only the antisymmetric measurement and Pauli string measurements, and it satisfies
                \begin{align}
                    |\widehat\epsilon-\epsilon|
                    \leqq
                    r\epsilon
                \end{align}
                with probability at least $1-\delta$ using
                \begin{align}
                    \begin{aligned}
                        N_{\rm Pauli}
                        &=
                        O\!\left[
                            \left(
                                \frac{D}{r^2\epsilon}
                                +
                                \frac{D^2}{r\epsilon}
                            \right)
                            \log\frac{1}{\delta}
                        \right]
                        \\
                        &=
                        O\!\left(
                            \frac{D^2}{r^2\epsilon}
                            \log\frac{1}{\delta}
                        \right)
                    \end{aligned}
                \end{align}
                copies of $\rho$.
                The hidden numerical constants are universal and independent of the target state.
            \end{thm}

            \begin{proof}
                Fix $R\neq I$ and write
                \begin{align}
                    Z_{R,i}
                    :=
                    X_{R,i}-p_R,
                    \qquad
                    v_R
                    :=
                    \operatorname{Var}(Z_{R,i})
                    =
                    1-m_R^2
                    \leqq
                    1.
                \end{align}
                The mean of $Z_{R,i}$ is $g_R$.
                The centered decomposition of Eq.\,\eqref{eq:multiqubit-pauli-u-statistic} is
                \begin{align}
                    \begin{aligned}
                        U_R-g_R^2
                        &=
                        \frac{2g_R}{N_P}
                        \sum_{i=1}^{N_P}
                        (Z_{R,i}-g_R)
                        \\
                        &\quad
                        +
                        \frac{1}{N_P(N_P-1)}
                        \sum_{\substack{1\leqq i,j\leqq N_P\\i\neq j}}
                        \big(
                            Z_{R,i}-g_R
                        \big)
                        \big(
                            Z_{R,j}-g_R
                        \big).
                    \end{aligned}
                \end{align}
                The conditional expectation of the second term with respect to any one centered variable is zero, so it is uncorrelated with the first term.
                The exact variance is
                \begin{align}
                    \operatorname{Var}(U_R)
                    =
                    \frac{4g_R^2v_R}{N_P}
                    +
                    \frac{2v_R^2}{N_P(N_P-1)}.
                    \label{eq:multiqubit-pauli-u-variance}
                \end{align}

                Distinct Pauli strings are measured on disjoint copies.
                Their estimates are therefore independent, and Eqs.\,\eqref{eq:multiqubit-pauli-distance} and \eqref{eq:multiqubit-pauli-u-variance} imply
                \begin{align}
                    \begin{aligned}
                        \operatorname{Var}\!\left(
                            \frac{1}{2}
                            \widehat h_{\rm Pauli}
                        \right)
                        &\leqq
                        \frac{1}{4D^2}
                        \sum_{R\neq I}
                        \left[
                            \frac{4g_R^2}{N_P}
                            +
                            \frac{2}{N_P(N_P-1)}
                        \right]
                        \\
                        &=
                        \frac{h}{DN_P}
                        +
                        \frac{D^2-1}{2D^2N_P(N_P-1)}
                        \\
                        &\leqq
                        \frac{2\epsilon}{DN_P}
                        +
                        \frac{1}{2N_P(N_P-1)}.
                    \end{aligned}
                    \label{eq:multiqubit-pauli-block-variance}
                \end{align}

                The $N_B$ Bell outcomes in the same block are independent Bernoulli variables with mean $b$.
                Equations\,\eqref{eq:multiqubit-component-bounds} and \eqref{eq:multiqubit-pauli-block-variance} bound the variance of the combined block estimate by
                \begin{align}
                    \operatorname{Var}(Y)
                    \leqq
                    \frac{\epsilon}{N_B}
                    +
                    \frac{2\epsilon}{DN_P}
                    +
                    \frac{1}{2N_P(N_P-1)}.
                    \label{eq:multiqubit-combined-block-variance}
                \end{align}

                One explicit choice is
                \begin{align}
                    N_B
                    :=
                    \left\lceil
                        \frac{24}{r^2\epsilon}
                    \right\rceil,
                    \qquad
                    N_P
                    :=
                    \left\lceil
                        \frac{48}{Dr^2\epsilon}
                        +
                        \frac{4}{r\epsilon}
                    \right\rceil.
                    \label{eq:multiqubit-pauli-block-sizes}
                \end{align}
                Since $N_P\geqq4$, the three terms in Eq.\,\eqref{eq:multiqubit-combined-block-variance} are at most $r^2\epsilon^2/24$, $r^2\epsilon^2/24$, and $r^2\epsilon^2/16$, respectively.
                Hence
                \begin{align}
                    \operatorname{Var}(Y)
                    \leqq
                    \frac{7}{48}
                    r^2\epsilon^2
                    <
                    \frac{1}{4}
                    r^2\epsilon^2.
                \end{align}
                Chebyshev's inequality (Theorem \ref{thm:chebyshev-inequality}) now bounds the failure probability of one block:
                \begin{align}
                    \Pr\!\left[
                        |Y-\epsilon|>r\epsilon
                    \right]
                    \leqq
                    \frac{
                        \operatorname{Var}(Y)
                    }{
                        r^2\epsilon^2
                    }
                    \leqq
                    \frac{1}{4}.
                \end{align}

                Choose an odd number of independent blocks satisfying $L\geqq8\log(1/\delta)$.
                For block $\ell$, introduce the indicator
                \begin{align}
                    I^{(\ell)}
                    :=
                    \mathbf 1\!\left\{
                        |Y^{(\ell)}-\epsilon|
                        >
                        r\epsilon
                    \right\},
                    \qquad
                    \mathbb E[I^{(\ell)}]
                    \leqq
                    \frac{1}{4}.
                \end{align}
                If the median differs from $\epsilon$ by more than $r\epsilon$, at least half of the blocks have failed.
                Hence
                \begin{align}
                    \begin{aligned}
                        \Pr\!\left[
                            \left|
                                \widehat\epsilon_{\rm{med}}
                                -\epsilon
                            \right|
                            >
                            r\epsilon
                        \right]
                        &\leqq
                        \Pr\!\left[
                            \frac{1}{L}
                            \sum_{\ell=1}^{L}I^{(\ell)}
                            \geqq
                            \frac{1}{2}
                        \right]
                        \\
                        &\leqq
                        \Pr\!\left[
                            \frac{1}{L}
                            \sum_{\ell=1}^{L}
                            \left(
                                I^{(\ell)}
                                -\mathbb E[I^{(\ell)}]
                            \right)
                            \geqq
                            \frac{1}{4}
                        \right]
                        \\
                        &\leqq
                        \exp\Bigl(
                            -2L\cdot\frac{1}{4^2}
                        \Bigl)
                        =
                        e^{-L/8}
                        \leqq
                        \delta.
                    \end{aligned}
                \end{align}
                The third inequality is Hoeffding's inequality (Theorem \ref{thm:hoeffding-inequality}) for the independent variables $I^{(1)},\ldots,I^{(L)}\in\{0,1\}$.
                Projection onto $[0,1]$ cannot increase the distance from $\epsilon\in[0,1]$.

                One block consumes
                \begin{align}
                    2N_B
                    +
                    (D^2-1)N_P
                \end{align}
                copies, where the factor of two accounts for the two inputs to each Bell measurement.
                Equation\,\eqref{eq:multiqubit-pauli-block-sizes} makes this count
                \begin{align}
                    O\!\left(
                        \frac{D}{r^2\epsilon}
                        +
                        \frac{D^2}{r\epsilon}
                    \right)
                \end{align}
                per block.
                The $L$ independent blocks produce the stated total.
            \end{proof}

            The first term in Eq.\,\eqref{eq:multiqubit-pauli-block-variance} contributes at most $O(\sqrt{\epsilon/(DN_P)})$ to the standard deviation.
            It plays the same role as the intermediate accuracy of order $r\sqrt{\epsilon}$ in the proof of Theorem \ref{thm:sample-complexity}, but it now controls the squared distance as a whole.
            For $D=2$, Eq.\,\eqref{eq:multiqubit-pauli-distance} contains the three nonidentity Pauli observables, so the measurement settings reduce to those used in Sec.\,\ref{sec:general_protocol}.

            \paragraph{Global Clifford measurements.}
            The preceding construction assigns a separate batch to each of the $D^2-1$ nonidentity Pauli strings.
            Classical shadows obtained from random measurements over the full $m$-qubit Clifford group avoid this setting-by-setting allocation and reduce the dimensional overhead \cite{Huang:2020classical}.
            Here a global Clifford acts on the entire register and may contain entangling gates; it is not restricted to a tensor product of single qubit Clifford gates.

            For each copy, we choose a uniformly random Clifford $U$, apply it to $\rho$, and measure every qubit in the computational basis.
            If the outcome is $b$, the corresponding snapshot is
            \begin{align}
                \widehat\rho
                :=
                (D+1)
                U^\dagger
                |b\rangle\langle b|
                U
                -
                I,
                \qquad
                \mathbb E[\widehat\rho]
                =
                \rho.
                \label{eq:multiqubit-global-clifford-snapshot}
            \end{align}
            The second identity follows from the fact that the Clifford group forms a unitary 2-design.
            The snapshot has unit trace but need not be positive.
            We split the snapshots into two independent batches of $N_P$ and denote their averages by $\overline\rho^{(1)}$ and $\overline\rho^{(2)}$.
            Their cross-correlation
            \begin{align}
                \widehat h_{\rm Clifford}
                :=
                \operatorname{tr}\!\left[
                    \left(
                        \overline\rho^{(1)}-P
                    \right)
                    \left(
                        \overline\rho^{(2)}-P
                    \right)
                \right]
                \label{eq:multiqubit-global-clifford-estimator}
            \end{align}
            is unbiased for $h$.
            Although $P$ occurs in the classical calculation,
            \begin{align}
                \operatorname{tr}(P\widehat\rho)
                =
                (D+1)
                \left|
                    \langle\psi|
                    U^\dagger
                    |b\rangle
                \right|^2
                -
                1,
            \end{align}
            no target projection is applied to the quantum state.

            For this measurement ensemble, the snapshot obeys the following second moment bound for every traceless Hermitian operator $O$ \cite{Huang:2020classical}:
            \begin{align}
                \mathbb E\!\left[
                    \operatorname{tr}\!\left(
                        O(\widehat\rho-\rho)
                    \right)^2
                \right]
                \leqq
                3\operatorname{tr}(O^2).
                \label{eq:multiqubit-global-clifford-moment}
            \end{align}
            Let $\Sigma$ be the covariance operator of $\widehat\rho-\rho$ on the real vector space of traceless Hermitian operators, equipped with the Hilbert--Schmidt inner product $\langle A,B\rangle_{\mathrm{HS}}:=\operatorname{tr}(AB)$.
            The operator norm satisfies $\|\Sigma\|_{\rm op}\leqq3$ by Eq.\,\eqref{eq:multiqubit-global-clifford-moment}.
            Moreover, every snapshot satisfies $\operatorname{tr}(\widehat\rho^2)=D^2+D-1$, and hence
            \begin{align}
                \operatorname{tr}\Sigma
                =
                \mathbb E\!\left[
                    \operatorname{tr}\!\left[
                        (\widehat\rho-\rho)^2
                    \right]
                \right]
                =
                D^2+D-1-\operatorname{tr}\rho^2
                \leqq
                D^2+D.
            \end{align}
            Since $\Sigma$ is positive semidefinite, $\operatorname{tr}(\Sigma^2)\leqq\|\Sigma\|_{\rm op}\operatorname{tr}\Sigma\leqq3(D^2+D)$.

            Write $E_a:=\overline\rho^{(a)}-\rho$ for $a=1,2$.
            The difference between Eq.\,\eqref{eq:multiqubit-global-clifford-estimator} and $h$ is
            \begin{align}
                \widehat h_{\rm Clifford}-h
                =
                \operatorname{tr}(\Delta E_1)
                +
                \operatorname{tr}(\Delta E_2)
                +
                \operatorname{tr}(E_1E_2).
            \end{align}
            The three terms are pairwise uncorrelated because the batches are independent and $\mathbb E[E_a]=0$.
            Consequently,
            \begin{align}
                \begin{aligned}
                    \operatorname{Var}(
                        \widehat h_{\rm Clifford}
                    )
                    &=
                    \frac{2}{N_P}
                    \langle\Delta,\Sigma\Delta\rangle_{\mathrm{HS}}
                    +
                    \frac{1}{N_P^2}
                    \operatorname{tr}(\Sigma^2)
                    \\
                    &\leqq
                    \frac{6h}{N_P}
                    +
                    \frac{3(D^2+D)}{N_P^2}
                    \\
                    &\leqq
                    \frac{12\epsilon}{N_P}
                    +
                    \frac{3(D^2+D)}{N_P^2}.
                \end{aligned}
                \label{eq:multiqubit-global-clifford-variance}
            \end{align}
            An independent Bell estimate from $N_B$ trials therefore satisfies
            \begin{align}
                \operatorname{Var}\!\left(
                    Y_-+\frac{1}{2}
                    \widehat h_{\rm Clifford}
                \right)
                \leqq
                \frac{\epsilon}{N_B}
                +
                \frac{3\epsilon}{N_P}
                +
                \frac{3(D^2+D)}{4N_P^2}.
            \end{align}
            The same argument based on independent blocks and their median applies to $Y_-+\widehat h_{\rm Clifford}/2$.
            The required batch sizes are
            \begin{align}
                N_B
                =
                \Theta\!\left(
                    \frac{1}{r^2\epsilon}
                \right),
                \qquad
                N_P
                =
                \Theta\!\left(
                    \frac{1}{r^2\epsilon}
                    +
                    \frac{D}{r\epsilon}
                \right).
            \end{align}
            Since the two shadow batches use $2N_P$ copies, the total is
            \begin{align}
                \begin{aligned}
                    N_{\rm Clifford}
                    &=
                    O\!\left[
                        \left(
                            \frac{1}{r^2\epsilon}
                            +
                            \frac{D}{r\epsilon}
                        \right)
                        \log\frac{1}{\delta}
                    \right].
                \end{aligned}
                \label{eq:multiqubit-global-clifford-copy-bound}
            \end{align}
            This bound counts copies rather than gates.
            Uniform global Clifford circuits generally contain entangling gates, and evaluating $|\langle\psi|U^\dagger|b\rangle|^2$ need not be efficient for an arbitrary target representation.

        \subsection{Optimal dependence on $r$, $\epsilon$ and $\delta$}
            To establish optimality in $r$, $\epsilon$, and $\delta$, we adapt the proof of Theorem \ref{thm:sample-optimality} to an arbitrary $m$-qubit target $\ket{\psi}$.
            The only change is to choose a state $\ket{\phi}$ orthogonal to $\ket{\psi}$ in the $m$-qubit Hilbert space and restrict the prepared states to the family
            \begin{align}
                \rho_p
                :=(1-p)\ket{\psi}\!\bra{\psi}
                +p\ket{\phi}\!\bra{\phi}.
            \end{align}
            These states lie in the two-dimensional subspace spanned by $\ket{\psi}$ and $\ket{\phi}$, so all remaining steps of that proof apply unchanged, including the bound for arbitrary collective measurements.
            We therefore obtain the following lower bound.

            \begin{thm}[Sample optimality for multi-qubit states] \label{thm:multiqubit-sample-optimality}
                Suppose $0<r\leqq1/8$ and $0<\delta<1/8$.
                For any $m$-qubit target state $\ket{\psi}$ and every sufficiently small $\epsilon>0$, any protocol that estimates the infidelity $\epsilon:=1-\bra{\psi}\rho\ket{\psi}$ of an arbitrary $m$-qubit state $\rho$ to relative error $r$ with probability at least $1-\delta$ requires, in the worst case,
                \begin{align}
                    N
                    =
                    \Omega\!\left(
                        \frac{1}{r^2\epsilon}
                        \log\frac{1}{\delta}
                    \right)
                \end{align}
                copies, even when arbitrary collective measurements are allowed.
            \end{thm}

            The optimality asserted here concerns the dependence on $r$ and $\epsilon$, together with the logarithmic dependence on $\delta$, at fixed $m$ and $|\psi\rangle$.
            Theorem \ref{thm:multi-qubit-extension} therefore matches the lower bound in its dependence on $r$ and $\epsilon$.

        \subsection{Numerical simulations} \label{append:multi_qubit_generalization_simulaton}
            We assess two finite copy consequences of Theorem \ref{thm:multi-qubit-extension}.
            Figure~\ref{fig:multiqubit-simulation-components} resolves the Bell and squared-distance terms at a representative copy budget before combining them into the infidelity estimate.
            Figure~\ref{fig:multiqubit-scaling} then varies the copy budget and the true infidelity in separate sweeps.
            The first sweep tests the $N^{-1/2}$ decrease of the absolute root mean squared error (RMSE) at fixed $\epsilon$.
            The second tests the predicted $\epsilon^{-1/2}$ and $\epsilon^{-1}$ dependence of the relative RMSE at fixed $N$.

            The family of target states consists of the three-qubit states
            \begin{align}
                \begin{aligned}
                    |\psi_\theta\rangle
                    &=
                    \operatorname{CZ}_{2,3}\operatorname{CZ}_{1,2}
                    \\
                    &\quad\times
                    \bigotimes_{j=1}^{3}
                    R_Z(\alpha_j\theta)|+\rangle_j,
                    \\
                    (\alpha_1,\alpha_2,\alpha_3)
                    &=
                    (1,1.4,1.8).
                \end{aligned}
                \label{eq:multiqubit-simulation-target}
            \end{align}
            The preparation noise consists of the channels
            \begin{align}
                \begin{aligned}
                    \mathcal D_{p_1}^{(1)}(\rho)
                    &:=
                    (1-p_1)\rho
                    +
                    \frac{p_1}{3}
                    \sum_{P\in\{X,Y,Z\}}
                    P\rho P,
                    \\
                    \mathcal D_{p_2}^{(2)}(\rho)
                    &:=
                    (1-p_2)\rho
                    +
                    \frac{p_2}{15}
                    \sum_{P\in\mathcal P_2\setminus\{II\}}
                    P\rho P.
                \end{aligned}
                \label{eq:multiqubit-simulation-noise}
            \end{align}
            The first channel follows every Hadamard and $R_Z$ gate.
            The second follows each controlled-$Z$ gate, and every $P$ in its sum acts on the two qubits of that gate.
            In this simulation, the Bell circuit, Bell readout, and Pauli measurements are ideal, so Eq.\,\eqref{eq:multiqubit-simulation-noise} is the only source of noise.
            See Fig.\,\ref{fig:multiqubit-simulation-circuits} for the preparation and measurement circuits.
            Flammia--Liu DFE is applied to the same noisy density matrix with ideal Pauli measurements.

            \begin{figure*}[t]
                \centering
                \begin{subfigure}[c]{0.60\textwidth}
                    \centering
                    \begin{quantikz}[column sep=0.35cm,row sep=0.25cm]
                        \lstick{$\ket{0}$} & \gate{H} & \gate{R_Z(\alpha_1\theta)} & \ctrl{1} & \qw & \ctrl{3} & \qw & \qw & \gate{H} & \meter{$Z$} \\
                        \lstick{$\ket{0}$} & \gate{H} & \gate{R_Z(\alpha_2\theta)} & \control{} & \ctrl{1} & \qw & \ctrl{3} & \qw & \gate{H} & \meter{$Z$} \\
                        \lstick{$\ket{0}$} & \gate{H} & \gate{R_Z(\alpha_3\theta)} & \qw & \control{} & \qw & \qw & \ctrl{3} & \gate{H} & \meter{$Z$} \\
                        \lstick{$\ket{0}$} & \gate{H} & \gate{R_Z(\alpha_1\theta)} & \ctrl{1} & \qw & \targ{} & \qw & \qw & \qw & \meter{$Z$} \\
                        \lstick{$\ket{0}$} & \gate{H} & \gate{R_Z(\alpha_2\theta)} & \control{} & \ctrl{1} & \qw & \targ{} & \qw & \qw & \meter{$Z$} \\
                        \lstick{$\ket{0}$} & \gate{H} & \gate{R_Z(\alpha_3\theta)} & \qw & \control{} & \qw & \qw & \targ{} & \qw & \meter{$Z$}
                    \end{quantikz}
                    \caption{Bell part.}
                \end{subfigure}
                \hfill
                \begin{subfigure}[c]{0.37\textwidth}
                    \centering
                    \begin{quantikz}[column sep=0.35cm,row sep=0.25cm]
                        \lstick{$\ket{0}$} & \gate{H} & \gate{R_Z(\alpha_1\theta)} & \ctrl{1} & \qw & \meter{$P_1$} \\
                        \lstick{$\ket{0}$} & \gate{H} & \gate{R_Z(\alpha_2\theta)} & \control{} & \ctrl{1} & \meter{$P_2$} \\
                        \lstick{$\ket{0}$} & \gate{H} & \gate{R_Z(\alpha_3\theta)} & \qw & \control{} & \meter{$P_3$}
                    \end{quantikz}
                    \caption{Distance part.}
                \end{subfigure}
                \caption{Circuits used for the Bell and distance parts in the three-qubit simulations.
                Both panels include the preparation in eq.\,\eqref{eq:multiqubit-simulation-target}.
                During preparation, $\mathcal{D}_{p_1}^{(1)}$ follows every Hadamard and $R_Z$ gate, and $\mathcal{D}_{p_2}^{(2)}$ follows each controlled-$Z$ gate, as defined in eq.\,\eqref{eq:multiqubit-simulation-noise}.
                Subsequent operations are ideal.
                In panel (a), the first three wires prepare one copy of $\rho$ and the last three prepare the other.
                The CNOTs and Hadamard gates implement Bell measurements between corresponding qubits.
                A singlet is detected when both qubits of a pair have outcome $Z=-1$, and an odd number of singlet outcomes determines the antisymmetric event used to estimate $b$.
                Panel (b) measures a Pauli string $P_1\otimes P_2\otimes P_3$, with identity factors requiring no measurement, to estimate the squared-distance term.
                DFE uses the same preparation and ideal Pauli measurements.}
                \label{fig:multiqubit-simulation-circuits}
            \end{figure*}
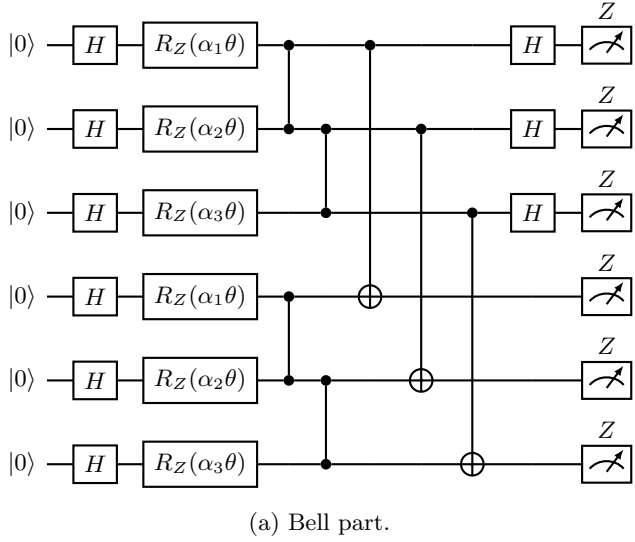
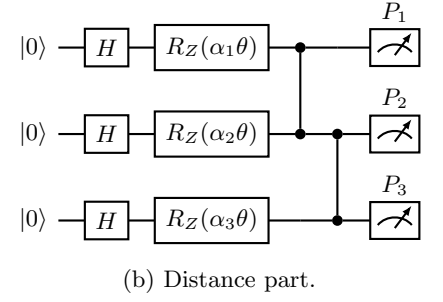

            Each estimation block uses $N_B$ shots for the Bell part and $N_P$ measurements of every nonidentity three-qubit Pauli string.
            Since there are $4^3-1=63$ such strings, $L$ blocks consume
            \begin{align}
                N
                =
                L\left(
                    2N_B+63N_P
                \right)
                \label{eq:multiqubit-simulation-copy-count}
            \end{align}
            resource states at each angle.
            The factor of two accounts for the two copies consumed by one trial of the Bell part.
            The expected number of copies consumed by DFE is chosen to be as close as possible to \(N\) without exceeding it.

            For the representative calculation in Fig.~\ref{fig:multiqubit-simulation-components}, we set
            \begin{align}
                \begin{aligned}
                    p_1
                    &=
                    p_2
                    =
                    0.005,
                    \qquad
                    L
                    =
                    3,
                    \\
                    N_B
                    &=
                    120{,}000,
                    \qquad
                    N_P
                    =
                    12{,}000,
                    \\
                    N
                    &=
                    2{,}988{,}000.
                \end{aligned}
            \end{align}
            We simulate $\theta/\pi\in\{0.01,0.02,0.03,0.04,0.05\}$ with 30 independent random seeds.
            The exact infidelity and its two contributions, calculated directly from the prepared density matrix, are
            \begin{align}
                \begin{aligned}
                    \epsilon
                    &=
                    2.824468\times10^{-2},
                    \\
                    b
                    &=
                    2.774559\times10^{-2},
                    \\
                    h
                    &=
                    9.981671\times10^{-4}
                \end{aligned}
            \end{align}
            at all five angles.
            These values satisfy $b+h/2=\epsilon$ to numerical precision.

            \begin{figure*}[t]
                \centering
                \begin{subfigure}[t]{0.32\textwidth}
                    \centering
                    \includegraphics[width=\linewidth]{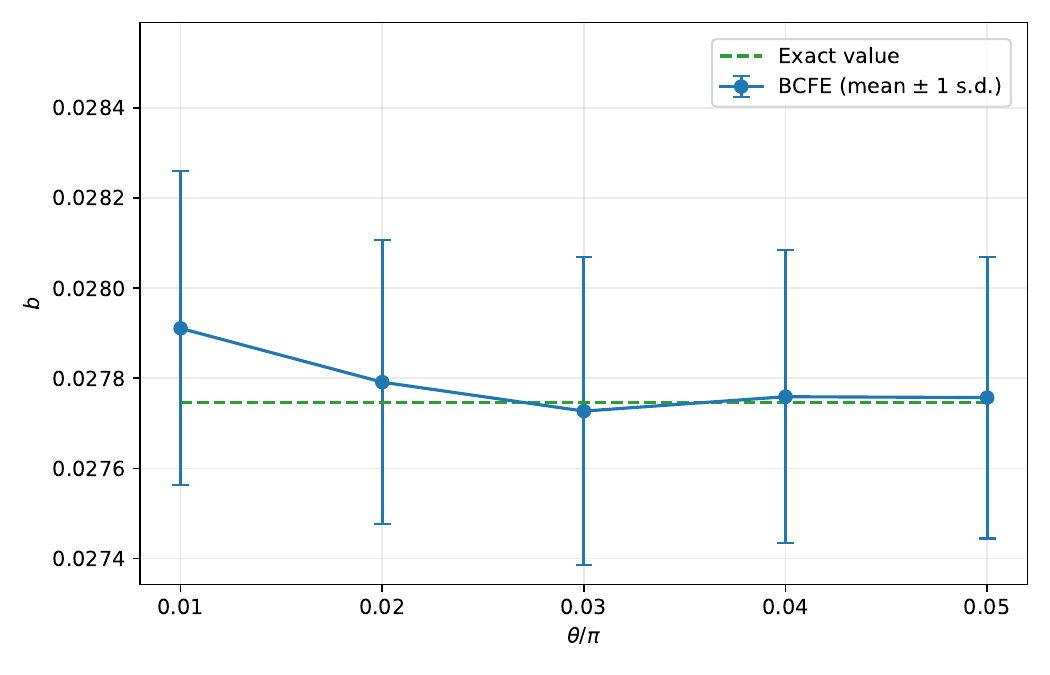}
                    \caption{Bell estimate of $b$.}
                \end{subfigure}
                \hfill
                \begin{subfigure}[t]{0.32\textwidth}
                    \centering
                    \includegraphics[width=\linewidth]{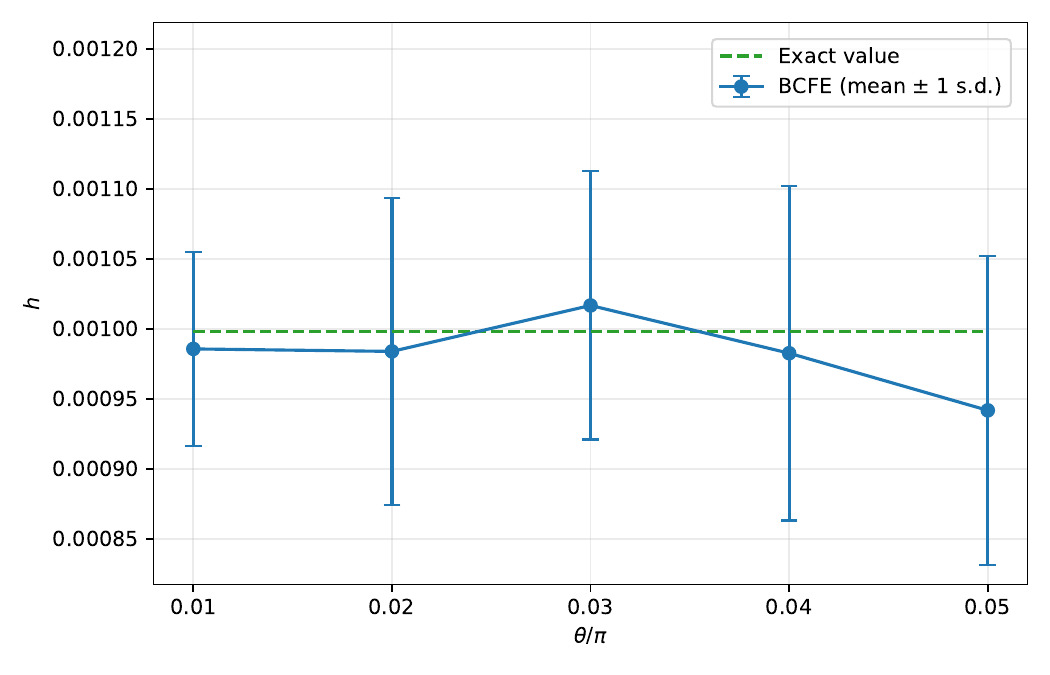}
                    \caption{Pauli estimate of $h$.}
                \end{subfigure}
                \hfill
                \begin{subfigure}[t]{0.32\textwidth}
                    \centering
                    \includegraphics[width=\linewidth]{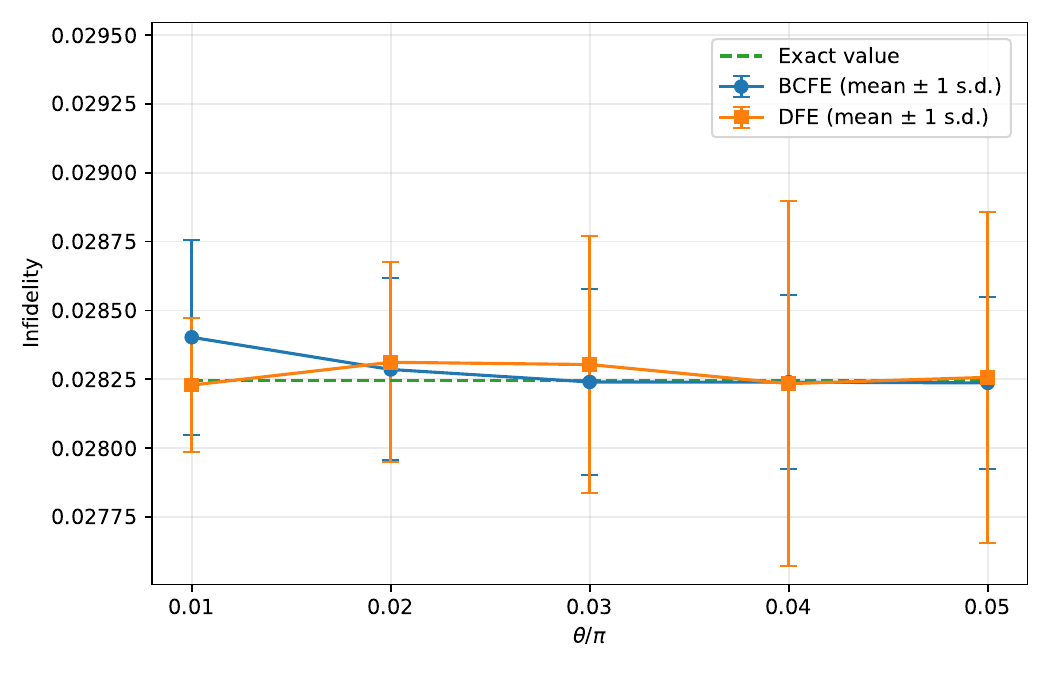}
                    \caption{Combined estimate and DFE.}
                \end{subfigure}
                \caption{Three-qubit estimates at $p_1=p_2=0.005$ and $N=2{,}988{,}000$.
                Each point is the mean over 30 independent seeds, and each error bar is one sample standard deviation across those seeds.
                Panels (a) and (b) resolve the two terms in Eq.\,\eqref{eq:multiqubit-direct-infidelity}.
                Panel (c) combines them and compares the resulting infidelity estimate with Flammia--Liu DFE at comparable expected numbers of prepared resource states.
                The dashed green curves are the exact values.}
                \label{fig:multiqubit-simulation-components}
            \end{figure*}

            The seed averages in panels (a) and (b) track the exact component values at all five angles.
            After the two components are combined, our protocol has an RMSE of $3.33\times10^{-4}$, compared with $8.34\times10^{-4}$ for DFE when the expected numbers of prepared resource states are comparable.

            We next compare the scaling with $N$.
            The four shot allocations are
            \begin{align}
                \begin{aligned}
                    (N_B,N_P)
                    \in
                    \{&
                    (12{,}000,1{,}200),
                    (36{,}000,3{,}600),
                    \\
                    &
                    (72{,}000,7{,}200),
                    (120{,}000,12{,}000)
                    \},
                \end{aligned}
                \label{eq:multiqubit-simulation-shot-allocations}
            \end{align}
            corresponding to $N=298{,}800$, $896{,}400$, $1{,}792{,}800$, and $2{,}988{,}000$.
            For either method $M$, the root mean squared error at a given copy budget is
            \begin{align}
                \operatorname{RMSE}_M(N)
                :=
                \left(
                    \frac{1}{150}
                    \sum_{s=1}^{30}
                    \sum_{k=1}^{5}
                    \big(
                        \widehat\epsilon_{M,s}(\theta_k)
                        -
                        \epsilon(\theta_k)
                    \big)^2
                \right)^{1/2}.
                \label{eq:multiqubit-simulation-rmse}
            \end{align}
            The 150 estimates comprise 30 seeds for each of the five target angles.

            To isolate the dependence on the infidelity, we instead fix $\theta/\pi=0.03$ and $N=2{,}988{,}000$, and take $p_1=p_2=p$ with
            \begin{align}
                \begin{aligned}
                    p
                    \in
                    \{&
                    10^{-4},
                    2\times10^{-4},
                    5\times10^{-4},
                    10^{-3},
                    \\
                    &
                    2\times10^{-3},
                    5\times10^{-3},
                    8\times10^{-3}
                    \}.
                \end{aligned}
            \end{align}
            The smallest and largest exact infidelities are $5.731631\times10^{-4}$ and $4.479214\times10^{-2}$, respectively.
            We quantify the error relative to the true infidelity through the relative root mean squared error (RRMSE)
            \begin{align}
                \operatorname{RRMSE}_M(\epsilon)
                :=
                \frac{1}{\epsilon}
                \Bigg(
                    \frac{1}{30}
                    \sum_{s=1}^{30}
                    \big(
                        \widehat\epsilon_{M,s}
                        -
                        \epsilon
                    \big)^2
                \Bigg)^{1/2}.
                \label{eq:multiqubit-simulation-relative-rmse}
            \end{align}
            For both methods, we calculate the RMSE from the estimates before clipping to $[0,1]$ to compare their statistical scaling without the reduction in error caused by clipping.

            \begin{figure*}[t]
                \centering
                \begin{subfigure}[t]{0.48\textwidth}
                    \centering
                    \includegraphics[width=0.85\linewidth]{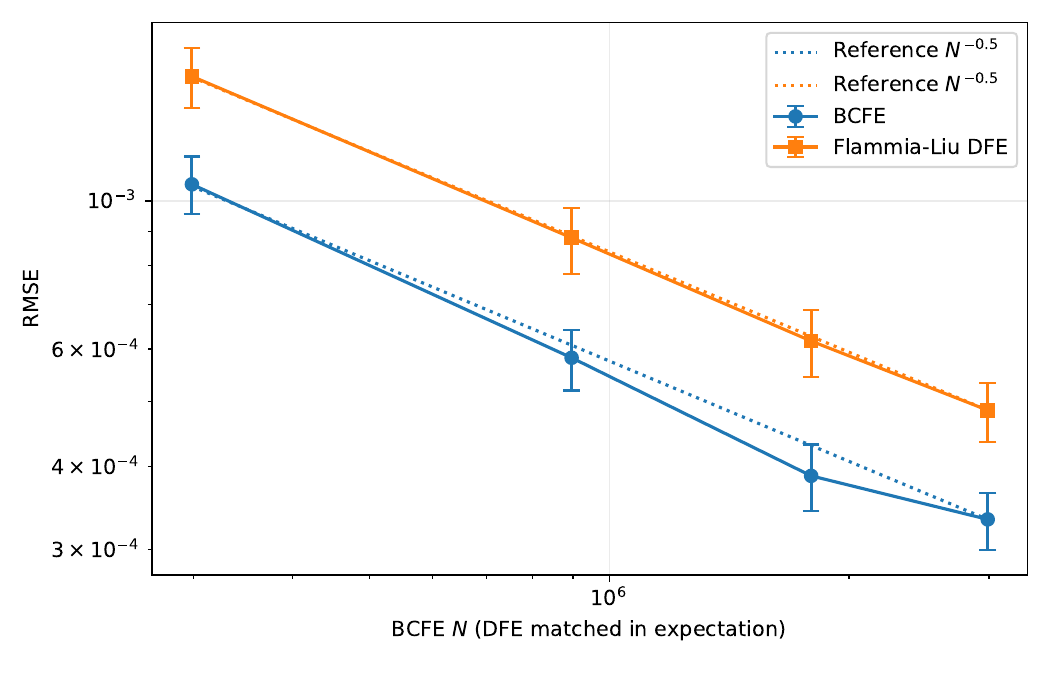}
                    \caption{RMSE as a function of $N$ at $p_1=p_2=0.005$.}
                \end{subfigure}
                \hfill
                \begin{subfigure}[t]{0.48\textwidth}
                    \centering
                    \includegraphics[width=0.85\linewidth]{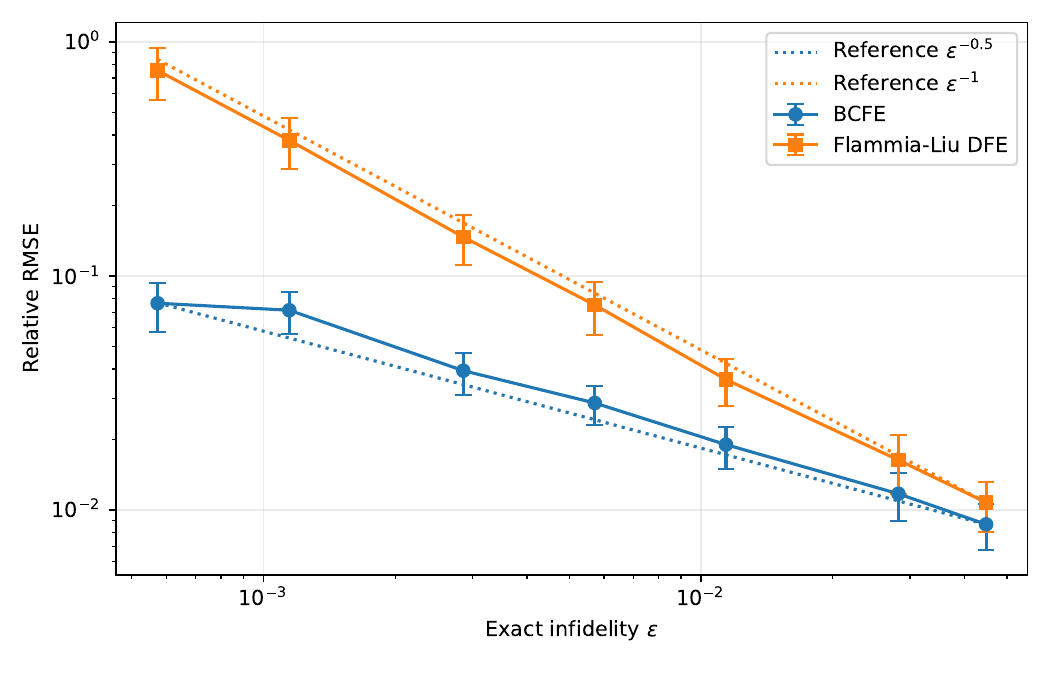}
                    \caption{Relative RMSE as a function of the exact infidelity at fixed $N$.}
                \end{subfigure}
                \caption{Finite-copy scaling of the three-qubit estimates.
                The points in panel (a) use Eq.\,\eqref{eq:multiqubit-simulation-rmse}, while panel (b) uses Eq.\,\eqref{eq:multiqubit-simulation-relative-rmse}.
                The error bars are 95\% bootstrap confidence intervals obtained by resampling the 30 seeds.
                In panel (a), the five angles associated with a selected seed remain together in each bootstrap sample.
                Solid lines connect the simulated values, and the dotted or dashed lines mark the indicated reference powers.}
                \label{fig:multiqubit-scaling}
            \end{figure*}

            In Fig.~\ref{fig:multiqubit-scaling}(a), least-squares fits in logarithmic coordinates have slopes $-0.517$ for our protocol and $-0.561$ for DFE.
            Both values agree with the $N^{-1/2}$ behavior expected at fixed nonzero $\epsilon$.
            Our protocol nevertheless has a smaller RMSE at every simulated budget, with the ratio between the DFE and protocol values ranging from $2.50$ to $2.85$.

            In Fig.~\ref{fig:multiqubit-scaling}(b), the fitted slopes are $-0.520$ for our protocol and $-0.991$ for DFE, close to the respective reference powers $-1/2$ and $-1$.
            The two panels are consistent with
            \begin{align}
                \operatorname{RRMSE}_{\rm{BCFE}}
                \propto
                \frac{1}{\sqrt{N\epsilon}},
                \qquad
                \operatorname{RRMSE}_{\rm{DFE}}
                \propto
                \frac{1}{\sqrt{N}\cdot\epsilon}.
            \end{align}
            Here the two costs diverge sharply as $\epsilon$ shrinks.
            A fixed relative accuracy then requires $N=O(1/\epsilon)$ copies for our protocol against $N=O(1/\epsilon^2)$ for DFE.
            Theorem \ref{thm:multi-qubit-extension} establishes the former dependence for arbitrary multi-qubit target states.
            The finite copy data corroborate the resulting distinction from DFE for the three-qubit family and preparation noise considered here.

    \section{Mini-introduction to the STAR architecture} \label{append:star}
        In this appendix, we provide a self-contained background of the Space-Time Efficient Analog Rotations (STAR) architecture \cite{Akahoshi:2023xck, Toshio:2024hut, Toshio:2026dpa}, which is the prerequisite for the benchmark simulation in Sec.\,\ref{sec:theory-and-simulation}.
        The STAR architecture itself is a highly efficient way of performing FTQC in the early-FTQC era where only $10^3$--$10^5$ qubits are available.
        Below, we elaborate on the motivation and the ideas.
        \paragraph{Motivation.}
        In surface code fault-tolerant quantum computing, logical Clifford gates are relatively inexpensive, whereas each logical $T$ gate consumes a high fidelity magic state produced by magic state distillation.
        An arbitrary rotation must first be approximated by a long Clifford+$T$ word, for example by the Solovay--Kitaev algorithm \cite{Dawson:2005blj}, or more practically by number-theoretic exact synthesis \cite{Ross:2014qpa} and probabilistic mixing of gate sequences \cite{Campbell:2017mix}.
        If the synthesis accuracy is $\epsilon_{\rm syn}$, the resulting cost has the schematic scaling
        \begin{align}
            R_P(\theta)
            \approx
            C_0T C_1T\cdots T C_{N_T},
            \quad
            N_T
            =
            O\!\left(\log^c\frac{1}{\epsilon_{\rm syn}}\right).
        \end{align}
        Here $c>1$ is the compilation exponent.
        The cost of a single accurate rotation is therefore paid twice, through both the synthesis length and the repeated production of distilled magic states.

        \paragraph{A concrete compilation cost.}
        The preceding overhead is visible even for the explicit rotation $e^{iZ\pi/128}$.
        Ross--Selinger synthesis with target operator-norm error $10^{-10}$ produces the Clifford+$T$ word \cite{Ross:2014qpa}
        \begin{widetext}
            \begin{align}
                \begin{aligned}
                    e^{iZ\pi/128} \approx \;&
                    H{\color{red}T}SH{\color{red}T}SH{\color{red}T}SH{\color{red}T}H{\color{red}T}SH{\color{red}T}SH{\color{red}T}SH{\color{red}T}H{\color{red}T}SH{\color{red}T}H{\color{red}T}H{\color{red}T}H{\color{red}T}SH{\color{red}T}SH{\color{red}T}SH{\color{red}T}H{\color{red}T}H\\[0.4em]
                    &{\color{red}T}H{\color{red}T}SH{\color{red}T}H{\color{red}T}SH{\color{red}T}SH{\color{red}T}SH{\color{red}T}H{\color{red}T}H{\color{red}T}H{\color{red}T}SH{\color{red}T}H{\color{red}T}H{\color{red}T}SH{\color{red}T}H{\color{red}T}H{\color{red}T}H{\color{red}T}SH{\color{red}T}SH{\color{red}T}\\[0.45em]
                    &H{\color{red}T}SH{\color{red}T}H{\color{red}T}H{\color{red}T}H{\color{red}T}H{\color{red}T}H{\color{red}T}SH{\color{red}T}H{\color{red}T}SH{\color{red}T}SH{\color{red}T}SH{\color{red}T}H{\color{red}T}SH{\color{red}T}SH{\color{red}T}SH{\color{red}T}SH{\color{red}T}H{\color{red}T}\\[0.45em]
                    &SH{\color{red}T}SH{\color{red}T}H{\color{red}T}H{\color{red}T}H{\color{red}T}H{\color{red}T}SH{\color{red}T}SH{\color{red}T}SH{\color{red}T}H{\color{red}T}SH{\color{red}T}H{\color{red}T}H{\color{red}T}H{\color{red}T}SH{\color{red}T}SH{\color{red}T}SH{\color{red}T}SH\\[0.45em]
                    &{\color{red}T}SH{\color{red}T}H{\color{red}T}SH{\color{red}T}SH{\color{red}T}H{\color{red}T}H{\color{red}T}SH{\color{red}T}H{\color{red}T}H{\color{red}T}SH{\color{red}T}H{\color{red}T}H{\color{red}T}H{\color{red}T}SH{\color{red}T}SH{\color{red}T}SH{\color{red}T}H{\color{red}T}SH\\[0.45em]
                    &{\color{red}T}H{\color{red}T}H{\color{red}T}SH{\color{red}T}H{\color{red}T}H{\color{red}T}SH{\color{red}T}H{\color{red}T}H{\color{red}T}SH{\color{red}T}H{\color{red}T}H{\color{red}T}SH{\color{red}T}SH{\color{red}T}H{\color{red}T}SH{\color{red}T}H{\color{red}T}SH{\color{red}T}SH{\color{red}T}.
                \end{aligned}
            \end{align}
        \end{widetext}
        Direct multiplication of the standard $H$, $S$, and $T$ matrices produces an operator-norm error of $1.85\times10^{-11}<10^{-10}$.
        Even this single rotation consumes 108 distilled magic states, one for each highlighted $T$ gate.

        As far as early-FTQC devices are concerned, this restriction remains acute, since error-corrected logical qubits may be available while the number of physical qubits is still too small to devote large factories to a sustained supply of distilled states \cite{Akahoshi:2023xck, Toshio:2024hut}.
        The STAR architecture circumvents this overhead by implementing arbitrary rotations directly from angle-dependent resource states, at the cost of a controlled residual rotation error.
        In compiled Trotter circuits, parallel resource state injection and adaptive updates of its target regions reduce the runtime overhead caused by injection and RUS \cite{Akahoshi:2024yme}.

        \paragraph{STAR ver.\,2.}
        The STAR architecture changes where the computational burden falls: Clifford operations remain fault-tolerant, while non-Clifford rotations are implemented directly and noisily \cite{Akahoshi:2023xck, Toshio:2024hut}.
        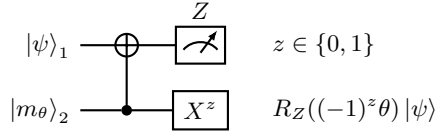
\begin{figure*}[t]
            \centering
            \begin{quantikz}[column sep=0.45cm, row sep=0.35cm]
                \lstick{$\ket{\psi}_1$}
                & \targ{}
                & \meter{$Z$}
                & \rstick{$z\in\{0,1\}$} \\
                \lstick{$\ket{m_\theta}_2$}
                & \ctrl{-1}
                & \gate{X^z}
                & \rstick{$R_Z((-1)^z\theta)\ket{\psi}$}
            \end{quantikz}
            \caption{Gate teleportation of an analog $Z$ rotation.
            The resource state controls the CNOT, and the input qubit is measured in the $Z$ basis.
            The correction $X^z$ is applied to the resource qubit, which carries the output state after the measurement.}
            \label{fig:star-gate-teleportation}
        \end{figure*}
        For a logical Pauli operator $P$, its basic resources are
        \begin{align}
            \ket{m_\theta}_L:=
                R_Z(\theta)\ket{+}_L, \quad
            R_P(\theta):=
            e^{i\theta P},
        \end{align}
        as in Eq.\,\eqref{eq:analog-rotation-resource-state}.
        The analog rotation is performed by the gate teleportation
        \begin{align}
            \begin{aligned}
                &\CNOT_{2\to1}
                \ket{\psi}_1\otimes\ket{m_\theta}_2
                \\
                &\qquad=
                \frac{1}{\sqrt{2}}
                \big(
                    \ket{0}_1\otimes R_Z(\theta)\ket{\psi}_2
                    +
                    \ket{1}_1\otimes XR_Z(-\theta)\ket{\psi}_2
                \big).
            \end{aligned}
        \end{align}
        The outcome $Z=+1$ leaves the desired state $R_Z(\theta)\ket{\psi}$ on the second qubit with probability $\frac{1}{2}$.
        For $Z=-1$, applying $X$ to that qubit results in $R_Z(-\theta)\ket{\psi}$.
        The first successful histories of the repeat-until-success (RUS) protocol are
        \begin{align}
            R_P(\theta)
            &=
            R_P(\theta),
            \notag\\
            R_P(-\theta)R_P(2\theta)
            &=
            R_P(\theta),
            \notag\\
            R_P(-\theta)R_P(-2\theta)R_P(4\theta)
            &=
            R_P(\theta),
            \notag\\[-0.3em]
            &\ \vdots
            \notag\\
            \prod_{k=0}^{j-1}R_P(-2^k\theta)\,
            R_P(2^j\theta)
            &=
            R_P(\theta).
        \end{align}
        Trial $j$ is reached only after $j$ consecutive failures and hence with probability $2^{-j}$.
        The expected number of teleportation attempts is thus
        \begin{align}
            \mathbb{E}[N_{\rm{RUS}}]
            =
            \sum_{j=0}^{\infty}2^{-j}
            =
            2.
        \end{align}

        The protocol proposed by Choi \textit{et al}.\ \cite{Choi:2023bvh} prepares $\ket{m_{\theta_L}}_L$ by physical rotations followed by stabilizer post-selection.
        STAR ver.\,2 implements this construction as transversal multi-rotation (TMR) \cite{Toshio:2024hut}.
        See Ref.\,\cite{McArdle:2026rvn} for a technique that improves the angle dependence of the RUS error rate without a digital synthesis fallback.
        Although originally formulated for superconducting qubits, TMR is not intrinsically hardware-specific and has an analog for neutral atoms \cite{Ismail:2025fry}.
        For a distance-$d$ logical path $\gamma$, the physical rotations produce
        \begin{align}
            \begin{aligned}
                \prod_{v\in\gamma}e^{i\theta Z_v}\ket*{+_L}
                &=
                \cos^d\!\theta\ket*{+_L}
                +i^d\sin^d\!\theta\ket*{-_L}
                \\
                &\quad
                +(\text{states outside the logical subspace}).
            \end{aligned}
        \end{align}
        Stabilizer measurements implement the projection $\Pi_L$ onto the logical subspace.
        When $d$ is odd, the projected component determines the accepted state and its logical angle,
        \begin{align}
            \begin{gathered}
                \Pi_L\prod_{v\in \gamma}e^{i\theta Z_v}\ket*{+_L}
                =
                \cos^d\!\theta\ket*{+_L}
                +i^d\sin^d\!\theta\ket*{-_L}
                \!\propto\!
                \ket*{m_{\theta_L}}_L,
                \\
                \theta_L
                =
                \arctan\!\big(
                    (-1)^{(d-1)/2}\tan^d\theta
                \big).
            \end{gathered}
        \end{align}
        The stabilizer outcomes are accepted precisely when the state lies in the logical subspace, so the post-selection probability is $\cos^{2d}\!\theta+\sin^{2d}\!\theta$.
        For $d\geqq3$, the small angle expansion is
        \begin{align}
            p_{\rm succ}
            =
            1-d\theta^2+O(\theta^4).
        \end{align}
        The post-selection therefore succeeds with probability close to one when the physical rotation angle is small.
        Post-selection removes components outside the logical subspace, but it cannot screen out every physical fault.
        A fault may still return an erroneous component to the accepted logical subspace.
        The resulting state $\rho_{\theta_L}$ is therefore characterized by the quantity studied in this paper,
        \begin{align}
            F_{\theta_L}:=
            \bra{m_{\theta_L}}\rho_{\theta_L}\!\ket{m_{\theta_L}}.
        \end{align}
        Let $p_{\rm ph}$ be the physical error rate and $P_{\rm ud}=O(p_{\rm ph})$ the probability that a physical fault remains undetected in an accepted TMR preparation.
        For the single-qubit TMR above, the faulty branch occurs with probability $P_L=O(|\theta_L|^{2/d}P_{\rm ud})$ and has an over-rotation angle $\Delta\theta_L=O(|\theta_L|^{1-2/d})$.
        PCEC combines over rotations of opposite signs, leaving a stochastic Pauli-$Z$ error with probability $2P_L\sin^2(\Delta\theta_L)$ \cite{Toshio:2024hut}.
        The worst-case channel error $\epsilon_\diamond$ of one teleportation attempt after PCEC therefore scales as
        \begin{align}
            \epsilon_\diamond
            =
            O\!\left(|\theta_L|^{2-2/d}p_{\rm ph}\right).
        \end{align}
        Each failed teleportation in RUS doubles the angle.
        A history that reaches $2^{K-1}|\theta_L|=O(1)$ occurs with probability $2^{-K}=O(|\theta_L|)$, while its accumulated error is $O(p_{\rm ph})$.
        Averaged over all RUS histories, the channel error becomes
        \begin{align}
            \epsilon_\diamond
            =O\!\left(|\theta_L|p_{\rm ph}\right).
        \end{align}
        In ver.\,2, TMR prepares the resource states used in the small angle regime, whereas a $[[4,1,1,2]]$ subsystem-code injection handles the large angle branches reached during RUS.\footnote{STAR ver.\,1 used the subsystem-code injection throughout, but its low injection efficiency motivated the introduction of TMR in ver.\,2 \cite{Akahoshi:2023xck, Toshio:2024hut}.}

        \paragraph{Relation to STAR ver.\,3.}
        STAR ver.\,3 changes the treatment of large angles, but it does not displace the fidelity problem studied here.
        STAR-magic mutation retains the same TMR resource state in the small angle regime but replaces the large angle RUS tail by magic-state-assisted digital synthesis \cite{Toshio:2026dpa}
        \begin{align*}
            \begin{aligned}
                R_P(\theta_{\rm RUS})
                &:\quad
                \begin{cases}
                    \begin{array}{@{}l@{}}
                        \text{TMR + higher-order}\\
                        \text{PCEC},
                    \end{array}
                        & |\theta_{\rm RUS}|<\theta_{\rm th},
                        \\[0.2em]
                    \begin{array}{@{}l@{}}
                        \text{Clifford+$T$ synthesis}\\
                        \text{with MSC or MSD},
                    \end{array}
                        & |\theta_{\rm RUS}|\geqq\theta_{\rm th}.
                \end{cases}
            \end{aligned}
        \end{align*}
        Thus ver.\,3 changes how the rare large angle regime is handled, but the fidelity of the arbitrary angle resource state prepared by TMR remains relevant to the small angle branch of ver.\,3.

    \section{Flammia--Liu direct fidelity estimation}
    \label{append:flammia_liu}
        In this appendix, we review the direct fidelity estimation protocol (DFE) \cite{Flammia:2011fda} that is used for comparison to our BCFE protocol in Sec.\,\ref{sec:theory-and-simulation}.
        The goal of DFE is to estimate the target fidelity directly from Pauli measurements, without reconstructing the prepared state, and to determine the number of copies required for a prescribed additive error.
        The theoretical structure and finite sample analysis are given in Appendix \ref{append:dfe-protocol}.
        Appendix \ref{append:dfe-simulation-setup} specifies the estimator, error bars, and calibration of the expected resource consumption used in the single-qubit simulations in Sec.\,\ref{sec:theory-and-simulation}.

        \subsection{The DFE protocol} \label{append:dfe-protocol}
            Suppose $\tau:=\ket*{\psi}\!\bra*{\psi}$ is an $n$-qubit target pure state and $\rho$ is the prepared, possibly noisy state, which would ideally coincide with $\tau$.
            For $D:=2^n$, let $\{W_k\}_{k=1}^{D^2}=\{I,X,Y,Z\}^{\otimes n}$ be the Pauli basis, normalized by
            \begin{align}
                \operatorname{tr}(W_kW_l)=D\delta_{kl}.
            \end{align}
            For any operator $\sigma$, we define the quantity
            \begin{align}
                \chi_{\sigma}(k):=\frac{1}{\sqrt{D}}\tr(\sigma W_k),
            \end{align}
            so that, by completeness of the Pauli basis, $\sigma$ admits the expansion
            \begin{align}
                \sigma
                =
                \frac{1}{D}\sum_{k=1}^{D^2}\operatorname{tr}(\sigma W_k)W_k
                =
                \frac{1}{\sqrt{D}}\sum_{k=1}^{D^2}\chi_\sigma(k)W_k.
            \end{align}
            In terms of these coefficients, the fidelity of the prepared state $\rho$ is given by
            \begin{align}
                F
                :=
                \bra*{\psi}\rho\ket*{\psi}
                =
                \operatorname{tr}(\tau\rho)
                =
                \sum_{k=1}^{D^2}\chi_\tau(k)\chi_\rho(k).
                \label{eq:dfe-pauli-fidelity}
            \end{align}
            Since $\tau$ is pure, $\sum_k\chi_\tau(k)^2=\operatorname{tr}\tau^2=1$, so
            \begin{align}
                p(k):=\chi_\tau(k)^2,
                \label{eq:dfe-relevance-distribution}
            \end{align}
            is a probability distribution.
            For the random variable $K\sim p$, let us define the quantity
            \begin{align}
                X(K)
                :=
                \frac{\chi_\rho(K)}{\chi_\tau(K)}
                =
                \frac{\operatorname{tr}(\rho W_K)}
                {\operatorname{tr}(\tau W_K)}.
                \label{eq:dfe-exact-random-variable}
            \end{align}
            Note that the indices with zero probability $\chi_\tau(k)=0$ are excluded from the ratio.
            Now, Eqs.\,\eqref{eq:dfe-pauli-fidelity} and \eqref{eq:dfe-relevance-distribution} imply
            \begin{align}
                \mathbb E[X]
                =
                F,
                \qquad
                \mathbb E[X^2]
                =
                \sum_{k:\,p(k)>0}\chi_\rho(k)^2
                \leqq
                \operatorname{tr}\rho^2
                \leqq
                1.
                \label{eq:dfe-first-two-moments}
            \end{align}
            We therefore estimate the fidelity through the mean of $X$.
            Its denominator $s_k:=\operatorname{tr}(\tau W_k)$ depends only on the target state $\tau=\ket{\psi}\!\bra{\psi}$ and can be computed before the simulation begins, so it suffices to estimate only the numerator $\operatorname{tr}(\rho W_k)$ from Pauli measurements.
            The advantage of DFE stems from drawing each Pauli string $W_k$ with probability $p(k)=\chi_\tau(k)^2$ rather than allocating the same number of measurements to every Pauli string.

            The actual finite sample estimation is carried out as follows.
            The sample mean of $X$ over $L$ independent draws $K_1,\ldots,K_L\sim p$ is $\overline X_L:=L^{-1}\sum_{i=1}^{L}X(K_i)$.
            The first moment in Eq.\,\eqref{eq:dfe-first-two-moments} fixes the expectation $\mathbb{E}[\overline X_L]=F$, and the second moment bounds the variance as
            \begin{align}
                \operatorname{Var}(\overline X_L)
                =
                \frac{\operatorname{Var}(X)}{L}
                =
                \frac{\mathbb{E}[X^2]-F^2}{L}
                \leqq
                \frac{1}{L}.
            \end{align}
            By Chebyshev's inequality (Theorem \ref{thm:chebyshev-inequality}), for any $\eta>0$,
            \begin{align}
                \Pr\!\left[
                    |\overline X_L-F|\geqq\eta
                \right]
                \leqq
                \frac{\operatorname{Var}(\overline X_L)}{\eta^2}
                \leqq
                \frac{1}{L\eta^2}.
                \label{eq:dfe-chebyshev-bound}
            \end{align}
            For $L=\lceil1/(\eta^2\delta)\rceil$, this probability is at most $\delta$.

            For each draw $K_i$, the numerator $\operatorname{tr}(\rho W_{K_i})$ is estimated by measuring $W_{K_i}$ on $m_i$ independently prepared copies of $\rho$.
            At this point, we write $A_{ij}\in\{+1,-1\}$ for the $j$th outcome and use the repetition counts and estimators
            \begin{align}
                \begin{aligned}
                    m_i
                    &:=
                    \left\lceil
                        \frac{2}{L\eta^2s_{K_i}^2}
                        \log\frac{2}{\delta}
                    \right\rceil,
                    \\
                    \widetilde X_i
                    &:=
                    \frac{1}{m_is_{K_i}}
                    \sum_{j=1}^{m_i}A_{ij},
                    \qquad
                    \widehat F_{\rm DFE}
                    :=
                    \frac{1}{L}\sum_{i=1}^{L}\widetilde X_i.
                \end{aligned}
                \label{eq:dfe-finite-shot-estimator}
            \end{align}
            For every $k$ with $p(k)>0$, the conditional means of $A_{ij}$ and $\widetilde{X}_i$ are
            \begin{align}
                \begin{aligned}
                    \mathbb{E}[A_{ij}|K_i=k]
                    &=
                    \operatorname{tr}(\rho W_k),
                    \\
                    \mathbb{E}[\widetilde X_i|K_i=k]
                    &=
                    \frac{1}{m_is_k}
                    \sum_{j=1}^{m_i}
                    \mathbb{E}[A_{ij}|K_i=k]
                    \\
                    &=
                    \frac{\operatorname{tr}(\rho W_k)}
                    {\operatorname{tr}(\tau W_k)}
                    =
                    X(k).
                \end{aligned}
            \end{align}
            Averaging over the draw $K_i$,
            \begin{align}
                \mathbb{E}[\widetilde X_i]
                =
                \sum_{k:\,p(k)>0}p(k)X(k)
                =
                F.
            \end{align}
            With $m_i$ chosen as in Eq.\,\eqref{eq:dfe-finite-shot-estimator}, we apply Hoeffding's inequality (Theorem \ref{thm:hoeffding-inequality}) for fixed $K_1,\ldots,K_L$ and get
            \begin{align}
                \Pr\!\left[
                    |\widehat F_{\rm DFE}-\overline X_L|
                    \geqq\eta
                    \,\middle|\,
                    K_1,\ldots,K_L
                \right]
                \leqq
                \delta.
            \end{align}
            This bound is uniform in $K_1,\ldots,K_L$ and hence also holds without conditioning,
            \begin{align}
                \Pr\!\left[
                    |\widehat F_{\rm DFE}-\overline X_L|
                    \geqq\eta
                \right]
                \leqq
                \delta.
                \label{eq:dfe-unconditional-measurement-bound}
            \end{align}
            By the triangle inequality, the total error is bounded by
            \begin{align}
                |\widehat F_{\rm DFE}-F|
                \leqq
                |\widehat F_{\rm DFE}-\overline X_L|
                +
                |\overline X_L-F|.
            \end{align}
            The condition $|\widehat F_{\rm DFE}-F|>2\eta$ therefore requires at least one of the two terms on the r.\,h.\,s.\ to be larger than $\eta$.
            With the union bound $\Pr(A\cup B)\leqq\Pr(A)+\Pr(B)$, its probability satisfies
            \begin{align}
                \begin{aligned}
                    \Pr\!\left[
                        |\widehat F_{\rm DFE}-F|>2\eta
                    \right]
                    &\leqq
                    \Pr\!\left[
                        |\widehat F_{\rm DFE}-\overline X_L|>\eta
                    \right]
                    \\
                    &\quad+
                    \Pr\!\left[
                        |\overline X_L-F|>\eta
                    \right]
                    \leqq
                    2\delta.
                \end{aligned}
                \label{eq:dfe-additive-guarantee}
            \end{align}
            Here, we employ Eqs.\,\eqref{eq:dfe-chebyshev-bound} and \eqref{eq:dfe-unconditional-measurement-bound} in the second inequality.
            The number of copies $M:=\sum_i m_i$ is random because each $m_i$ depends on the sampled Pauli string.
            Since $p(k)=s_k^2/D$,
            \begin{align}
                \sum_{k:\,p(k)>0}
                \frac{p(k)}{s_k^2}
                =
                \frac{1}{D}
                \left|
                    \{k:s_k\neq0\}
                \right|
                \leqq
                D,
            \end{align}
            where the last inequality follows because the Pauli basis contains $D^2$ strings.
            Using $\lceil x\rceil\leqq1+x$ and $L=\lceil1/(\eta^2\delta)\rceil$, we obtain
            \begin{align}
                \begin{aligned}
                    \mathbb{E}[M]
                    &\leqq
                    L
                    +
                    \frac{2D}{\eta^2}\log\frac{2}{\delta}
                    \\
                    &\leqq
                    1+
                    \frac{1}{\eta^2\delta}
                    +
                    \frac{2D}{\eta^2}\log\frac{2}{\delta}.
                \end{aligned}
                \label{eq:dfe-expected-copy-bound}
            \end{align}
            For the corresponding infidelity estimator $\widehat{\epsilon}_{\rm DFE}:=1-\widehat F_{\rm DFE}$, Eq.\,\eqref{eq:dfe-additive-guarantee} is equivalent to
            \begin{align}
                \Pr\!\left[
                    |\widehat{\epsilon}_{\rm DFE}-\epsilon|
                    \leqq
                    2\eta
                \right]
                \geqq
                1-2\delta.
            \end{align}
            To guarantee a relative error of at most $r$ with the same success probability, we set $\eta=r\epsilon/2$.
            At this value of $\eta$, the upper bound in Eq.\,\eqref{eq:dfe-expected-copy-bound} is
            \begin{align}
                \mathbb{E}[M]
                \leqq
                1
                +
                \frac{4}{r^2\epsilon^2\delta}
                +
                \frac{8D}{r^2\epsilon^2}
                \log\frac{2}{\delta}.
            \end{align}
            The $1/\epsilon^2$ dependence of this upper bound is tight in the worst case, so the $O(1/(r^2\epsilon))$ sample complexity proved in Sec.\,\ref{sec:general_protocol} is still a genuine improvement in the $\epsilon$ exponent compared to DFE.

            To prove the worst case lower bound, we exhibit one target and two nearby states with disjoint relative-error intervals but Pauli outcome probabilities that differ by only $O(r\epsilon)$.
            Consider the pure single qubit target and the associated family of states
            \begin{align}
                \tau
                :=
                \frac{1}{2}
                \left(
                    I+\frac{X+Y+Z}{\sqrt{3}}
                \right),
                \qquad
                \rho_a
                :=
                (1-a)\tau+a(I-\tau).
            \end{align}
            The infidelity of $\rho_a$ with respect to $\tau$ is $a$.
            For every $P\in\{X,Y,Z\}$, a measurement of $P$ on $\rho_a$ returns $+1$ with probability
            \begin{align}
                u_a
                =
                \frac{1}{2}
                \Bigl(
                    1+\frac{1-2a}{\sqrt{3}}
                \Bigl).
            \end{align}
            We compare $a_0=\epsilon$ and $a_1=(1+4r)\epsilon$, where $0<r\leqq1/8$, $0<\delta<1/4$, and $a_1\leqq1/2$.
            The relative entropy between the outcome distributions of one Pauli measurement satisfies
            \begin{align}
                \begin{aligned}
                    D_{\mathrm{KL}}
                    \left(
                        \operatorname{Bern}(u_{a_0})
                        \middle\Vert
                        \operatorname{Bern}(u_{a_1})
                    \right)
                    &\leqq
                    \frac{(u_{a_0}-u_{a_1})^2}
                    {u_{a_1}(1-u_{a_1})}
                    \\
                    &\leqq
                    32r^2\epsilon^2.
                \end{aligned}
            \end{align}
            Here, $u_{a_1}(1-u_{a_1})\geqq1/6$ and $|u_{a_0}-u_{a_1}|=4r\epsilon/\sqrt{3}$.
            The DFE sampling probabilities depend only on the target and are therefore identical for the two states.
            If $\mathsf{P}_0$ and $\mathsf{P}_1$ denote the complete measurement records, the chain rule for relative entropy yields
            \begin{align}
                D_{\mathrm{KL}}(\mathsf{P}_0\Vert\mathsf{P}_1)
                \leqq
                32r^2\epsilon^2\mathbb{E}_{0}[M].
            \end{align}
            The relative error intervals around $a_0$ and $a_1$ are disjoint because
            \begin{align}
                (1+r)\epsilon
                <
                (1-r)(1+4r)\epsilon.
            \end{align}
            An estimator that succeeds with probability at least $1-\delta$ for both states therefore distinguishes $\rho_{a_0}$ from $\rho_{a_1}$ with average error probability at most $\delta$.
            The Bretagnolle--Huber inequality (Theorem \ref{thm:bretagnolle-huber-inequality}) then requires
            \begin{align}
                \mathbb{E}_{0}[M]
                \geqq
                \frac{1}{32r^2\epsilon^2}
                \log\frac{1}{4\delta}.
            \end{align}
            Thus, already for single qubit targets and fixed confidence, the worst-case expected copy complexity of DFE with the sampling rule and Pauli measurements defined above is $\Theta(1/(r^2\epsilon^2))$.

        \subsection{Relevance to the simulations in Sec.\,\ref{sec:theory-and-simulation}} \label{append:dfe-simulation-setup}
            We specify how the DFE estimator in Eq.\,\eqref{eq:dfe-finite-shot-estimator} is used in the simulations of Sec.\,\ref{sec:theory-and-simulation}, with its expected number of consumed resource states matched to that of BCFE.
            The target state is
            \begin{align}
                \begin{aligned}
                    \tau_\theta
                    &:=
                    \ket*{m_\theta}\!\bra*{m_\theta}
                    =
                    \frac{1}{2}
                    \left[
                        I+\cos(2\theta)X-\sin(2\theta)Y
                    \right].
                \end{aligned}
                \label{eq:dfe-rotation-target}
            \end{align}
            For $r_P:=\operatorname{tr}(\rho P)$, the fidelity is
            \begin{align}
                F
                =
                \frac{1}{2}
                \left[
                    1+\cos(2\theta)r_X-\sin(2\theta)r_Y
                \right].
                \label{eq:dfe-rotation-fidelity}
            \end{align}
            The relevance probabilities are
            \begin{align}
                \begin{aligned}
                    p(I)&=\frac{1}{2},
                    &\qquad
                    p(X)&=\frac{\cos^2(2\theta)}{2},
                    \\
                    p(Y)&=\frac{\sin^2(2\theta)}{2},
                    &
                    p(Z)&=0.
                \end{aligned}
            \end{align}
            When the identity is drawn, $\widetilde{X}_i=1$ exactly and no copy is consumed.
            For every other draw, the selected Pauli observable is measured on $m_i$ copies, with $m_i$ given by Eq.\,\eqref{eq:dfe-finite-shot-estimator}.

            In the simulations of physical resource states in Figs.\,\ref{fig:physical-simulation} and \ref{fig:physical-simulation-general}, BCFE uses $N=2N_B+3N_P$ copies, with two per Bell measurement and $N_P$ in each Pauli basis.
            To choose the DFE measurement counts at this resource consumption, we fix the failure probability $\delta=0.05$ and set
            \begin{align}
                \eta=\left[\delta\left(L-\frac{1}{2}\right)\right]^{-1/2}.
            \end{align}
            This choice ensures $\lceil1/(\eta^2\delta)\rceil=L$ as required by the prescription above.
            We select the largest integer $L$ for which the expected number of copies consumed by DFE does not exceed $N$, retaining the ceilings in $m_i$ and excluding identity draws from the copy count.
            In the plotted cases, this expectation differs from $N$ by less than one copy.
            All $L$ draws and their prescribed measurements are completed, so the actual number of copies fluctuates.

            For the logical resource states in Figs.\,\ref{fig:bell-benchmark-on-qec} and \ref{fig:logical-realistic-shot-benchmark}, every preparation attempt is counted, including those rejected by syndrome post-selection.
            BCFE uses the same total $N=2N_B+3N_P$ of attempts, with $(N_B,N_P)=(400{,}000,100{,}000)$ and $(40{,}000,10{,}000)$ in the two figures, respectively.
            Each DFE copy undergoes the same preparation and stabilizer projection as in the coherence part of BCFE, and rejected trials are repeated until $m_i$ accepted copies are measured for each nonidentity draw.
            We estimate the acceptance probability $\widehat{p}_{\rm{acc}}$ from the independent preparation counts in the coherence part of BCFE and apply the preceding choice of $L$ to $\lfloor N\widehat{p}_{\rm{acc}}\rfloor$ accepted copies.
            This procedure matches the estimated expected number of preparation attempts, rather than circuit depth or the number of two-qubit gates.
            The matching depends on the acceptance probability estimated from finite data.

            For both physical and logical resource states, the DFE error bars are computed from all $L$ estimates $\widetilde{X}_i$, including identity draws, as
            \begin{align}
                \widehat{\operatorname{SE}}(\widehat{F}_{\rm{DFE}})
                =\Bigg(
                    \frac{1}{L(L-1)}
                    \sum_{i=1}^{L}
                    \left(\widetilde{X}_i-\widehat{F}_{\rm{DFE}}\right)^2
                \Bigg)^{1/2}.
                \label{eq:dfe-simulation-standard-error}
            \end{align}
            The infidelity estimate has the same standard error because $\widehat{\epsilon}_{\rm{DFE}}=1-\widehat{F}_{\rm{DFE}}$.
            Only the plotted point estimate is clipped to $[0,1]$, and the standard error is evaluated before clipping.

    \makeatletter
    \let\pre@bibdata\@empty
    \makeatother
    \bibliographystyle{JHEP}
    \bibliography{main}

\end{document}